\documentclass{article}
\pdfoutput=1
\usepackage[left=1.25in,top=1.25in,right=1.25in,bottom=1.25in,head=1.25in]{geometry}
\usepackage{url}
\usepackage[inline]{showlabels}
\newtheorem{theorem}{Theorem}

\usepackage{amsfonts,amsmath,amssymb,amsthm}
\usepackage{verbatim,float,url,graphicx,psfrag}
\usepackage{natbib}
\usepackage{graphicx}
\usepackage{mathtools} 
\usepackage{mathabx}   
\usepackage{bbm}      
\usepackage{cancel}   
\usepackage{mathtools} 
\usepackage{dsfont}
\usepackage{subcaption}
\usepackage[export]{adjustbox} 
\usepackage{tabularx} 
\usepackage[section]{placeins} 
\usepackage{caption}
\usepackage[ruled,noend]{algorithm2e}
\usepackage{array}
\usepackage{microtype}

\usepackage[usenames,dvipsnames,svgnames,table]{xcolor}
\usepackage[colorlinks=true,
            linkcolor=blue,
            urlcolor=blue,
            citecolor=blue]{hyperref}
\usepackage{booktabs}

\newtheorem{proposition}[theorem]{Proposition}

\usepackage{amsmath}
\makeatletter
\newcommand{\doublewidetilde}[1]{{%
  \mathpalette\double@widetilde{#1}%
}}
\newcommand{\double@widetilde}[2]{%
  \sbox\z@{$\m@th#1\widetilde{#2}$}%
  \ht\z@=.9\ht\z@
  \widetilde{\box\z@}%
}
\makeatother

\usepackage{marginnote}
\usepackage{mdframed} 
\usepackage{comment}
\usepackage{framed}

\usepackage{xr}
\newcommand{\argmax}{\mathop{\mathrm{argmax}}}

\def\one{\mathds{1}}
\def\sign{\text{sign}}
\def\obs{\mathrm{obs}}

\def\sign{\mathrm{sign}}

\def\R{\mathbb{R}}

\def\cN{\mathcal{N}}

\def\cV{\mathcal{V}}

\usepackage[plain]{fancyref}
\makeatletter
\def\mkfancyprefix#1#2{%
\expandafter\def\csname fancyref#1labelprefix\endcsname{#1}%
\begingroup\def\x{\endgroup\frefformat{plain}}%
    \expandafter\x\csname fancyref#1labelprefix\endcsname
    {\MakeLowercase{#2}\fancyrefdefaultspacing##1}%
\begingroup\def\x{\endgroup\Frefformat{plain}}%
    \expandafter\x\csname fancyref#1labelprefix\endcsname
    {#2\fancyrefdefaultspacing##1}%
\begingroup\def\x{\endgroup\frefformat{vario}}
    \expandafter\x\csname fancyref#1labelprefix\endcsname
    {\MakeLowercase{#2}\fancyrefdefaultspacing##1##3}%
\begingroup\def\x{\endgroup\Frefformat{vario}}%
    \expandafter\x\csname fancyref#1labelprefix\endcsname
    {#2\fancyrefdefaultspacing##1##3}%
}
\makeatother

\mkfancyprefix{THM}{Theorem}
\mkfancyprefix{thm}{Theorem}
\mkfancyprefix{LEM}{Lemma}
\mkfancyprefix{lem}{Lemma}
\mkfancyprefix{lemma}{Lemma}
\mkfancyprefix{prop}{Proposition}
\mkfancyprefix{cor}{Corollary}
\mkfancyprefix{def}{Definition}
\mkfancyprefix{cl}{Claim}
\mkfancyprefix{ss}{Subsection}
\mkfancyprefix{alg}{Algorithm}
\mkfancyprefix{ass}{Assumption}
\mkfancyprefix{rem}{Remark}

\begin{document}
\title{Post-Selection Inference for Changepoint Detection Algorithms
  with Application to Copy Number Variation Data}
\author{SANGWON HYUN$^\ast$, KEVIN LIN, MAX G'SELL, RYAN J. TIBSHIRANI\\[4pt]
\textit{
Department of Statistics,
Carnegie Mellon University,
132 Baker Hall, Pittsburgh, PA 15213.
}
\\[2pt]
}

\markboth%
{S. Hyun and others}
{Post-selection inference for changepoint detection}

\maketitle

\footnotetext{To whom correspondence should be addressed: \href{mailto:shyun@cmu.edu}{robohyun66@gmail.com}.}
 


\maketitle

\begin{abstract}
{Changepoint detection methods are used in many areas of science and engineering,
e.g., in the analysis of copy number variation data, to detect abnormalities in
copy numbers along the genome.  Despite the broad array of available tools,
methodology for quantifying our uncertainty in the strength (or presence) of
given changepoints, {\it post-detection}, are lacking. Post-selection inference
offers a framework to fill this gap, but the most straightforward application of
these methods results in low-powered tests and leaves open several important
questions about practical usability. In this work, we carefully tailor
post-selection inference methods towards changepoint detection, focusing as
our main scientific application on copy number variation data. As for
changepoint algorithms, we study binary segmentation, and two of its most popular
variants, wild and circular, and the fused lasso.  We implement some of the
latest developments in post-selection inference theory: we use auxiliary
randomization to improve power, which requires implementations of MCMC
algorithms (importance sampling and hit-and-run sampling) to carry out our
tests. We also provide recommendations for improving practical useability,
detailed simulations, and an example analysis on array comparative genomic
hybridization (CGH) data.}
{CGH analysis; changepoint detection; copy number variation;
hypothesis tests; post-selection inference; segmentation algorithms}
\end{abstract}

\section{Introduction} 
\label{sec:introduction}

Changepoint detection is the problem of identifying changes in data
distribution along a sequence of observations. We study the canonical
changepoint problem, where changes occur only in the mean: let vector
$Y=(Y_1,\ldots,Y_n) \in \R^n$ be a data vector with independent entries
following    
\begin{equation}
\label{eq:data-model}
Y_i \sim \cN(\theta_i, \sigma^2), \quad i=1,\ldots,n, 
\end{equation}
where the unknown mean vector $\theta \in \R^n$ forms a piecewise constant
sequence. That is, for locations $1 \leq b_1 < \cdots < b_t \leq n-1$,
\[
\theta_{b_j+1} = \ldots = \theta_{b_{j+1}}, \quad j=0,\ldots,t.
\]
where for convenience we write $b_0=0$ and $b_{t+1}=n$.  We call $b_1,
\ldots, b_t$ {\it changepoint} locations of $\theta$. Changepoint detection
algorithms typically focus on estimating the number of changepoints
$t$ (which could possibly be 0), as well as the locations $b_1, \ldots, b_t$,
from a single realization $Y$.  Roughly speaking, changepoint methodology (and 
its associated literature) can be divided into two classes of algorithms: {\it
  segmentation} algorithms and {\it penalization} algorithms.   The former class
includes {\it binary segmentation} (BS) \citep{vostrikova1981detecting} and
popular variants like
{\it wild binary segmentation} (WBS) \citep{fryzlewicz2014wild}   
and {\it circular binary segmentation} (CBS) \citep{olshen2004circular}; the
latter class includes the {\it fused lasso} (FL)
\citep{tibshirani2005sparsity} 
(also called {\it total variation denoising} \citep{rudin1992nonlinear} in
signal processing), and the {\it Potts estimator}
\citep{boysen2009consistencies}. These two classes have different strengths;
see, e.g., \citet{lin2016approximate} for more discussion.

Having estimated changepoint locations, a natural follow-up goal would be to
conduct statistical inference on the significance of the changes in mean at
these locations. Despite the large number of segmentation algorithms and
penalization algorithms available for changepoint detection, there has been very
little focus on formally valid inferential tools to use {\it post-detection}.
In this work, we describe a suite of inference tools to use after a changepoint
algorithm has been applied---namely, BS, WBS, CBS, or FL.  We
work in the framework of {\it post-selection inference}, also called {\it
  selective inference}. The specific machinery that we build off was first
introduced in \citet{lee2016exact,tibshirani2016exact}, and further developed in
various works, notably
\citet{fithian2014optimal,fithian2015selective,tian2018selective}, whose
extensions we rely on in particular.  The basic inference procedure we
develop can be outlined as follows.   

\begin{enumerate}
 \item Given data $Y$, apply a changepoint algorithm to detect some fixed number  
   of changepoints $k$. Denote the sorted estimated changepoint locations by
   \begin{equation} 
     \label{eq:estimated-changepoints}
     1 \leq \hat{c}_1 < \cdots < \hat{c}_k \leq n-1,
   \end{equation}
   and their respective changepoint directions (whether the estimated change in
   mean was positive or negative) by \smash{$\hat{d}_1, \ldots,\hat{d}_k \in 
   \{-1,1\}$}.  For notational convenience, we set \smash{$\hat{c}_0= 0$} and 
 \smash{$\hat{c}_{k+1} = n$}.  The specifics of the changepoint algorithms that
 we consider are given in \Fref{sec:algorithms}.

 \item Form contrast vectors $v_1,\ldots, v_k \in \R^n$, defined so that for
   arbitrary $y \in \R^n$,
   \begin{equation} 
     \label{eq:segment-contrast}
     v_j^T y = \hat{d}_j \bigg( \frac{1}{\hat{c}_{j+1}-\hat{c}_j}
     \Big(\sum_{i=\hat{c}_j+1}^{\hat{c}_{j+1}} y_i \Big)- 
     \frac{1}{\hat{c}_j-\hat{c}_{j-1}+1}
     \Big(\sum_{i=\hat{c}_{j-1}+1}^{\hat{c}_j} y_i \Big)\bigg), 
   \end{equation}
   the difference between the sample means of segments to right and left of
   \smash{$\hat{c}_j$}, for $j=1,\ldots,k$.
   
 \item For each $j = 1,\ldots, k$, we test the hypothesis $H_0: v_j^T \theta=0$ by
   rejecting for large values of a statistic $T(Y, v_j)$, which is computed
   based on knowledge of the changepoint algorithm that produced
   \eqref{eq:estimated-changepoints} in Step 1, and the desired contrast vector
   \eqref{eq:segment-contrast}  formed in Step 2.  Each statistic
      yields an exact p-value under the null (assuming Gaussian errors
   \eqref{eq:data-model}). The details are given in Sections
   \ref{sec:post-selection} and \ref{sec:inference-ours}. 

\item Optionally, we can use Bonferroni correction and multiply our p-values by
  $k$, to account for multiplicity.  
\end{enumerate}

It is worth mentioning that several variants of this basic
procedure are possible.  For example, the number of changepoints $k$ in Step 1
need not be seen as fixed and may be itself estimated from data;
the set of estimated changepoints \eqref{eq:estimated-changepoints} may be
pruned after Step 1 to eliminate changepoints that lie too close to others,
and alternative contrast vectors to \eqref{eq:segment-contrast} in
Step 2 may be used to measure more localized mean changes; these are all
briefly described in \Fref{sec:practicalities}.
Though not covered in our paper, the p-values from our tests can  
be inverted to form confidence intervals for population contrasts
$v_j^T \theta$ for $j = 1,\ldots, k$  
\citep{lee2016exact,tibshirani2016exact}. 

At a more comprehensive level, our contributions in this work are to implement 
theoretically valid inference tools and practical guidance for each combination of
the following choices that a typical user might face in a changepoint analysis:
the algorithm (BS, WBS, CBS, or FL),
number of estimated changepoints $k$ (fixed or data-driven),
the null hypothesis model (saturated or selected model, to be explained in 
\Fref{sec:post-selection}),
what type of conditioning (plain or marginalized, to be explained in \Fref{sec:randomization}),
and the error variance $\sigma^2$ (known or unknown).
In \Fref{sec:practicalities}, we summarize the tradeoffs underlying each of
these choices.   

Finally, as the primary application of our inference tools, we study 
comparative genomic hybridization (CGH) data, making particular suggestions
geared towards this problem throughout the paper.  We begin with a motivating
CGH data example in the next subsection, and return to it at the end of the
paper. 

\subsection{Motivating example: array CGH data analysis}

We examine array CGH data from the 14th chromosome of cell line GM01750, one of
the 15 datasets from \citet{Snijders2001}; more background can be found in
\citet{lai2005comparative} and references therein. Array CGH data are $\log_2$
ratios of dye intensities of diseased to healthy subjects' measurements, mixed
across many samples. Normal regions of the gene are thought to have an
underlying mean $\log_2$ ratio of zero, and aberrations are regions of upward or
downward departures from zero because the gene in that region has been mutated
-- duplicated or deleted. The presence and locations of aberrations are well
studied in the biomedical literature to be associated with the presence of a
wide range of genetically driven diseases -- as many types of cancer, Alzheimer,
and autism \citep{fanciulli2007fcgr3b, sebat2007strong, international2008rare,
  stefansson2008large, walters2010new, bochukova2010large}. Accurate changepoint
analysis of array CGH data is thus useful in studying association with diseases,
and for medical diagnosis. 

The data is plotted in the left panel of \Fref{fig:intro}. Two locations
\smash{$\hat{c}_1 < \hat{c}_2$}, marked A and B respectively, were detected by
running 2-step WBS.  Ground truth in this data set can be defined via an
external process called called karyotyping; this is done by \citet{Snijders2001}
who finds only one true changepoint at location A. (To be precise, they do not
report exact locations of abnormalities, but find a single start-to-middle
deviation from zero level.)

Without access to any post-selection inference tools, we might treat locations A
and B as fixed, and simply run t-tests for equality of means of neighboring data
segments, to the left and right of each location. This is precisely testing the
null hypothesis $H_0 : v_j^T\theta=0$, $j = 1,2$, where the contrast vectors are
as defined in \eqref{eq:segment-contrast}.  P-values from the t-tests are
reported in the first row of the table in \Fref{fig:intro}: we see that location
A has a p-value of $< 10^{-5}$, but location B also has a small
p-value of $5 \times 10^{-4}$, which is troublesome. 
The problem is that
location B was specifically selected by WBS because (loosely put) the sample
means to left and right of B are well separated, thus a t-test a location B
is bound to be optimistic.

Using the tools we describe shortly, we test $H_0 :v_j^T \theta=0$,  $j = 1,2$ 
in two ways: using a {\it saturated model} and a {\it selected model} on the mean
vector $\theta$.  The satured model assumes nothing about $\theta$, while the
selected model assumes $\theta$ is constant between the intervals
formed by $A$ and $B$.  
Both tests yield a p-value $< 10^{-5}$ at location A, but only a moderately small p-value at
location B.  If we were to use the Bonferroni correction at a nominal
significance level $\alpha=0.05$, then in neither case would we reject
the null at location B.  

\begin{figure}[ht!]
  \begin{minipage}{0.5\linewidth}
    \centering
    \includegraphics[width=\linewidth]{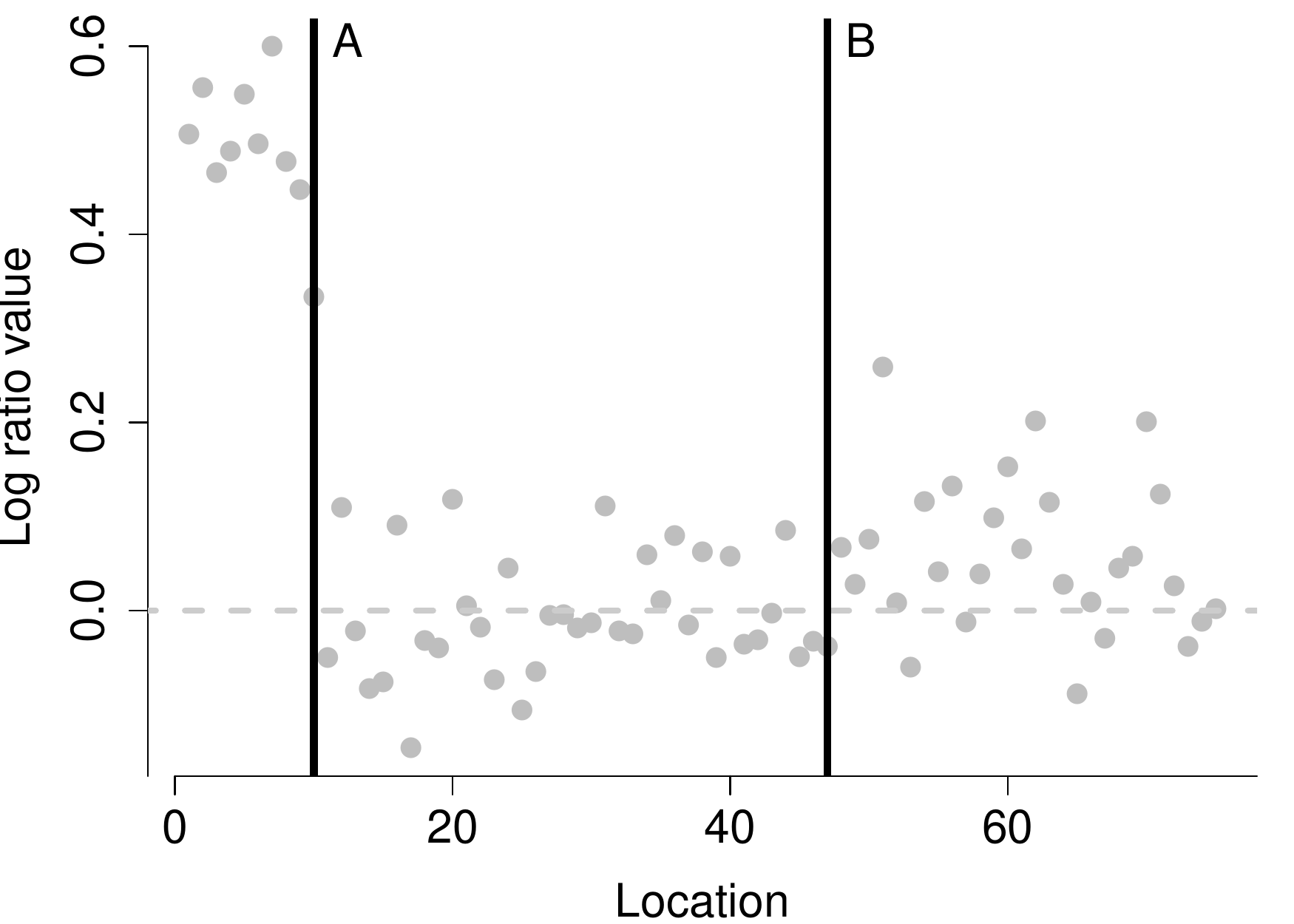} 
  \end{minipage} 
  \begin{minipage}{0.5\linewidth}
    \centering\small
    \begin{tabular}{ccc}
      \toprule
      Location & A & B \\ 
      \midrule
      Karyotype & True & False \\ 
      \midrule
      Classical t-test & 0 & $5 \times 10^{-4}$  \\ 
      Saturated model test & 0 & 0.050 \\ 
      Selected model test & 0 & 0.027 \\
      \bottomrule
    \end{tabular}
  \end{minipage}
  \caption{\it\small Left: array CGH data from the 14th chromosome of
    fibroblast cell line GM01750, from \citet{Snijders2001}.  The x-axis denotes
    the relative index of the genome position, and the y-axis denotes the log
    ratio in fluorescence intensities of the test and reference samples.  The
    dotted horizontal line denotes a log ratio of 0 for reference.  The bold
    vertical lines denote the locations A and B from running WBS for 2 steps.  Right: the
    p-values using classical (naive) t-tests, saturated model tests, and
    selected model tests, at each location A and B.  The ground truth is also 
    given, as determined by karyotyping. The saturated model test used an estimated
    noise level $\sigma^2$ from the entire 23-chromosome data set. The selected
    model test was performed in the unknown $\sigma^2$ setting.}
  \label{fig:intro}
\end{figure}
\subsection{Related work}

In addition to the references on general post-selection inference methodology
given previously, we highlight the recent work of \citet{hyun2018exact}, who
study post-selection inference for the generalized lasso, a special case of
which is the fused lasso. These authors already characterize the polyhedral form
of fused lasso selection events, and study inference using contrasts as in
\eqref{eq:segment-contrast}. While writing the current paper, we became aware of
the independent contributions of \citet{umezu2017selective}, who study
multi-dimensional changepoint sequences, but focus problems in which the mean
$\theta$ has only one changepoint.  Aside from these papers, there is little
focus on valid inference methods to apply post-detection in changepoint
analysis.  On the other hand, there is a huge literature on changepoint
estimation, and inference for {\it fixed} hypotheses in changepoint problems; we
refer to \citet{jandhyala2013,aueHorvath2013,horvath2014}, which collectively
summarize a good deal of the literature.

\section{Preliminaries}

\subsection{Review: changepoint algorithms}
\label{sec:algorithms}

Below we describe the changepoint algorithms that we will study in this paper.  
For the first three segmentation algorithms, we will focus on formulations
that run the algorithm for a given number of steps $k$; these algorithms
are typically described in the literature as being run until internally
calculated statistics do not exceed a given threshold level $\tau$.  The reason
that we choose the former formulation is twofold: first, we feel it is easier
for a user to specify a priori a reasonable number of steps $k$, versus a
threshold level $\tau$; second, we can use the method in \citet{hyun2018exact}
to adaptively choose the number of steps $k$ and still perform valid
inferences.  In what follows, we use the
notation \smash{$y_{a:b}=(y_a, y_{a+1}, \ldots, y_b)$} and 
\smash{$\bar{y}_{a:b} = (b-a+1)^{-1} \sum_{i=a}^b y_i$}
for a vector $y$.

\paragraph{\textbf{Binary segmentation (BS).}} Given a data vector $y \in \R^n$, the
$k$-step BS algorithm \citep{vostrikova1981detecting} sequentially splits the
data based on the cumulative sum 
(CUSUM) statistics, defined below.  At a step $\ell = 1,\ldots,k$, let 
\smash{$\hat{b}_{1:(\ell-1)}$} be the changepoints estimated so far, and let
$I_j$, $j=1,\ldots,\ell-1$ be the partition of $\{1,\ldots,n\}$ induced by
\smash{$\hat{b}_{1:(\ell-1)}$}.  Intervals of length 1 are discarded.
Let $s_j$ and $e_j$ be the start and end indices of $I_j$.
The next changepoint \smash{$\hat{b}_\ell$} and maximizing interval
\smash{$\hat{j}_\ell$} are chosen to maximize the absolute CUSUM statistic:  
\begin{gather}
 \big\{\hat j_{\ell}, \hat b_{\ell}\big\} =
 \argmax_{\substack{j \in \{1, \ldots, \ell-1\} \\ 
     b \in \{s_j, \ldots, e_j-1 \}}} 
 \big|g^T_{(s_j, b, e_j)} y\big|,
\quad \text{where} \nonumber\\
\label{eq:bs-g-fun}
g_{(s,b,e)}^Ty = \sqrt{\frac{1}{\frac{1}{|e-b|}+\frac{1}{|b+1-s| 
}}}\big(\bar y_{(b+1):e} - \bar y_{s:b}\big).
\end{gather}
Additionally, the direction \smash{$\hat{d}_\ell$} of the new changepoint is
calculated by the sign of the maximizing absolute CUSUM statistic,
\smash{$\hat{d}_{\ell} = \sign(g_{(s_j, b_{\ell}, e_j)}^Ty)$} for 
$j = \hat j_{\ell+1}$.  

\paragraph{\textbf{Wild binary segmentation (WBS).}} The $k$-step WBS algorithm
\citep{fryzlewicz2014wild} is a modification of BS that calculates
CUSUM statistics over randomly drawn segments of the data. 
Denote by
$w = \{w_1, \ldots, w_B\} = \{(s_1, \ldots, e_1), \ldots, (s_B, \ldots, 
e_B)\}$
  a set of $B$ uniformly randomly drawn
intervals with $1 \leq s_i < e_i \leq n$, $i=1,\ldots,B$.  At a step
$\ell=1,\ldots,k$, let $J_\ell$ to be the index set of the intervals
in $w$ which do not intersect with the changepoints 
\smash{$\hat{b}_{1:(\ell-1))}$} estimated so far. The next
changepoint \smash{$\hat{b}_{\ell}$} and the maximizing interval  
\smash{$\hat{j}_{\ell}$} are obtained by: 
\begin{equation*}
  \big\{\hat j_{\ell}, \hat b_{\ell}\big\} = 
  \argmax_{\substack{j \in J_\ell \\
      b \in \{s_j, \ldots, e_j-1\}}}
  \big| g^T_{(s_j, b, e_j)} y \big|,
\end{equation*}
where $g_{(s,b,e)}^T y$ is as defined in \eqref{eq:bs-g-fun}.
Similar to BS, the direction of the changepoint \smash{$\hat{d}_\ell$} is
defined by the sign of the maximizing absolute CUSUM statistic. 

\paragraph{\textbf{Circular binary segmentation (CBS).}} The $k$-step CBS algorithm
\citep{olshen2004circular} specializes in detecting {\it pairs} of changepoints
that have alternating directions.  At a step $\ell=1,\ldots,k$, let \smash{$\hat 
a_{1:(\ell-1)}$}, \smash{$\hat b_{1:(\ell-1)}$} be the changepoints estimated
so far (with the pair $a_j$, $b_j$ estimated at step $j$), and let $I_j$,
$j=1,\ldots,2(\ell-1)+1$ be the associated partition of $\{1,\ldots,n\}$. 
Intervals of length 2 are discarded.  
Let $s_j$ and $e_j$ denote the start and end index of $I_j$. The next
changepoint pair \smash{$\hat a_{\ell}$} and \smash{$\hat b_{\ell}$}, and the 
maximizing interval \smash{$\hat j_{\ell}$}, are found by:
\begin{gather}
  \label{eq:cbs-opt-prob}
  \big\{\hat j_{\ell}, \hat a_{\ell}, \hat b_{\ell}\big\} =  
  \argmax_{\substack{ j \in \{1,\ldots,2(\ell-1)+1)\} \\ 
      a < b \in \{s_j, \ldots, e_j-1\}   }}  
  \big| g^T_{(s_j, a, b, e_j)}y
  \big| \quad \text{where} \\
  \label{eq:cbs-g-fun}
  g_{(s,a,b,e)}^Ty = 
  \sqrt{\frac{1}{\frac{1}{|b-a|}+\frac{1}{|e-s-b+a|}}} 
  \Big(\bar y_{(a+1):b} - \bar y_{\{s:a\}\cup\{(b+1):e\}}\Big).
\end{gather}
As before, the new changepoint direction \smash{$\hat d_{\ell}$} is defined 
based on the sign of the (modified) CUSUM statistic,
\smash{$\hat d_{\ell} = \sign(g^T_{(s_j, a_{\ell+1}, b_{\ell+1}, e_j)}y)$} for 
\smash{$j = \hat j_{\ell+1}(y)$.}


\paragraph{\textbf{Fused lasso.}} The fused lasso (FL) estimator
\citep{rudin1992nonlinear,tibshirani2005sparsity} is defined by solving the
convex optimization problem:
\begin{equation}
 \label{eq:fl}
 \min_{\theta \in \R^n} \; \sum_{i=1}^n (y_i - \theta_i)^2 + \lambda 
 \sum_{i=1}^{n-1} |\theta_i - \theta_{i+1}|,
\end{equation}
for a tuning parameter $\lambda \geq 0$.  The fused lasso can be seen as a
$k$-step algorithm by sweeping the tuning parameter from $\lambda=\infty$ down 
to $\lambda=0$. Then, at given values of $\lambda$ (called knots), 
the FL estimator introduces an additional changepoint in the solution in
\eqref{eq:fl} \citep{hoefling2010path}. 

\subsection{Review: post-selection inference} 
\label{sec:post-selection}

We briefly review post-selection inference as developed in 
\citet{lee2016exact,tibshirani2016exact,fithian2014optimal}.  For a more
thorough and general treatment, we refer to these papers or to
\citet{hyun2018exact}.  Our description here will be cast towards changepoint 
problems. For clarity, we notationally distinguish between a random vector $Y$
distributed as in \eqref{eq:data-model}, and $y_\obs$, a
single data vector we observe for changepoint analysis.  When a changepoint
algorithm---such as BS, WBS, CBS, or FL---is applied to the data $y_\obs$, it  
selects a particular changepoint model $M(y_\obs)$. The specific forms
of such models are described in \Fref{sec:polyhedra}; for now,
loosely, we may think of $M(y_\obs)$ as the estimated changepoint locations and 
directions made by the algorithm on the data at hand.
Post-selection inference revolves around the selective distribution, i.e., the
law of
\begin{equation}
\label{eq:selective-distribution}
v^T Y \; | \;  \Big(M(Y) = M(y_\obs),\; q(Y) = q(y_\obs)\Big),
\end{equation}
under the null hypothesis $H_0: v^T \theta = 0$,
for any $v$ that is a measurable function of $M(y_\obs)$. Here $q(Y)$ is a
vector of sufficient statistic of nuisance parameters that need to be
conditioned on in order to tractably compute inferences based on
\eqref{eq:selective-distribution}.  The explicit form of $q(Y)$ differs based on
the assumptions imposed on $\theta$ under the null model. Broadly, there are two
classes of null models we may study: saturated and selected
models \citep{fithian2014optimal}. 
Computationally, in either null models, it is important for the selection event $\{ y: M(y) = M(y_\obs)\}$ be polyhedral. This is described in detail in  Section \ref{sec:polyhedra},
where we show that this holds for 
BS, WBS, CBS, and FL.

\paragraph{\textbf{Saturated model.}} The {\it saturated model} assumes that $Y$ is
distributed as in \eqref{eq:data-model} with known error variance
$\sigma^2$, and assumes nothing about the mean vector $\theta$. 
We set $q(Y) = \Pi_v^\perp Y$, the projection of $Y$ onto the hyperplane
orthogonal to $v$. The selective distribution becomes the law of    
\begin{equation}
\label{eq:selective-distribution-saturated}
v^TY \; | \; \Big(M(Y) = M(y_\obs),\; \Pi_v^\perp Y = \Pi_v^\perp y_\obs\Big).
\end{equation}

\paragraph{\textbf{Selected model.}} The {\it selected model} again assumes that $Y$
follows \eqref{eq:data-model}, but additionally assumes that the mean vector
$\theta$ is piecewise constant with changepoints at the sorted estimated
locations \smash{$\hat c_{1:k}=\hat c_{1:k}(y_\obs)$} (assuming we have run our 
changepoint algorithm for $k$ steps).  That is, we assume 
\[
\theta_{\hat c_j+1} = \ldots = \theta_{\hat c_{j+1}}, \quad j\in\{0,\ldots,k\}.
\]
where for convenience we use \smash{$\hat c_0=0$} and \smash{$\hat c_{k+1}=n$}.
Under this assumption, the law of $Y$ becomes a $(k+1)$-parameter Gaussian distribution.
Additionally, with the 
contrast vector $v_j$ defined as in \eqref{eq:segment-contrast}, for any fixed 
$j=1,\ldots,k$, the quantity $v_j^T \theta$ of interest is simply the difference
between two of the parameters in this distribution.  Assuming $\sigma^2$ is known, the
sufficient statistics $q(Y)$ for the nuisance parameters in the Gaussian family
are simply sample averages of the appropriate data segments, and the selective 
distribution becomes the law of 
\begin{equation}
\label{eq:selective-distribution-selected-known-sigma}
\big( \bar{Y}_{(\hat c_j + 1) : \hat c_{j+1}} - 
\bar{Y}_{(\hat c_{j-1}+1) : \hat c_j} \big) \; \big| \; 
\Big(M(Y) = M(y_\obs),\; 
\bar{Y}_{(\hat c_\ell + 1) : \hat c_{\ell+1}} = 
\big(\bar{y}_\obs\big)_{(\hat c_\ell + 1) : \hat c_{\ell+1}}, \; \ell \neq j\Big).
\end{equation}
Part of the strength of the selected model is that we can properly treat
$\sigma^2$ as unknown; in this case, we must only additionally condition on the
Euclidean norm of $y_\obs$ to cover this nuisance parameter, and the selective
distribution becomes the law of
\begin{multline}
\label{eq:selective-distribution-selected-unknown-sigma}
\big( \bar{Y}_{(\hat c_j + 1) : \hat c_{j+1}} - 
\bar{Y}_{(\hat c_{j-1}+1) : \hat c_j} \big) \; \big| \; 
\Big(M(Y) = M(y_\obs),\; 
\bar{Y}_{(\hat c_\ell + 1) : \hat c_{\ell+1}} = 
\big(\bar{y}_\obs\big)_{(\hat c_\ell + 1) : \hat c_{\ell+1}}, \; \ell \neq j, \\
\|Y\|_2 = \|y_\obs\|_2\Big).
\end{multline}


  
\section{Inference for changepoint algorithms}
\label{sec:inference-ours}

We describe our contributions that enable post-selection inference for
changepoint analyses, beginning with the form of model selection events for
common changepoint algorithms.  We then describe computational details for
saturated and selected model tests, and auxiliary randomization.

\subsection{Polyhedral selection events} 
\label{sec:polyhedra}

We show that, for each of the BS, WBS, and CBS algorithms, there is a
parametrization for their models such that event $\{y : M(y) = M(y_\obs)\}$ is a
polyhedron---in fact a convex cone---of the form $\{ y : \Gamma y \geq 0 \}$,
for a matrix $\Gamma \in \R^{m \times n}$ that depends on $M(y_\obs)$ (and we
interpret the inequality $\Gamma y \geq 0$ componentwise). 
Throughout the description of the polyhedra for each algorithm, we display
the number of rows in $\Gamma$ since it loosely denotes how ``complex'' each
model selection event is. 
 The same
was already shown for FL in \citet{hyun2018exact}, and we omit details, but
briefly comment on it below.  
Overall, the
$\Gamma$ matrices for FL and BS are linear in $n$, while it is quadratic in $n$
for CBS, and $O(Bkp)$ for WBS using intervals of length $p$.
This number can grow faster than linear in $n$ if $B\ge n$, which
is recommended in practice \citep{fryzlewicz2014wild}.

\paragraph{\textbf{Selection event for BS.}} We define the model for the $k$-step BS estimator as  
\[
M^{\mathrm{BS}}_{1:k}(y_\obs)  = \big\{\hat b_{1:k}(y_\obs), \; \hat d_{1:k}(y_\obs)\big\},
\]
where \smash{$\hat b_{1:k}(y_\obs)$} and \smash{$\hat d_{1:k}(y_\obs)$} are the
changepoint locations and directions when the algorithm is run on $y_\obs$, as
described in Section \ref{sec:algorithms}.   

\begin{proposition}
\label{prop:bs-polyhedral-event}
\textit{Given any fixed $k \geq 1$ and $b_{1:k},d_{1:k}$, we can explicitly construct
$\Gamma$ where
\[
\big\{y : M_{1:k}^{\mathrm{BS}}(y) = \{b_{1:k}, d_{1:k} \} \big\} = 
\{ y : \Gamma y \geq 0 \},
\]
and where $\Gamma$ has \smash{$2 \sum_{\ell=1}^k (n - \ell - 1)$}   
rows.}
\end{proposition}

\begin{proof}  When $k=1$, $2(n-2)$ linear inequalities
  characterize the single changepoint model $\{b_1, d_1\}$: 
\[
d_1 \cdot g^T_{(1, b_1, n)} y \geq g^T_{(1,b,n)} y, 
\quad \text{and} \quad 
d_1 \cdot g^T_{(1, b_1, n)} y \geq -g^T_{(1,b,n)} y, 
\quad 
b \in \{1,\ldots, n-1\} \backslash \{ b_1\}. 
\]
Now by induction, assume we have constructed a polyhedral representation of the
selection event up through step $k-1$.  All that remains is to characterize the
$k$th estimated changepoint and direction $\{b_k, d_k\}$ by inequalities that
are linear in $y$.  This can be done with $2 (n-k-1)$
inequalities. To see this, assume without a loss of generality that the maximizing
interval is $j_k=k$; then $\{b_k,d_k\}$ must satisfy the $2 (|I_k|-2)$
inequalities  
\[
 d_k \cdot g^T_{(s_k, b_k, e_k)} y \geq g^T_{(s_k, b, e_k)} y  
 \quad \text{and} \quad
 d_k \cdot g^T_{(s_k, b_k, e_k)} y \geq -g^T_{(s_k, b, e_k)} y,
 \quad 
 b \in \{s_k, \ldots, e_k - 1\} \backslash \{b_k\}.
\]
For each interval $I_\ell$, $\ell=1,\ldots,k-1$, we also have $2 (|I_\ell|-1)$ 
inequalities 
\[
 d_k \cdot g^T_{(s_k, b_k, e_k)}
 y \geq g^T_{(s_\ell, b, e_\ell)} y  \quad\text{and}\quad
  d_k \cdot g^T_{(s_k, b_k, e_k)}
 y \geq -g^T_{(s_\ell, b, e_\ell)} y,
 \quad
 b \in \{s_\ell, \ldots, e_\ell- 1\}.
\]
The last two displays together completely determine $\{b_k,d_k\}$, and as 
\smash{$\sum_{\ell=1}^k |I_\ell| = n$}, we get our desired total of 
$2 (n-k-1)$  inequalities.  
\end{proof}

\paragraph{\textbf{Selection event for WBS.}}  We define the model of the $k$-step WBS estimator as  
\[
M^{\mathrm{WBS}}_{1:k}(y_\obs, w) = \big\{\hat b_{1:k}(y_\obs), \; \hat d_{1:k}(y_\obs),\; \hat
j_{1:k}(y_\obs) \big\},
\]
where $w$ is the set of $B$ intervals that the
algorithm uses, \smash{$\hat b_{1:k}(y_\obs)$} and \smash{$\hat d_{1:k}(y_\obs)$} are the
changepoint locations and directions, and \smash{$\hat j_{1:k}(y_\obs)$} are the
maximizing intervals.    

\begin{proposition}
\label{prop:wbs-polyhedral-event}
\textit{Given any fixed $k \geq 1$, and $\{w,b_{1:k},d_{1:k},j_{1:k}\}$, we can explicitly
construct $\Gamma$ where
\[
\big\{y : M_{1:k}^{\mathrm{WBS}}(y, w) = \{b_{1:k},d_{1:k},j_{1:k}\} \big\} = 
\big\{ y : \Gamma y \geq 0 \big\}.
\]
The number of rows in $\Gamma$ will vary depending on the configuration of $w$
and $b_{1:k}$, but if each of the $B$ intervals in $w$ has length $p$, it will be at most 
$2\sum_{\ell=1}^{k}((B-\ell)\cdot(p-1) + (p-2))$.}
\end{proposition}

The proof of \Fref{prop:wbs-polyhedral-event} is only slightly more complicated 
than that of \Fref{prop:bs-polyhedral-event}, and is deferred until Appendix 
\ref{sec:proofs}. Note that unlike BS, the maximizing intervals $\hat j_{1:k}$ are part of WBS's model.

\paragraph{\textbf{Selection event for CBS.}} Finally, we define the model for the $k$-step CBS 
estimator
as
\[
M^{\mathrm{CBS}}_{1:k}(y_\obs) = \big\{\hat a_{1:k}(y_\obs), \; \hat b_{1:k}(y_\obs),\; \hat
d_{1:k}(y_\obs) \big\},
\]
where now \smash{$\hat a_{1:k}(y_\obs)$} and \smash{$\hat b_{1:k}(y_\obs)$} are the pairs
of estimated changepoint locations, and \smash{$\hat d_{1:k}(y_\obs)$} are the
changepoint directions, as described in Section \ref{sec:algorithms}.    

\begin{proposition}
\label{prop:cbs-polyhedral-event}
\textit{Given any fixed $k \geq 1$ and $\{a_{1:k},b_{1:k},d_{1:k}\}$, we can 
explicitly construct $\Gamma$ where
\[
\big\{y : M_{1:k}^{\mathrm{CBS}}(y, w) = \{a_{1:k},b_{1:k},d_{1:k}\} \big\} = 
\big\{ y : \Gamma y \geq 0 \big\}.
\]
Let $I_j^{(\ell)}$ denote the 
$j$th interval formed and $j_\ell$ be the selected interval defined in \eqref{eq:cbs-opt-prob}
for an intermediate step $\ell \in \{1,\ldots,k\}$, 
and let $C(x,2) = {x \choose 2}$.
Then $\Gamma$ has a number of rows equal to
\[
2 \sum_{\ell = 1}^{k} \Big[C(|I^{(\ell)}_{j_k}| - 1, 2) -1  + \sum_{j' \neq j_k} C(|I^{(\ell)}_{j'}| -1 , 2)\Big].
\]}
\end{proposition}

The proof of \Fref{prop:cbs-polyhedral-event} is only slightly more complicated 
than that of \Fref{prop:bs-polyhedral-event}, and is deferred until Appendix 
\ref{sec:proofs}. 

\paragraph{\textbf{Selection events for FL, and a brief comparison.}}
The model for the $k$-step FL estimator is:
\[
  M^{\mathrm{FL}}_{1:k}(y_\obs) = \big\{ \hat b_{1:k}(y_\obs), \; \hat d_{1:k}(y_\obs) , \; \hat
  R_{1:k}(y_\obs)\big\},
\]
where \smash{$\hat b_{1:k}(y)$} and \smash{$\hat d_{1:k}(y)$} are changepoint
locations and directions, and
$\smash{\hat R_{\ell}(y)\in\R^{n-\ell}, \ell=1,\ldots,k}$ whose elements
represent signs of a certain statistic $h_i(y)$ calculated at location $i$ in
competition for maximization with $\hat b_\ell$ at step $\ell$. These statistics
$h_i(y)$ are weighted mean differences at location $i$ and are analogous to
CUSUM statistics in BS.  \cite{hyun2018exact} make this representation more
explicit, proving that for any fixed $k \geq 1$ and $b_{1:k},d_{1:k}, R_{1:k}$,
we can explicitly construct $\Gamma$ such that
\[
  \big\{y : M_{1:k}^{\mathrm{FL}}(y) = \{b_{1:k}, d_{1:k}, R_{1:k} \} \big\} =
  \{ y : \Gamma y \geq 0 \},
\]
where $\Gamma$ has the same number of rows as a $k$-step BS event. 

\subsection{Computation of p-values}\label{sec:computation}

Given a precise description of the polyhedral selection event
$\{y : M(y) = M(y_\obs)\}$, we can describe the methods to compute the p-value,
i.e. the tail probability of the selective distributions described in
\Fref{sec:post-selection}.  Without loss of generality, all of our descriptions
will be specialized to testing the null hypothesis of $H_0: v^T\theta = 0$
against the one-sided alternative $H_1: v^T \theta>0$.  For
saturated model tests, this exact calculation has been developed in previous
work and we review it as it is relevant to our contributions on increasing its
power.  For selected model tests, an approximation was described in previous
work, but we develop a new hit-and-run sampler that has not been implemented
before.

\paragraph{\textbf{Saturated model tests: exact formulae.}}

As shown in \citet{lee2016exact} and \citet{tibshirani2016exact}, 
the saturated selective distribution \eqref{eq:selective-distribution-saturated} has
a particularly computationally convenient distribution when $Y$ is Gaussian and 
the model selection event $\{y : M(y)=M(y_\obs)\}$ is a polyhedral set in
$y$. In this case, the law of \eqref{eq:selective-distribution-saturated} is  
a {\it truncated Gaussian} (TG), whose truncation limits depend only on
\smash{$\Pi_v^\perp y_\obs$}, and can be computed explicitly.  
Its tail probability can be computed in closed form (without Monte Carlo sampling). 
That is, the probability that $v^TY \geq v^Ty_\obs$ under the law of \eqref{eq:selective-distribution-saturated} is exactly equal to
 \begin{equation}\label{eq:tg_statistic}
 (\Phi(\cV_{\text{up}}/\tau) - \Phi(v^Ty_\obs/\tau))/(\Phi(\cV_{\text{up}}/\tau) - \Phi(\cV_{\text{lo}}/\tau))
 \end{equation}
where $\Phi(\cdot)$ represents the standard Gaussian CDF, 
$\tau = \sigma^2 \|v\|^2_2$, $\rho = \Gamma v /\|v\|^2_2$ and 
\begin{equation} \label{eq:vlo_vup}
\cV_{\text{lo}} = v^Ty_\obs - \min_{j: \rho_j > 0}
\big(\Gamma y_\obs \big)_j/\rho_j,\quad \text{and}\quad
\cV_{\text{up}} =  v^Ty_\obs - \max_{j: \rho_j < 0}
\big(\Gamma y_\obs \big)_j/\rho_j.
\end{equation}
This above equation is commonly referred as the TG statistic.
Since this statistic is a pivot, it is the p-value used for the saturated model test.

\paragraph{\textbf{Selected model tests: hit-and-run sampling.}}

To compute the p-value for selected model tests,
\cite{fithian2015selective} proposed a hit-and-run
strategy for sampling from the distribution for the known $\sigma^2$ setting,
\eqref{eq:selective-distribution-selected-known-sigma}.
This was implemented by the authors, and we briefly review the details in Appendix  \ref{app:known_sigma}.
For the unknown $\sigma^2$ setting, \cite{fithian2014optimal} 
suggested an importance sampling strategy for sampling the distribution
\eqref{eq:selective-distribution-selected-unknown-sigma}.
However, we find that an intuitive hit-and-run strategy can be adapted to the unknown $\sigma^2$
setting and implement this as a new algorithm. 

Given a changepoint $j = 1,\ldots, k$, observe that 
we can design a segment test contrast $v$ where
sampling from
\eqref{eq:selective-distribution-selected-unknown-sigma} is equivalent to 
sampling uniformly from the set
 \begin{equation} \label{eq:selected_set}
 \Big\{ v^TY: M(Y) = M(y_\obs),\;  \|Y\|_2 = \|y_{\text{obs}}\|_2, \; \bar{Y}_{(\hat{c}_{\ell}+1):\hat{c}_{\ell+1}} = \bar{y}_{\obs, (\hat{c}_{\ell}+1):\hat{c}_{\ell+1}} \ell \neq j
 \Big\}.
\end{equation}
Note that the above set no longer depends on $\theta$ or $\sigma^2$.
This is because we conditioned all the relevant sufficient statistics under the
selected model.
Our hit-and-run sampler then sequentially draws samples $v^TY$ from the above set.
For notational convenience, observe that the last $k$ constraints in \eqref{eq:selected_set} 
can be rewritten as $AY = Ay_{\text{(obs)}}$ for some matrix $A \in \R^{k \times n}$.
Our new hit-and-run algorithm is then shown in \Fref{alg:hitandrun}.  

\begin{algorithm}[t]
Choose a number $M$ of iterations.\\
Set $y^{(0)} = y_\obs$.\\
 \For{$m \in \{1,\ldots,M\}$}{
Uniformly sample two unit vectors $s$ and $t$ in the
nullspace of $A$.\\
Compute the set $\mathcal{I} \subseteq [-\pi/2, \pi/2]$
that intersects the set 
\begin{equation*} 
\Big\{ y\; : \; y = y^{(m-1)} + r(\omega)\sin(\omega) \cdot s + r(\omega) \cos(\omega)\cdot t \quad \text{for any }\omega \in [-\pi/2,\pi/2]\Big\},
\end{equation*}
for the radius function $r(\omega) = -2(y^{(m-1)})^T(\sin(\omega)\cdot s + \cos(\omega)\cdot t)$,
with the polyhedral set implied by the selected model $M(y_\obs)$ based on \Fref{sec:polyhedra}.\\
Uniformly sample $\omega^{(m)}$ from $\mathcal{I}$ and form the next sample
\[
y^{(m)} = y^{(m-1)} + r(\omega^{(m)})\sin(\omega^{(m)})\cdot s + r(\omega) \cos(\omega^{(m)})\cdot t.
\]
}
Return the approximate for the tail probability of \eqref{eq:selective-distribution-selected-unknown-sigma},
$
\sum_{m=1}^{M} \one[v^Ty^{(m)} \geq v^Ty_\obs]/M.
$
 \caption{MCMC hit-and-run algorithm for selected model test with unknown $\sigma^2$} \label{alg:hitandrun}
\end{algorithm}

\subsection{Randomization and marginalization}
\label{sec:randomization}

We apply the ideas of randomization in \cite{randomized-selinf} that improve the
power of selective inference to changepoint algorithms and devise explicit
samplers.  We investigate two specific forms of randomization: randomization
over additive noise and randomization over random intervals.  We specialize the
following descriptions to saturated models.  We note that similar randomization
of selected model inferences is also possible but is doubly computationally
burdensome.

\paragraph{\textbf{Marginalization over additive noise.}}
\cite{randomized-selinf} shows that performing inference based on the selected
model $ M(y_\obs + w_\obs)$ where $w_\obs$ is additive noise
and then marginalizing over $W$ leads to improved power.
Here, $w_\obs$ is a realization of a random
component $W$ sampled from $\cN(0, \sigma_{\text{add}}^2I_n)$, where
$\sigma_{\text{add}}^2 > 0$ is set by the user.
\cite{fithian2014optimal} provides a mathematical basis for pursuing such randomization,
stating that less conditioning results in an increase in Fisher
information. 
For additive noise, the above model
selection event is:
$$\{ y : \Gamma(y + w_\obs) \geq 0\} = \{y : \Gamma y \geq -\Gamma w_\obs\}.$$
This means the new polyhedron formed by the model selection event based on
perturbed data $y_\obs + w_\obs$ is slightly shifted.

Porting the ideas of \cite{randomized-selinf} to our setting, to
test the one-sided null hypothesis $H_0: v^T\theta = 0$, we want to compute
the following tail probability of the marginalized selective distribution,
\begin{equation}\label{eq:cond-dist-addnoise-marg}
  T(y_\obs, v) = \mathbb{P}\bigg(v^TY \geq v^Ty_\obs ~\big|~ \Big( M(Y + W) = M(y_\obs + W), \; \Pi_v^\perp Y = \Pi_v^\perp y_{\text{obs}}\Big)\bigg).
 \end{equation}
It is hard to directly compute this.
 However, the formulas in \eqref{eq:tg_statistic} and \eqref{eq:vlo_vup} 
 give us
 exact formulas to compute the
 non-marginalized tail-probabilities,
 \begin{equation*}
   T(y_\obs, v, w_\obs) = \mathbb{P}\bigg(v^TY \geq v^Ty_\obs~\big|~ \Big(M(Y + W) = M(y_\obs + W), \;\Pi_v^\perp Y = \Pi_v^\perp y_{\text{obs}},\; W=w_{\text{obs}}\Big)\bigg).
 \end{equation*}
The following proposition shows that we can compute $T(y_\obs, v) $
 by reweighting instances of $T(y_\obs, v, w_\obs) $ via importance sampling.
 Here, let  $E_1 = \one[M(Y + W) = M(y_\obs + W)]$ and
 $E_2 = \one[\Pi_v^\perp Y = \Pi_v^\perp y_{\text{obs}}]$.

\begin{proposition}\label{prop:additive_noise}
\textit{Let $\Omega$ denote the support of the random component $W$.
If the distribution of $W$ is independent of the random 
event
$E_2$, 
\eqref{eq:cond-dist-addnoise-marg} can be exactly
computed as
\begin{equation}\label{eq:additive_noise}
T(y_\obs, v) = \int_{\Omega} T(y_\obs, v, w_\obs) \cdot a(w_\obs) \; dP_W(w_\obs) 
= \frac{\int_{\Omega} \Phi\big(\cV_{\text{up}}/\tau\big) - \Phi\big(v^Ty_\obs/\tau\big) \;dP_W(w_\obs)}{\int_{\Omega} \Phi\big(\cV_{\text{up}}/\tau\big) - \Phi\big(\cV_{\text{lo}}/\tau\big) \; dP_W(w_\obs)}. 
\end{equation}
where the weighting factor is
$a(w_\obs)= \mathbb{P}(W = w_\obs | E_1, E_2)/\mathbb{P}(W = w_\obs)$.}
\end{proposition}
The first equality in \eqref{eq:additive_noise} demonstrates the reweighting of $T(y_\obs, v, w_\obs) $,
but the second equality
gives a sampling strategy where we approximate the integrals.
\Fref{alg:additive-importance-sampler} describes this, where for one realization $w_\obs$,
we let  $k(w_\obs)$ and $g(w_\obs)$ denote the integrand of the 
last term's numerator and denominator in \eqref{eq:additive_noise} respectively.
  
  \paragraph{\textbf{Marginalization over WBS intervals.}}
In contrast to the above setting where $W$ represents Gaussian noise,
in wild binary segmentation described in \Fref{sec:algorithms}, $W$ represents
the set of $B$ randomly drawn intervals. 
Observe that  \Fref{prop:additive_noise} still applies to this setting, where 
$M(y_\obs + w_\obs)$ is now replaced with $M(y_\obs, w_\obs)$,
as described in \Fref{sec:polyhedra}.
However, one
additional complication is that the maximizing intervals $\hat j_{1:k}$ in the model $M(y_\obs, w_\obs)$
are embedded in the construction of the matrix $\Gamma$ representing the
polyhedra. 
%
This prevents a naive resampling of all $B$ intervals.

We describe how to overcome this complication.
Let
 \smash{$\{W_{\hat j_1}, \ldots, W_{\hat j_k}\}$}  be the maximizing intervals. 
We resample all other intervals,
$W_\ell$ for \smash{$\ell \in \{1, \ldots, B\} \backslash \{\hat j_1, \ldots, \hat j_k\}$}.
Specifically, for each of such intervals $W_\ell = (s_\ell, \ldots, e_\ell)$,
$s_\ell$ and $e_\ell$ are sampled uniformly between $1$ to $n$ where $s_\ell < e_\ell$.
After all $B-k$ intervals are resampled, a check is performed to ensure that 
  $\{W_{\hat j_1}, \ldots, W_{\hat j_k}\}$ are still the maximizing intervals when WBS is
  applied again to $y_\obs$.
  The full algorithm is in \Fref{alg:wbs-importance-sampler}.


\hspace{-.5cm}\begin{minipage}[t]{6cm}
  \vspace{0pt}  
  \begin{algorithm}[H]
    \caption{Marginalizing over additive noise} \label{alg:additive-importance-sampler}
   Choose a number $T$ of trials.\\
   \For{$t \in \{1, \ldots, T\}$}{
  Sample the additive noise $w_j$ from $\cN(0, \sigma^2_{\text{add}}I_n)$.\\
  Compute $k(w_t)$ and $g(w_t)$.\\
 }
 Return the approximate for the tail probability \eqref{eq:additive_noise},
 \[
\frac{\sum_{t=1}^{T}k(w_t)}{\sum_{t=1}^{T}g(w_t)}.
 \]
  \end{algorithm}
\end{minipage}
\hspace{0.5cm}
\begin{minipage}[t]{9.5cm}
  \vspace{0pt}
  \begin{algorithm}[H]
    \caption{Marginalizing over random intervals}\label{alg:wbs-importance-sampler}
  Choose a number $T$ of trials.\\
      \For{$t \in \{1, \ldots, T\}$}{
  Sample  the non-maximizing intervals
  $w_\ell = (s_\ell, \ldots, e_\ell)$ for $\ell \in \{1, \ldots, B\} \backslash \{\hat j_{1:k}\}$ 
  where $s_\ell, e_\ell$ are uniformly drawn from 1 to $n$ and $s_\ell < e_\ell$.\\
  Check to see that $\{\hat j_{1:k}\}$ are still the indices of the maximizing intervals. If not,
  return to the previous step.\\
 Compute $k(w_t)$ and $g(w_t)$.\\
 }
 Return the approximate for the tail probability \eqref{eq:additive_noise},
 \[
\frac{\sum_{t=1}^{T}k(w_t)}{\sum_{t=1}^{T}g(w_t)}.
 \]
      \end{algorithm}
\end{minipage}

\section{Practicalities and extensions}
\label{sec:practicalities}



The above sections formalize the mechanisms to perform selective inference with
respect to the basic procedure highlighted in \Fref{sec:introduction}.  We now
briefly summarize the all the combination of choices that the user faces based
on the methods developed in the above sections and their practical impact.

\subsection{Practical considerations}
There are some practical choices that the user needs to make when implementing
the procedure. Here, we outline a few, each related with a key element of the
broader inference procedure.


\begin{itemize}
\item \textbf{Algorithm} (BS, WBS, CBS and FL): It is useful for the user to be
  able to compare algorithms.  CBS is specialized for pairs of changepoints, and
  WBS specializes in localized changepoint detection compared to BS. FL and BS
  have similar mechansims which sequentially admit changepoints by maximizing a
  statistic. However, BS has a simpler mechanism and a less complex selection
  event, potentially giving higher post-selection conditional power.

\item \textbf{Conditioning} (Plain or marginalized): Marginalizing over a source
  of randomness yields tests with higher power than plain inference, but at two
  costs: increased computational burden due to MCMC sampling being required, and
  worsened detection ability when using additive noise marginalization. Also,
  the marginalized p-values are subject to the sampling randomness, and the
  number of trials $T$ needed to reduce the p-values' intrinsic
  variability 
  scales with $\sigma^2_{\text{add}}$.

\item \textbf{Number of estimated changepoints $k$} (Fixed or data-driven):
  As currently described in \Fref{sec:algorithms}, the changepoint algorithms
  discussed in our paper require the user to pre-specify the number of estimated
  changepoints $k$.
  However, we can adopt local stopping rules from \cite{hyun2018exact} 
  to adaptively choose $k$.
  This variation 
  increases the complexity of the polyhedra compared to those in 
   \Fref{sec:polyhedra}, leading to
  lower statistical power than its fixed-$k$
  counterpart. This is shown in Appendix \ref{app:ic}.

\item \textbf{Assumed null model} (Saturated or selected): As mentioned in
  \Fref{sec:post-selection}, selected model tests are valid under a stricter set
  of assumptions but often yield higher power. 
  Computationally, saturated model tests are often simpler to perform
  than selected model tests due to the closed form expression of the tail probability.

\item \textbf{Error variance $\sigma^2$} (Known or unknown): Saturated model
  tests require $\sigma^2$ to be known. In practice, we need to estimate it
  in-sample from a reasonable changepoint mean fitted to the same data, or
  estimated out-of-sample on left-out data. Selected model tests have the
  advantage of not requiring knowledge of $\sigma^2$.

\end{itemize}

\subsection{Extensions}

As mentioned in \cite{hyun2018exact}, there are many practically-motivated
extensions to the baseline procedure mentioned in \Fref{sec:introduction}
to either improve power or
interpretability. We highlight these below. All
of these extensions will still give proper Type-I error control under the
appropriate null hypotheses.


\begin{itemize}
\itemsep-.5em 
\item \textbf{Designing linear contrasts}: The user can make many
types of contrast vectors $v$ to fit their analysis, in addition to the
  segment test contrasts \eqref{eq:segment-contrast}, as long as it
  measurable with respect to $M(y_\obs)$. 
  One example is the spike test from
  \citep{hyun2018exact} of single location mean changes. For CNV analysis, it
  could be useful to test regions between an adjacent pair of changepoints away
  from the immediately surrounding regions. 
  Also, a step-sign plot (a plot
  that shows the locations and direction of the changepoints, but not their magnitude)
  can help the user design contrasts \citep{hyun2018exact}.

\item \textbf{Post-processing the estimated changepoints}: Multiple detected
  changepoints too close to one another can hurt the power of segment tests. 
  Post-processing the estimated changepoints based on decluttering
  \citep{hyun2018exact} or filtering \citep{lin2017sharp} so the new set of
  changepoints are well-separated can lead to contrasts that yield higher
  power. We show empirical evidence of this improving power of the fused lasso,
  in Appendix \ref{app:unique-detection}.




\item \textbf{Pre-cutting}: We can also modify all the algorithms in
  \Fref{sec:algorithms} to start with an initial existing set of 
  changepoints. This is useful in CGH analyses, when it is not meaningful to
  consider segments that start in one chromosome and end in another.  By pooling
  information in this manner from separate chromosomal regions, the pre-cut
  analysis is an improvement over conducting separate analyses in individual
  chromosomes.  

\end{itemize}

\section{Simulations} \label{sec:simulation}

\subsection{Gaussian simulations}

In this section, we show simulation examples to demonstrate properties of the
segmentation post-selection inference tools presented in the current paper.  The
mean $\theta$ consists of two alternating-direction changepoints of size
$\delta$ in the middle as in \eqref{eq:middle-mutation}, chosen to be a
realistic example of mutation phenomena as observed in array CGH datasets
\citep{Snijders2001}. We vary the signal size $\delta \in (0,4)$, while
generating Gaussian data from a fixed noise level $\sigma^2=1$.

This is the \textit{duplication}
mutation scenario. The sample size $n=200$ is chosen
to be in the scale of the chromosomal data. An example of this synthetic dataset
can be seen in Figure \ref{fig:power-comparison-data}.


\begin{equation}\label{eq:middle-mutation}
  \hspace{-20mm}\textbf{Middle mutation:}\hspace{5mm}
  y_i \sim \cN(\theta_i, 1), \;\; 
  \theta_i = \begin{cases}
    \delta & \text{ if } 101\le i \le 140\\
    0 & \text{ if otherwise } \\
  \end{cases}
\end{equation}

\begin{figure}[ht]
  \centering
  \includegraphics[width=.5\linewidth]{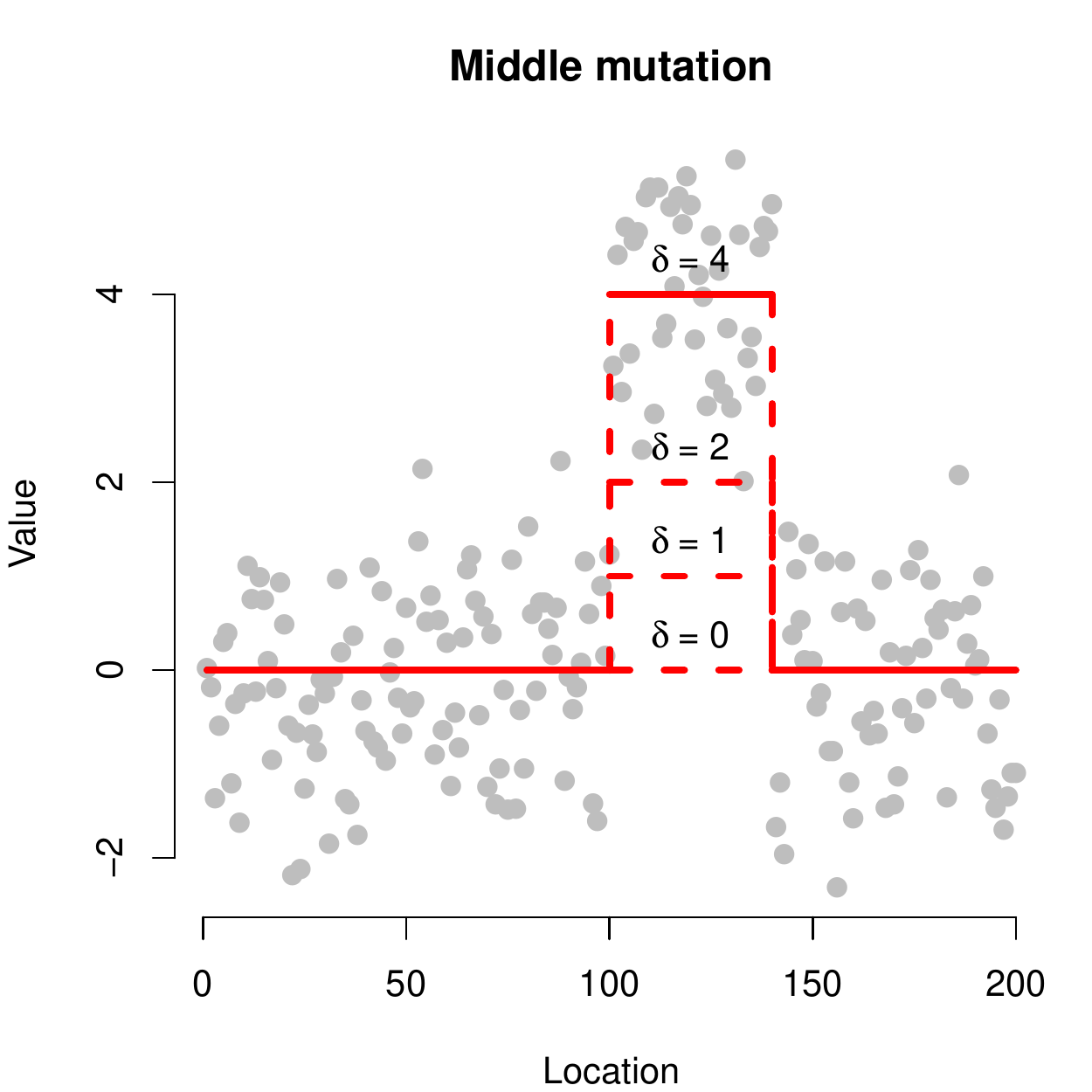}
  \caption{\it\small Example of simulated Gaussian data for middle mutation as
    defined in \eqref{eq:middle-mutation} with $\delta=4$, with data length
    $n=200$ and noise level $\sigma=1$. The possible mean vectors $\theta$ for
    $\delta = 0, 1, 2$ are also shown.}
  \label{fig:power-comparison-data}
\end{figure}

\paragraph{\textbf{Methodology.}}
In the following simulations, we consider the following four estimators (BS, WBS,
CBS and FL) each run for two steps. From each, we perform both saturated
and selected model tests. For the latter, we only include the results of BS
and FL for simplicity, for both settings of known and unknown noise parameter
$\sigma^2$. 
We use the basis procedure outlined in \Fref{sec:introduction} with
a significance level of $\alpha=0.05$. We verify the Type-I error control of our
methods next.  Throughout the entire simulation suite to come, the standard
deviation in each of the power curves and detection probabilities is less than
0.02. For each method, for each signal-to-noise size $\delta$, we run more than
250 trials.





\paragraph{\textbf{Type-I error control verification.}}
We examine all our statistical inferences under the global null
where $\theta = 0$ to
demonstrate their validity -- uniformity of null p-values, or type I error control. 
Specifically, any simulations from the no-signal regime $\delta=0$ from the 
middle mutation \eqref{eq:middle-mutation} can be
used. When there is no signal, the null scenario $v^T\theta=0$ is always true so we
expect all p-value to be uniformly distributed between 0 and 1.  
We verify this
expected behavior in \Fref{fig:null-dist}.
We notice that the methods that require MCMC (marginalized saturated and
selected model tests) requires more trials to converge towards
the uniform distribution compared to their counterparts that have exact calculations.

\begin{figure}
  \centering
    \includegraphics[width=.33\linewidth]{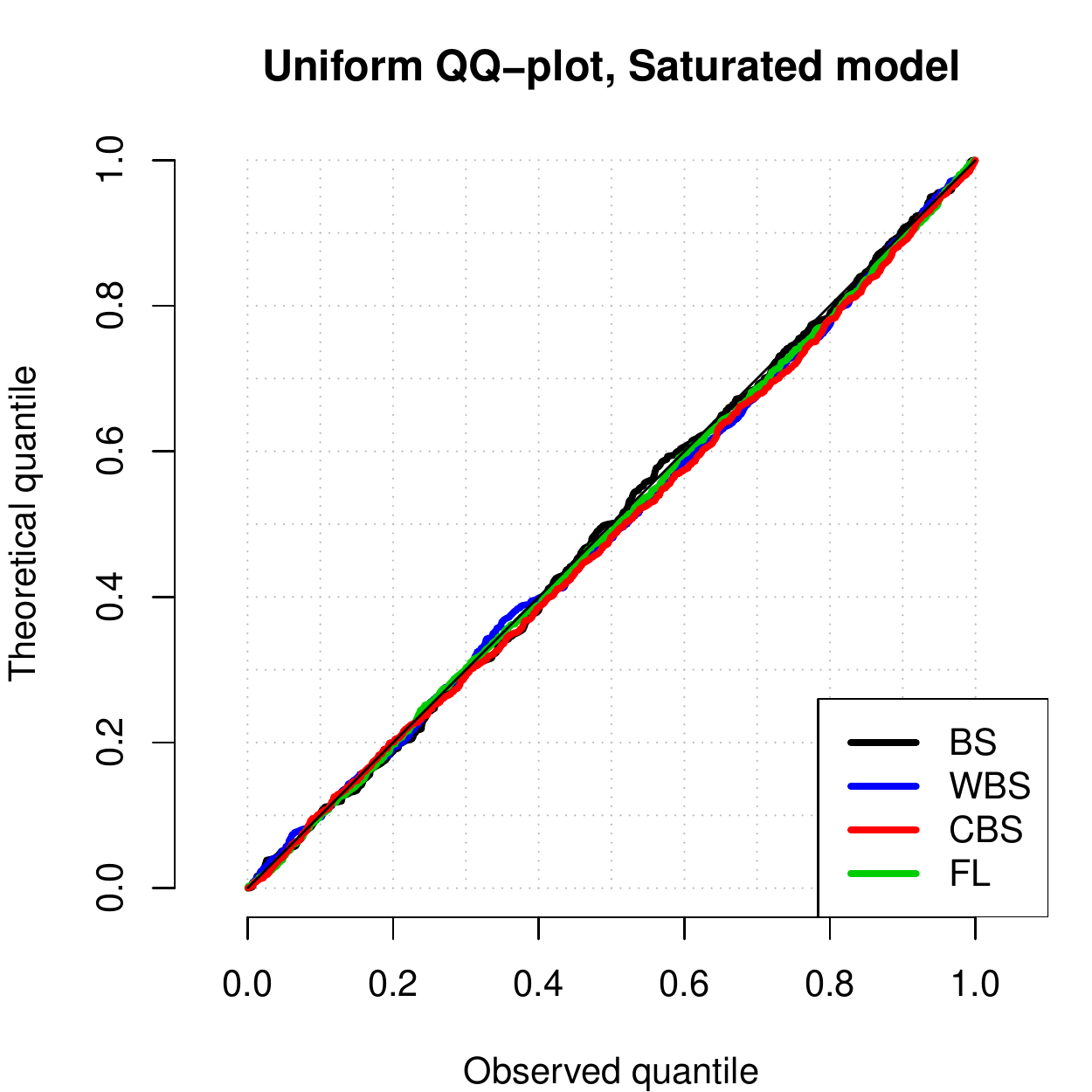}\hspace{-3mm} 
    \includegraphics[width=.33\linewidth]{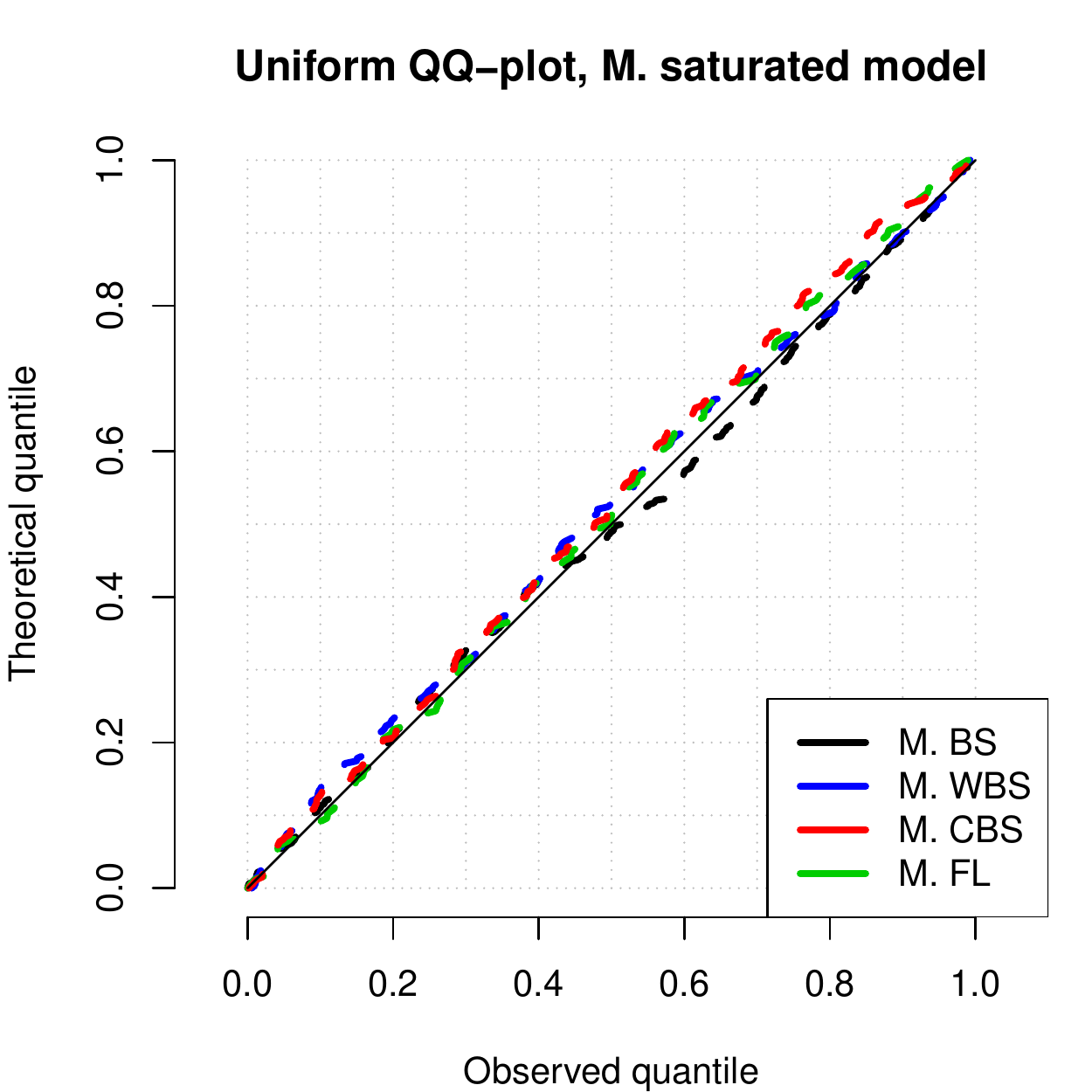}\hspace{-3mm}
    \includegraphics[width=.33\linewidth]{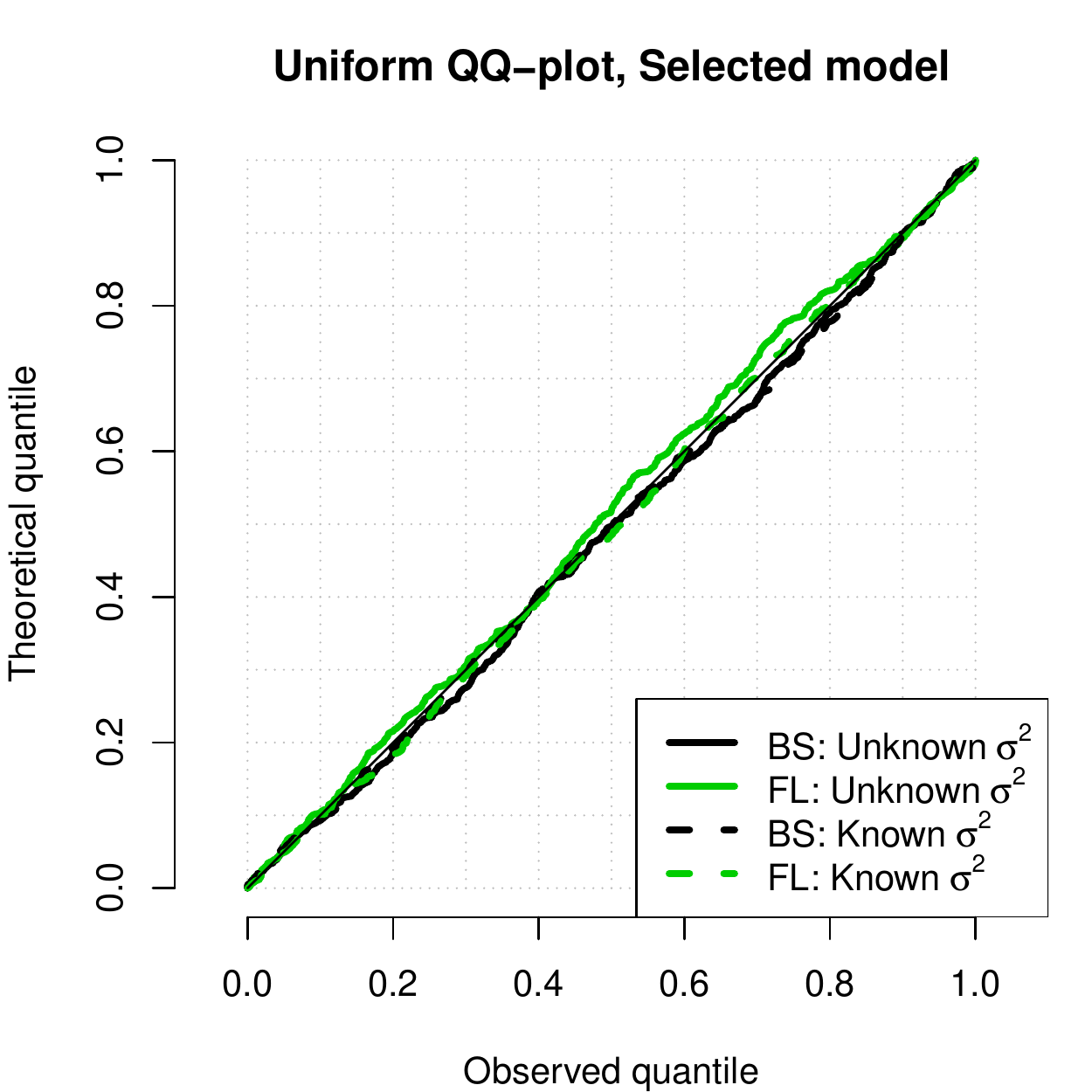}\hspace{-3mm}
    \caption{\it\small 
    All plots showing the p-values of various statistical inferences under the global null,
    with colors of lines given according to \Fref{fig:power-comparison} and \ref{fig:power-comparison-selected}.
    (Left): Saturated model tests, specifically BS (black), WBS (blue), CBS (red) and
    FL (green). (Middle): Marginalized variants of the left plot.
    (Right): Selected model tests, specifically BS (black) and FL (green), either
    with unknown $\sigma^2$ (solid) or known $\sigma^2$ (dashed).
      }
      \label{fig:null-dist}
\end{figure}

\paragraph{\textbf{Calculating power.}}
Since the tests are performed only when a changepoint is selected, it is necessary to
separate the detection ability of the estimator from power of the test. To
that end, we define the following quantities, 
\begin{align}
  \text{Conditional power} &=  \frac{\# \;\text{correctly detected \& rejected}}{ \#\; \text{correctly detected}}\label{eq:powdef2}\\
  \text{Detection probability} &=  \frac{\# \;\text{correctly detected}}{ \# \;\text{tests conducted}}\label{eq:powdef1}\\
  \text{Unconditional power} &=  \frac{\# \;\text{correctly detected \& rejected}}{\# \;\text{tests
                               conducted}} = \text{Detection} \times \text{Conditional power} \label{eq:powdef3}
\end{align}
The overall power of an inference tool can only be assessed by examining the
conditional and unconditional power together.
We consider a detection to be
correct if it is within $\pm 2$ of the true changepoint locations.

\paragraph{\textbf{Power comparison across signal sizes $\delta$.}}
For saturated model tests, we perform additive-noise inferences
using Gaussian $\cN(0,\sigma_{\text{add}}^2)$ with
$\sigma_{\text{add}}=0.2$ for BS, FL, and CBS. For WBS, we employ the randomization scheme as
described in \Fref{sec:randomization} with $B=n$. With the metrics in
\eqref{eq:powdef1}-\eqref{eq:powdef3}, we examine the performance of the four
methods. 
The solid lines in \Fref{fig:power-comparison} show the ``plain'' method where
model selection based on $M(y_\obs)$.
The dotted lines show the marginalized counterparts where
the model selection is $M(y_\obs, W)$, margnialized over $W$.

WBS and CBS have higher conditional and unconditional power than
BS. This is as expected since the former two are more adept for localized
change-points of alternating directions. FL noticeably under-performs in power
compared to segmentation methods. This is partially caused by FL's detection
behavior, and can be explained by examining alternative measures of detection
and improved with post-processing. This investigation is deferred to Appendix
\ref{app:unique-detection}. The marginalized versions of each algorithm have
noticeably improved power, but almost unnoticeably worse detection than their
non-randomized, plain versions (middle panel of
\Fref{fig:power-comparison}) 
. Combined, in
terms of unconditional power, marginalized inferences clearly dominate their
plain counterparts.
 
Selected model inference simulations are shown in
\Fref{fig:power-comparison-selected}. 
Surprisingly, there is an almost inconceivable
drop in power from unknown $\sigma^2$ to known $\sigma^2$.
Compared to the saturated model tests in \Fref{fig:power-comparison},
there is smaller power gap between FL and BS. Also, selected model tests appear
to have higher power than saturated model tests. In general however, it is hard
to compare the power of saturated and selected models due to the clear
difference in model assumptions.  



\paragraph{\textbf{Comparison with sample-splitting.}}

Sample splitting is another valid inference technique. After splitting the dataset 
in half based on even and odd indices, we run a changepoint algorithm on one dataset and
conduct classical one-sided t-test on the other. This is the most comparable
test, as it 
does not assume
$\sigma^2$ is known and conducts a one-sided test of the null
$H_0:v^T\theta=0$. Instead of $\pm 2$ slack used for calculating detection in
selective inference detection (dotted and dashed lines), $\pm 1$ was used for
sample splitting inference (solid line). The loss in detection accuracy in the
middle panel of \Fref{fig:samplesplit} shows the downside of halving data size
for detection. Unconditional power for marginalized saturated model tests and
selected model tests are noticeably higher than the other two. 


\begin{figure}
  \centering
    \includegraphics[width=.33\linewidth]{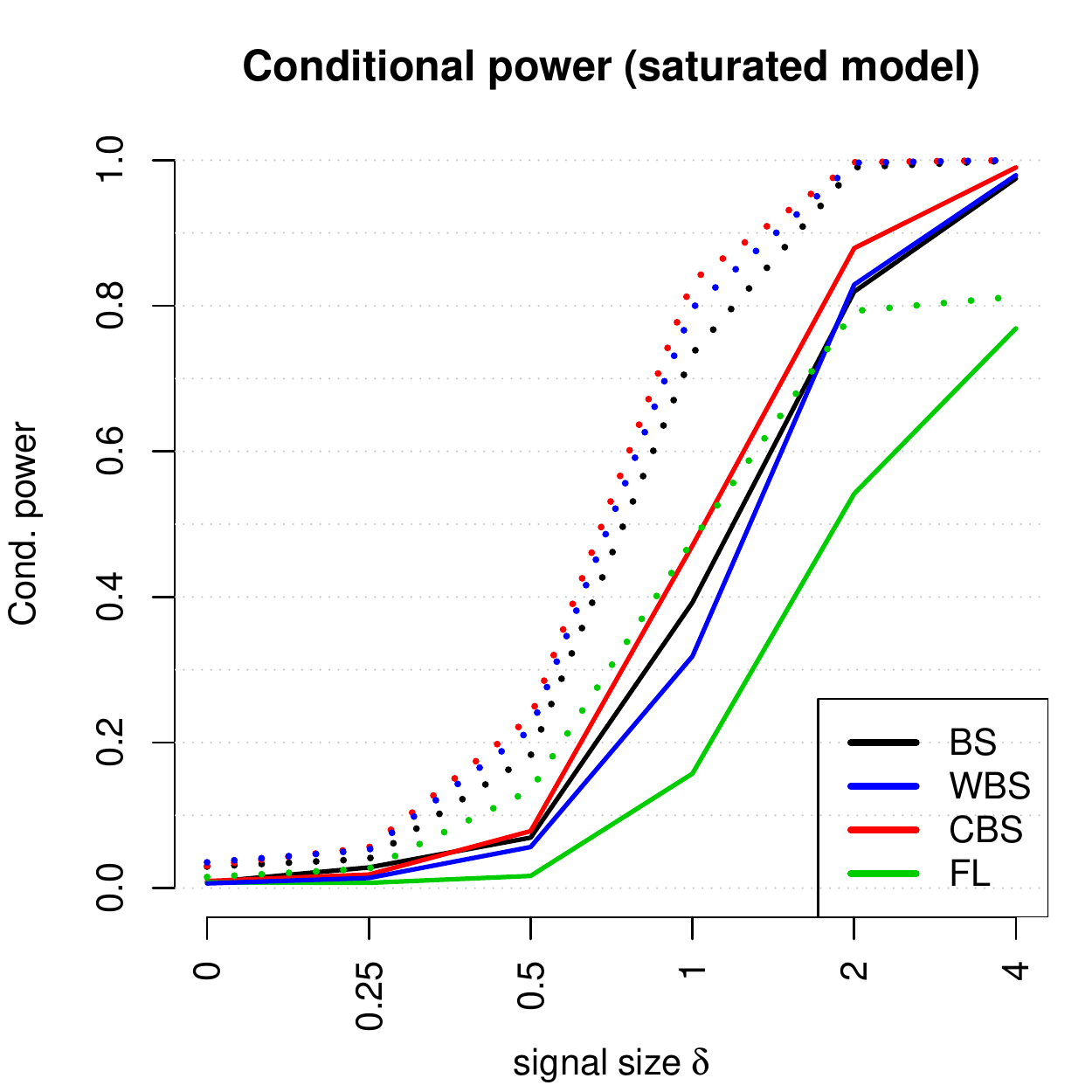}\hspace{-3mm}
    \includegraphics[width=.33\linewidth]{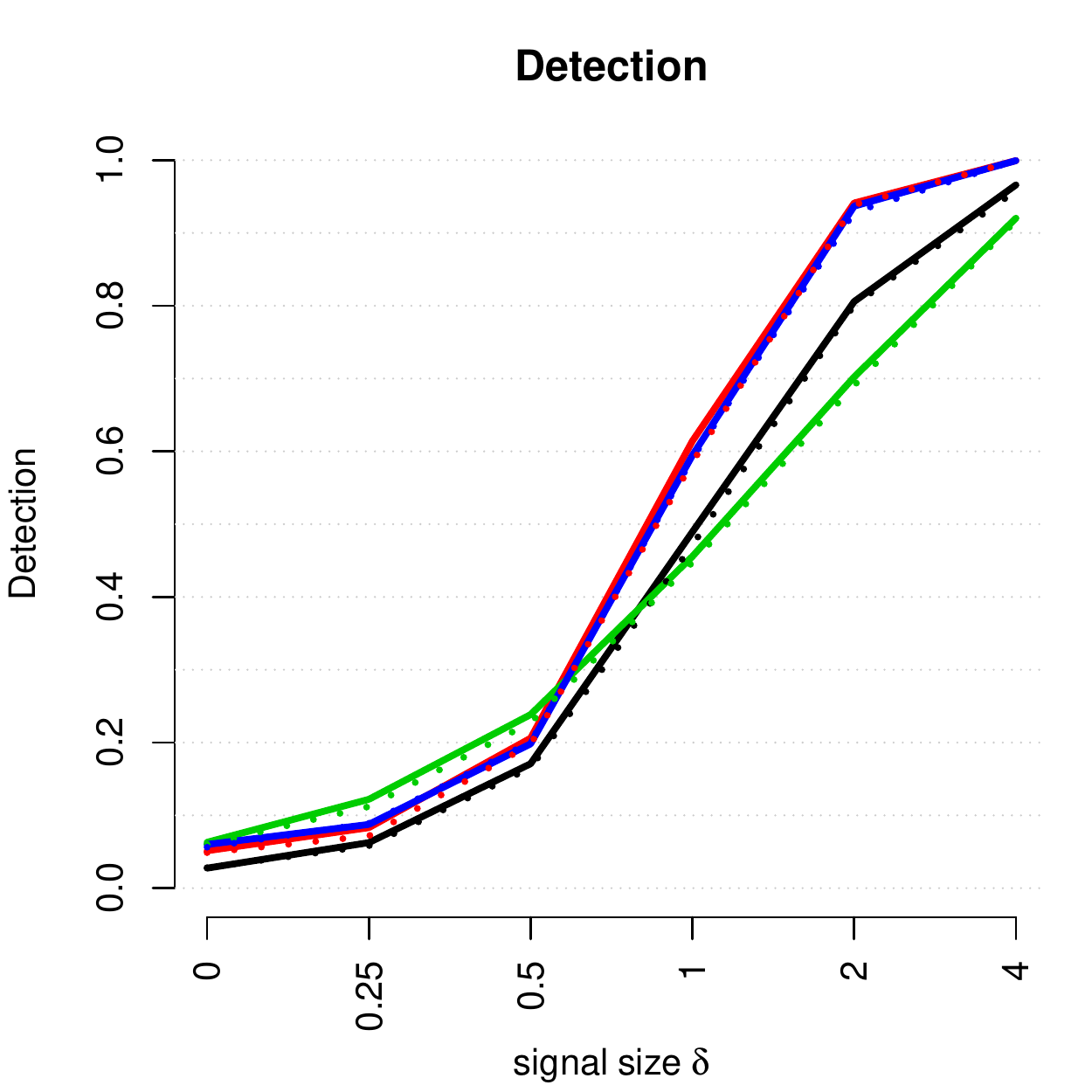}\hspace{-3mm}
    \includegraphics[width=.33\linewidth]{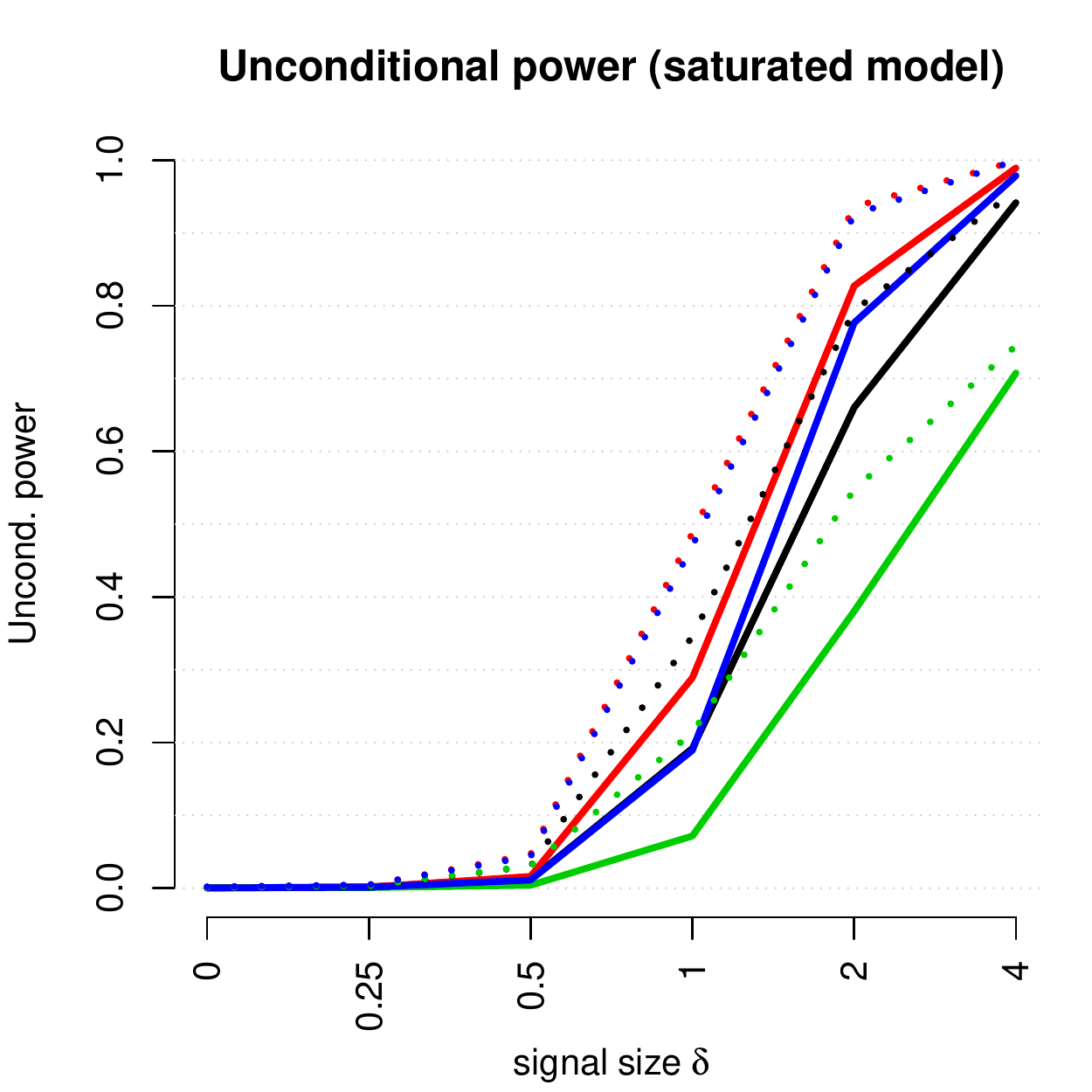}\hspace{-3mm}
    \caption{\it\small Data was simulated from two settings over signal size
      $\delta \in (0,4)$ with $n=200$ data points.  Several two-step algorithms
      (WBS, SBS, CBS, FL) were applied, and post-selection segment test
      inference was conducted on the resulting two detected changepoints from
      each method. The dotted lines are the marginalized versions of each
      test. }
\label{fig:power-comparison}

  \centering
    \includegraphics[width=.33\linewidth]{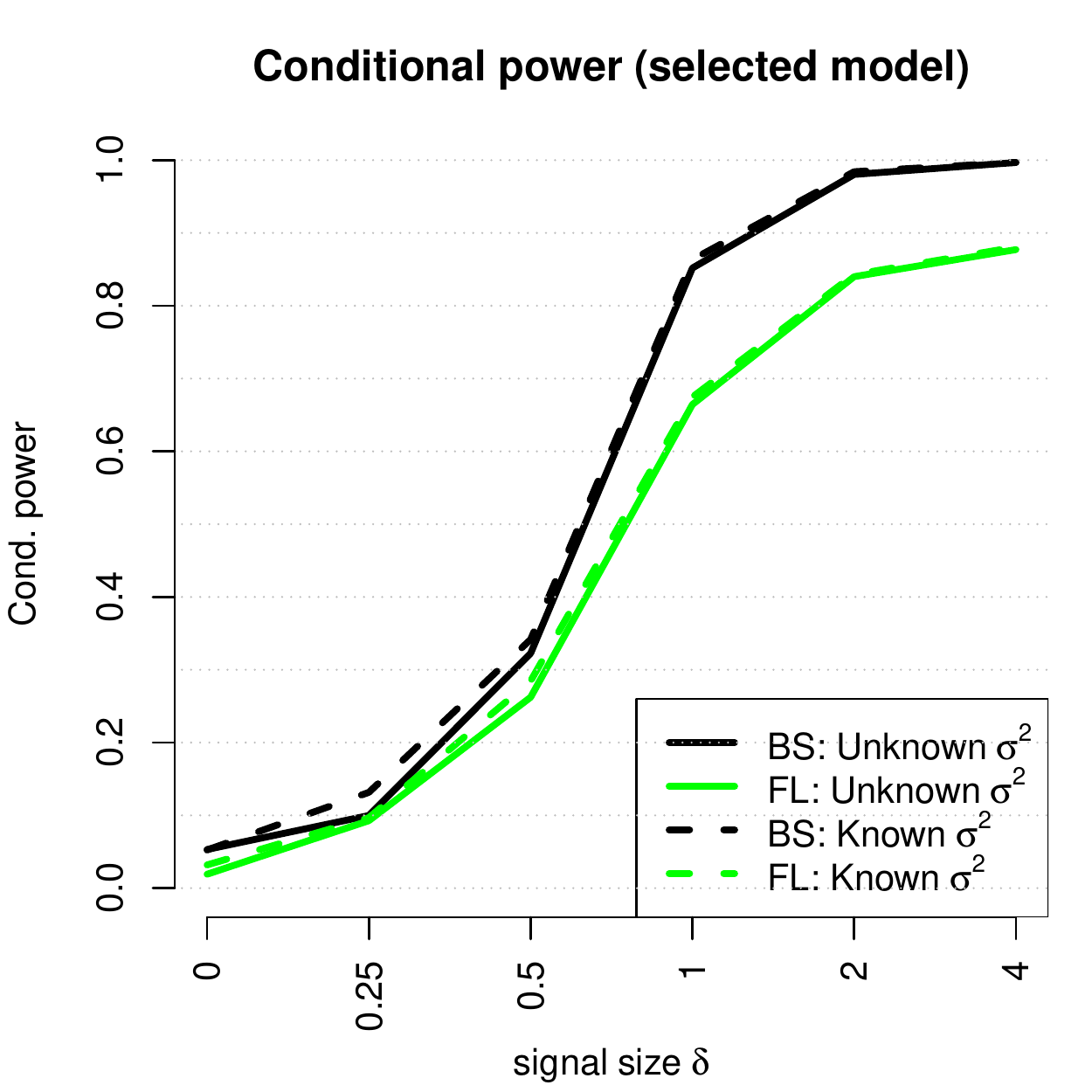}\hspace{-3mm}
    \includegraphics[width=.33\linewidth]{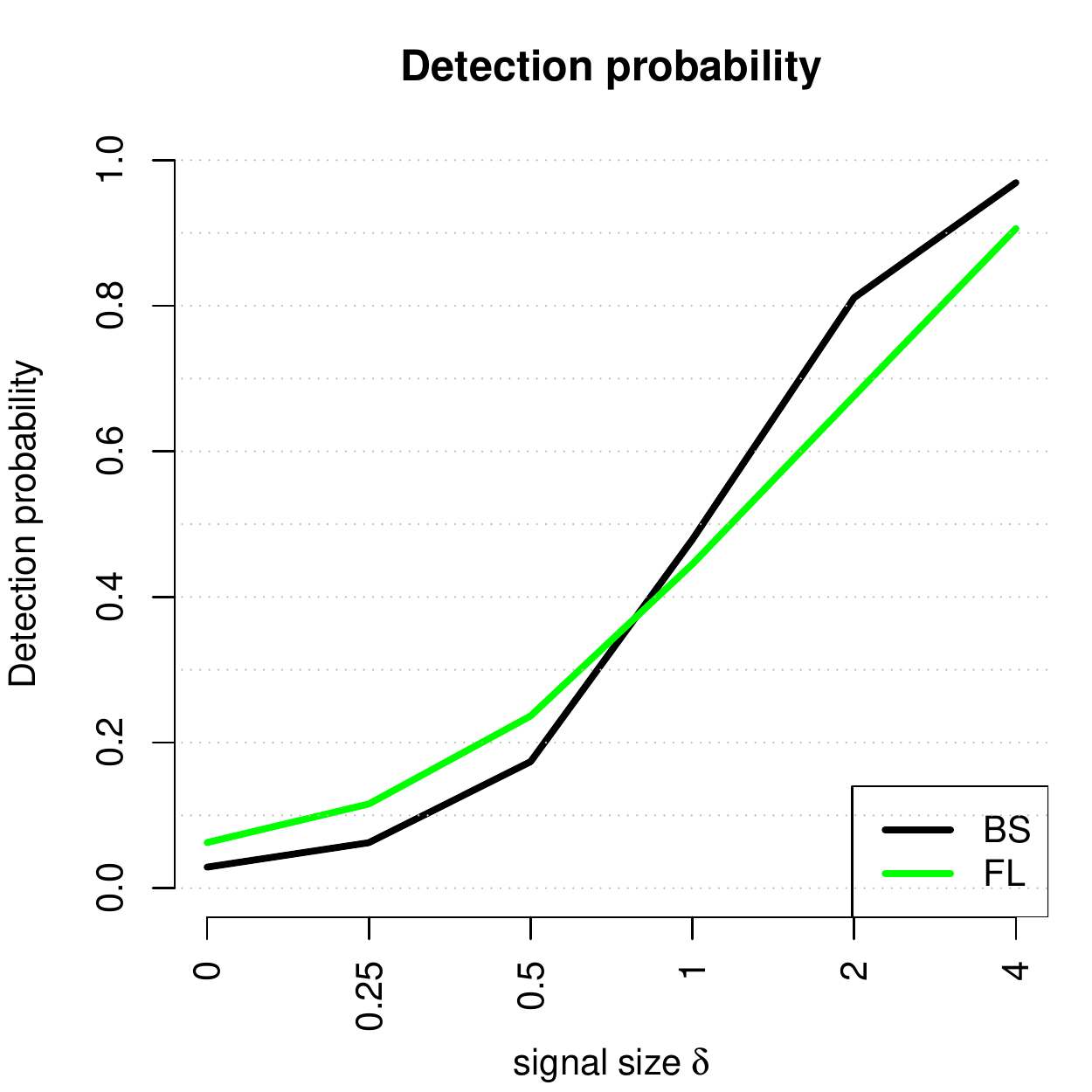}\hspace{-3mm}
    \includegraphics[width=.33\linewidth]{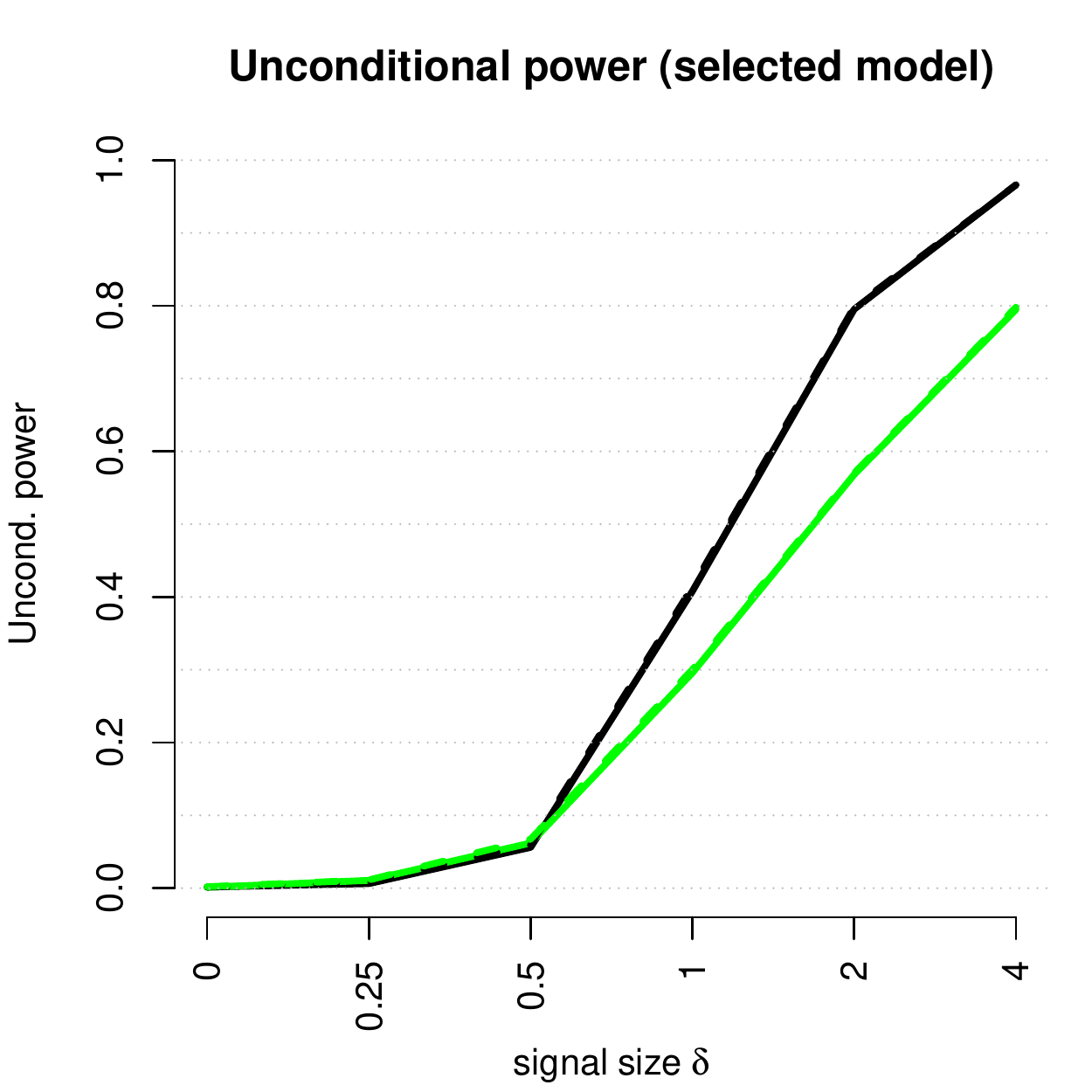}\hspace{-3mm}
    \caption{\it\small Setup similar to \Fref{fig:power-comparison} but for
      selected model tests. Only BS (black) and FL (green)  are shown.
      but the selected model test is applied to both known  (dashed line) and unknown noise
      parameter $\sigma^2$ (solid line). 
      }
\label{fig:power-comparison-selected}

\centering
  \includegraphics[width=.33\linewidth]{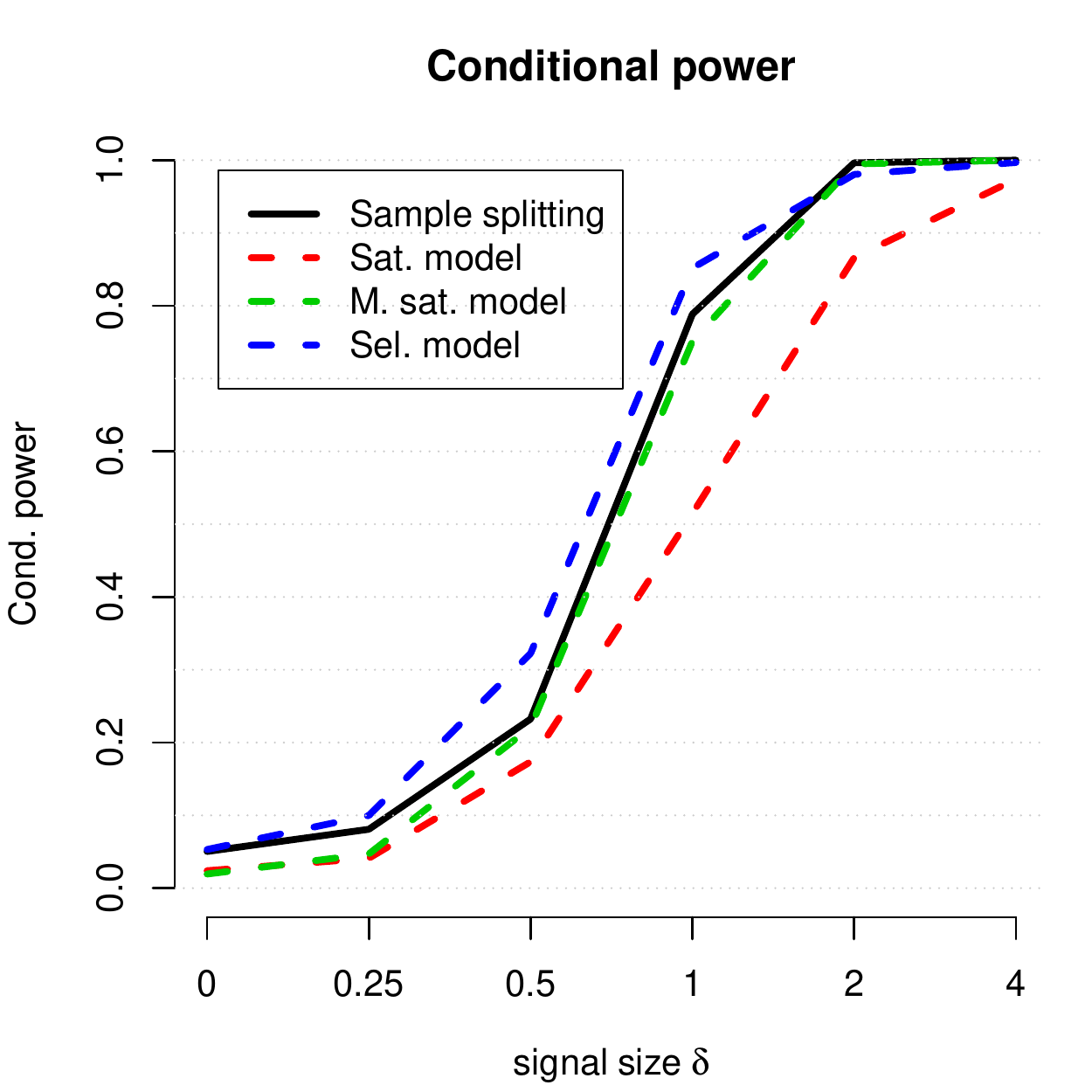}\hspace{-3mm} 
  \includegraphics[width=.33\linewidth]{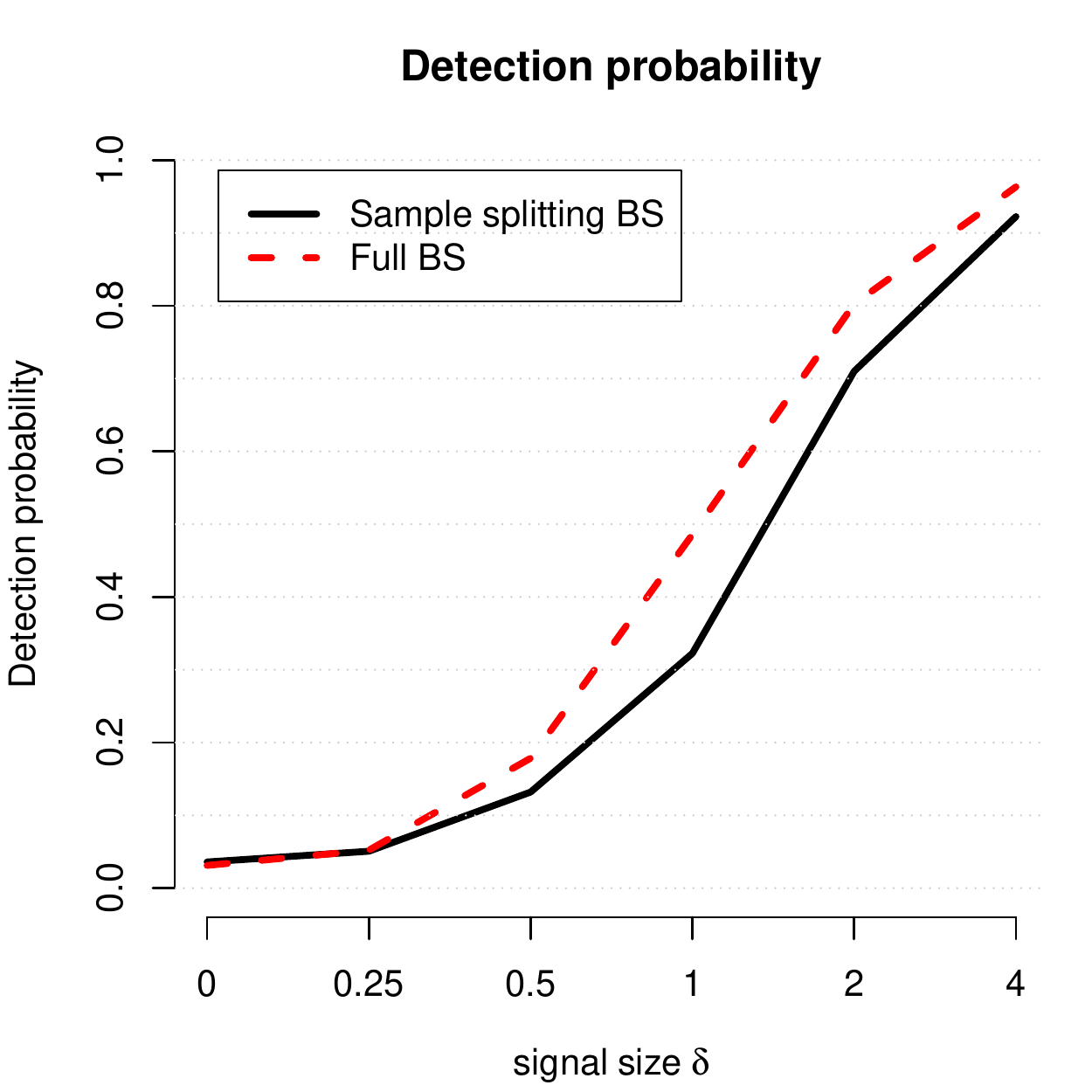}\hspace{-3mm}
  \includegraphics[width=.33\linewidth]{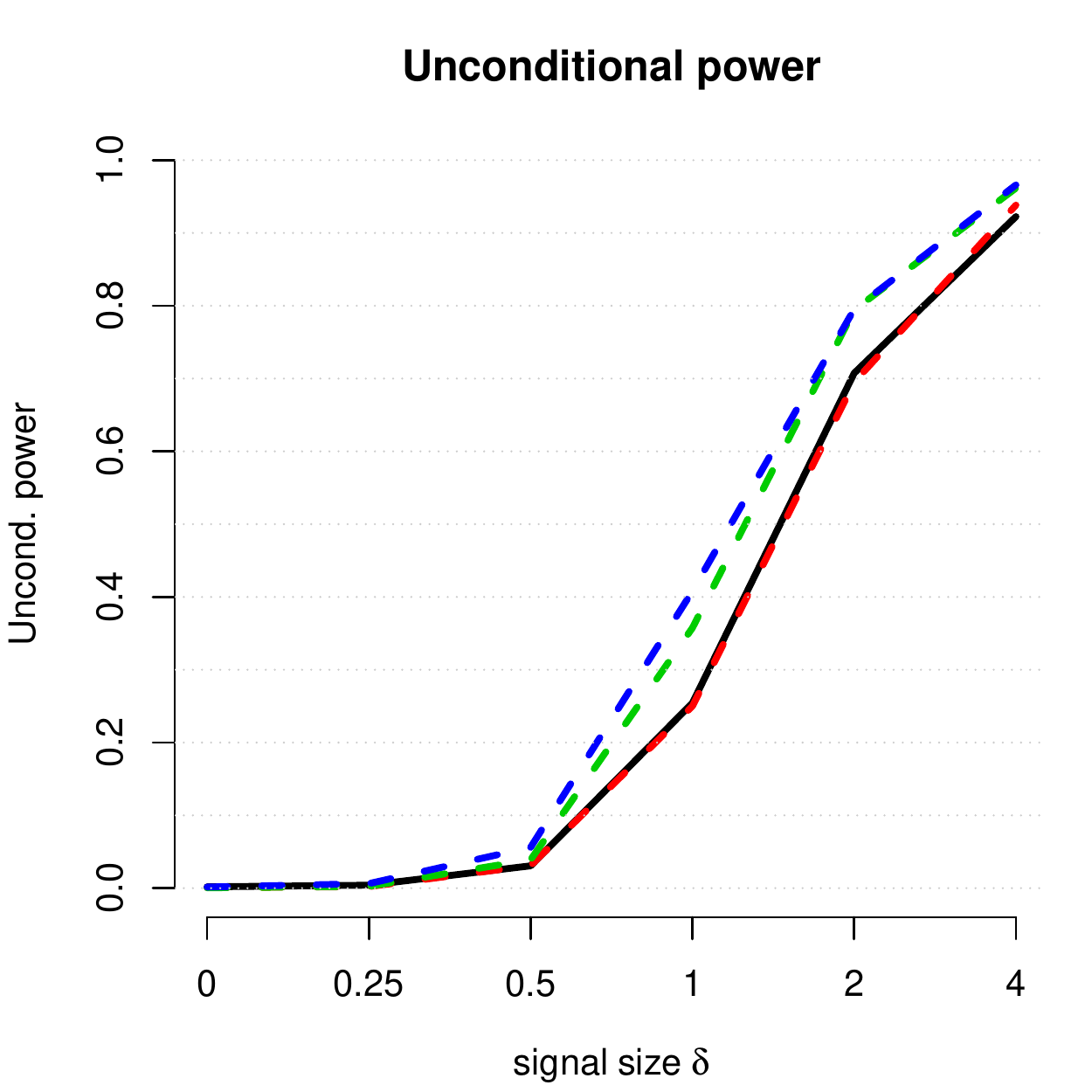}\hspace{-3mm} 
  \caption{\it\small 
  Setup similar to \Fref{fig:power-comparison} but 
  comparing sample
    splitting (black solid), plain saturated model test (red dashed),
    additive noise marginalized saturated model test (green dashed),
    and selected model test with unknown $\sigma^2$ (blue dashed),
   all using a 2-step binary segmentation.
    (Middle): Detection probability for the binary segmentation
    applied on the sample split dataset (black solid)
    or the full dataset (red dashed). 
    (Right): Unconditional power, computed by multiplying the conditional power curve and its
    relevant detection probability curve.
    }
  \label{fig:samplesplit}
\end{figure}




\subsection{Pseudo-real simulation with heavy tails}
\label{sec:heavytail}

We present pseudo-real datasets based on a single chromosome --
chromosome 9 in GM01750 -- in order to investigate how heavy-tailed distributions
affect our inferences. 
We only present saturated model tests for brevity. 
From
the original data, we estimate a 1-changepoint mean $\theta$, shown in the bold red
line in Figure \ref{fig:pseudoreal}, and residuals $r$, both based on a fitted
1-step wild binary segmentation model. 
The QQ plot shows that these residuals have heavier tails than a Gaussian
(top middle panel of \Fref{fig:pseudoreal}), and are close in distribution to
a Laplacian.
This motivates us to generate synthetic
data $y = \theta + \epsilon$ by adding noise $\epsilon$ in three ways:
\begin{enumerate}
\itemsep-.5em 
\item Gaussian noise $\epsilon \sim \cN(0, \sigma^2 I)$ (black), 
\item Laplace noise $\epsilon \sim \operatorname{Laplace}(0, \sigma/\sqrt{2})$ (green), and
\item Bootstrapped residuals, $\epsilon = b(r)$, where $b(\cdot)$ samples the residuals with
  replacement (red).  \label{eq:bootstrap-data}
\end{enumerate}



\begin{figure}[h!]
  \centering
  \begin{subfigure}{.35\linewidth}
    \includegraphics[width=\linewidth]{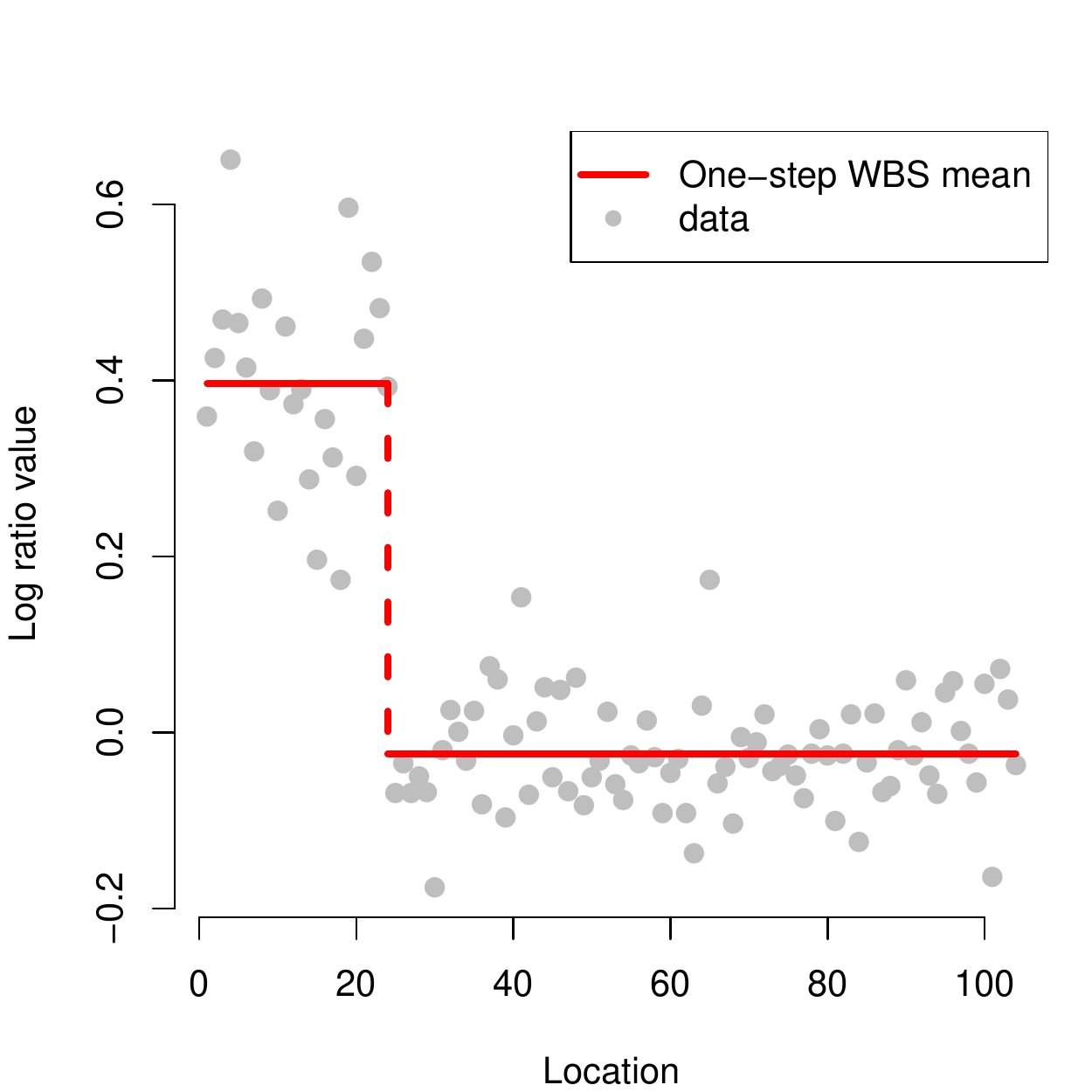} 
  \end{subfigure} 
\begin{subfigure}{.6\textwidth}
 \includegraphics[width=.45\linewidth]{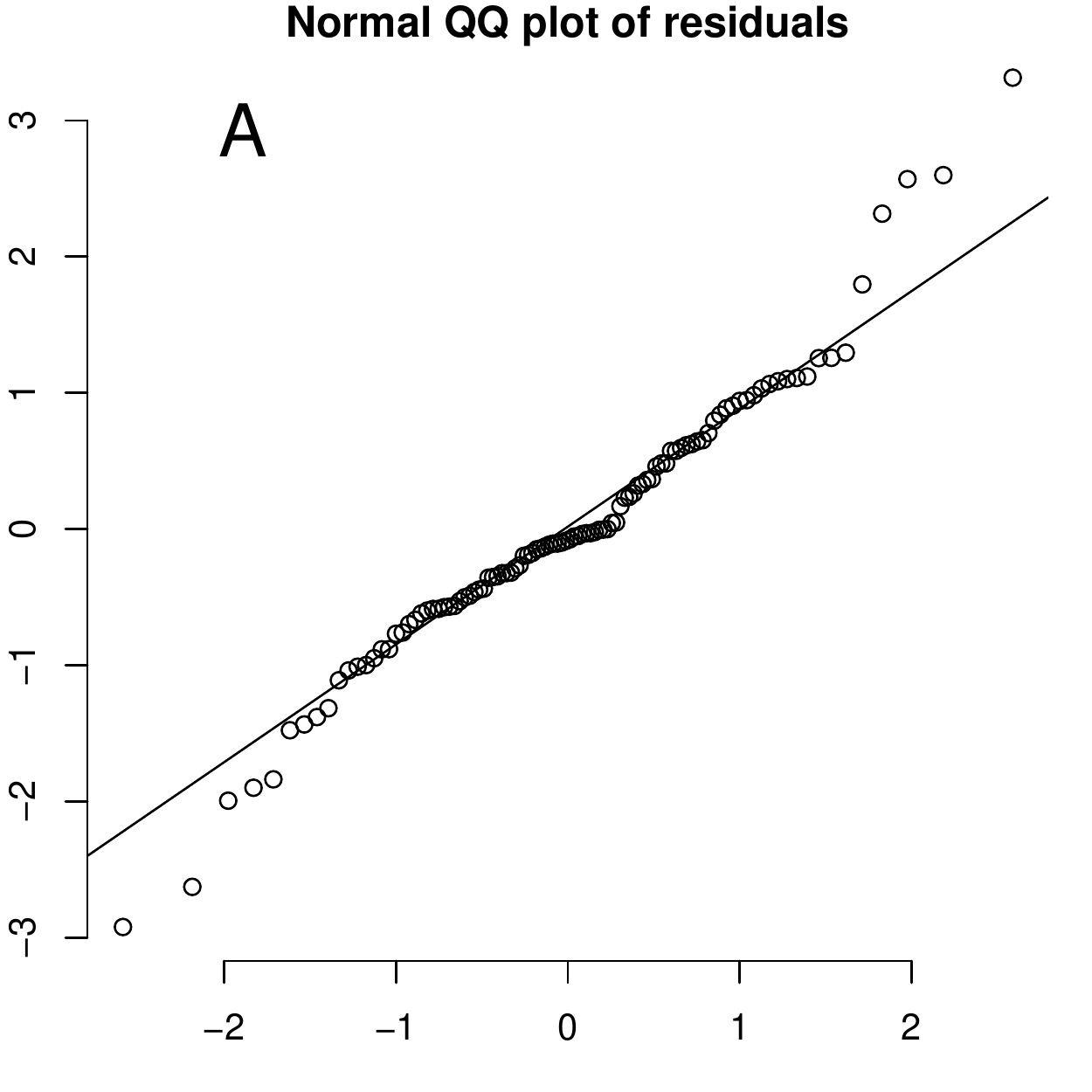} 
  \includegraphics[width=.45\linewidth]{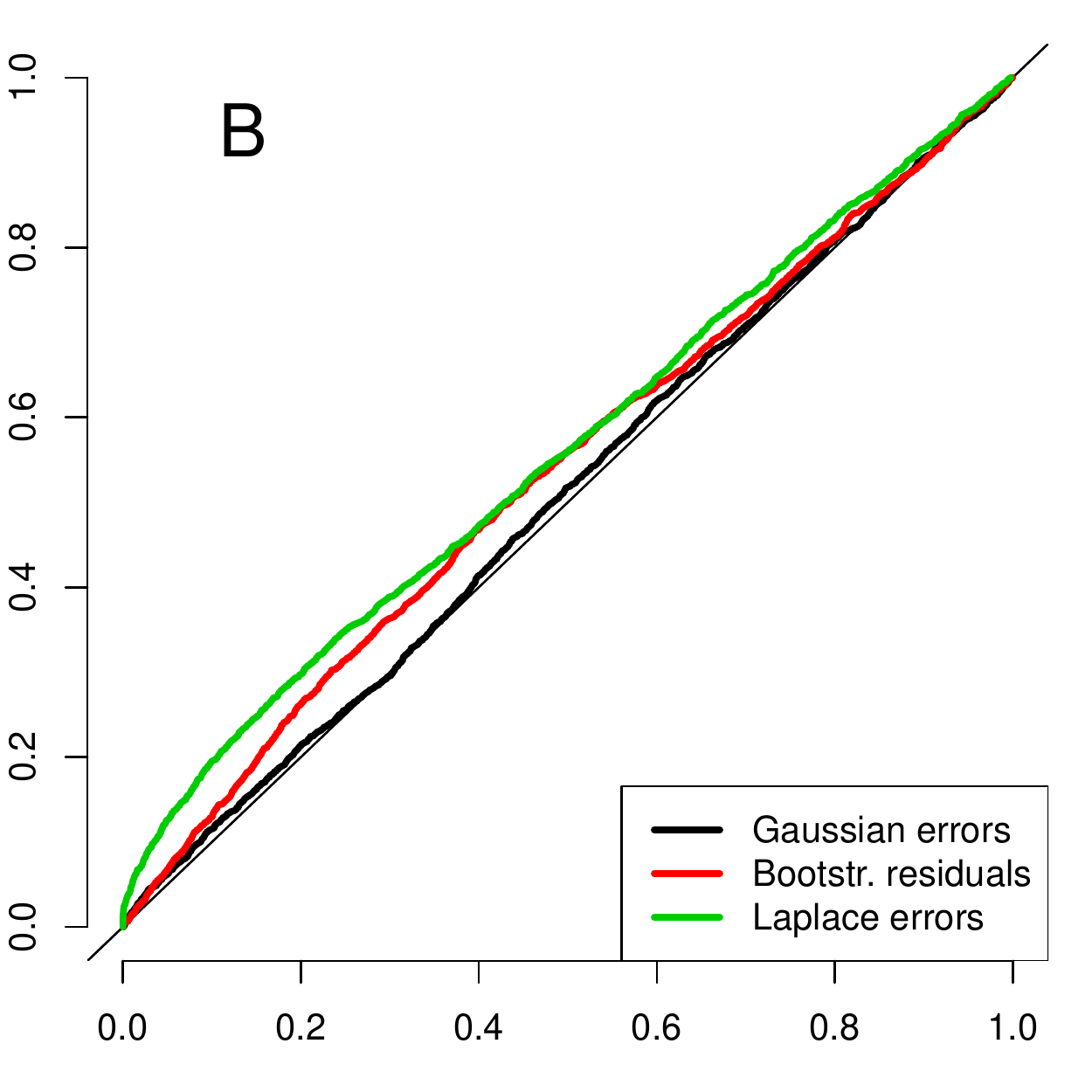} 
  \includegraphics[width=.45\linewidth]{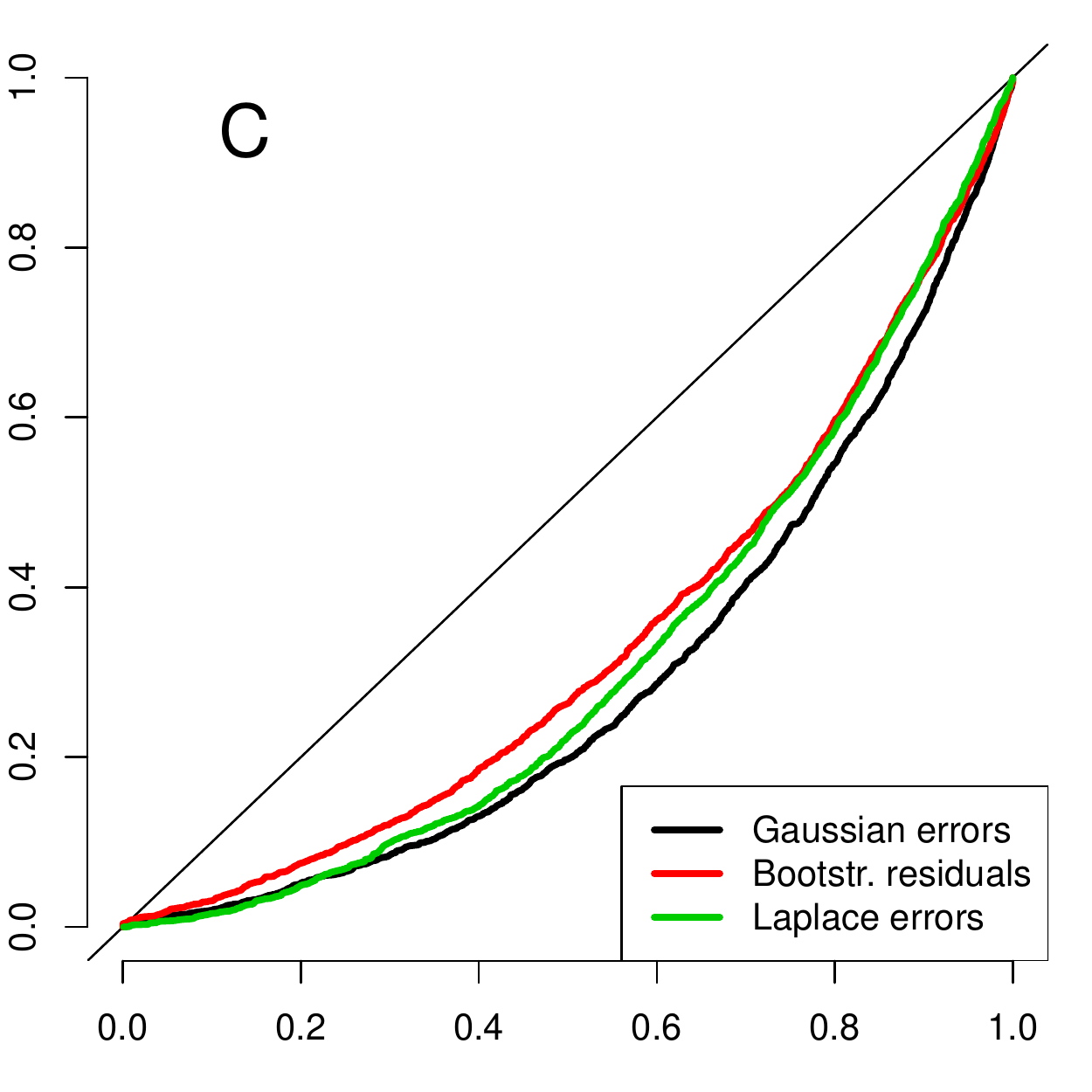}\hspace{4mm}
  \includegraphics[width=.45\linewidth]{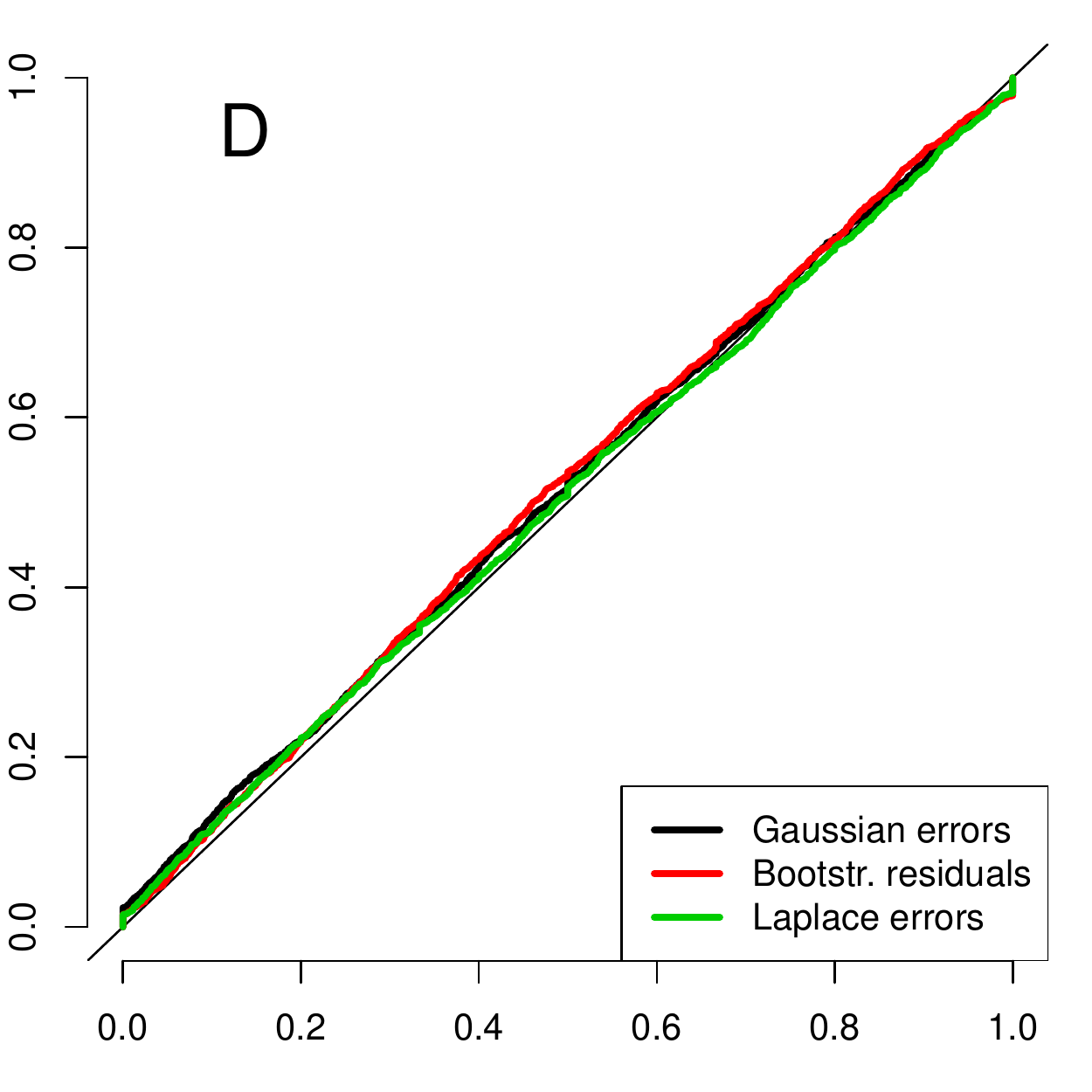}
  \end{subfigure}
  \caption{\small \it {
  (Left) Bootstrapped residuals added to the artificially
    constructed mean, generated from chromosome 9 in GM01750. 
    (Panel A): QQ plot of residuals.  The remaining 3 panels show the p-values
    of saturated model tests under three different noise models, Gaussian
    (black), bootstrapped residuals (red) and Laplacian (green).  (Panel B):
    Application of vanilla saturated model tests (no modifications).  (Panel C):
    P-values after using the bootstrap substitution method \citep{asymppostsel}.
    (Panel D): P-values after using our modified bootstrap substitution method
    that involves bootstrapping $y - \hat\theta$ instead of $y - \bar y$.
    }}
  \label{fig:pseudoreal}
\end{figure}

We then investigate the behavior of saturated model tests after a 3-step binary
segmentation across all three types of noises when the null hypothesis
$H_0: v^T\theta = 0$ is true.  To set $\sigma^2$ for these saturated model
tests, we compute the empirical variance after fitting a pre-cut 10-step wild
binary segmentation across the entire cell line.
The results are shown in Figure
\ref{fig:pseudoreal}. Exactly valid null p-values would follow the theoretical
$U(0,1)$ distribution, optimistic (superuniform) p-values would lie below the
diagonal, and conservative (subuniform) p-values would lie above the diagonal. We
see that the inferences are exactly valid with Gaussian noise 
but is optimistic
with both Laplacian noise and bootstrapped residuals (panel B of \Fref{fig:pseudoreal}).

To overcome this optimism, we modify the \textit{bootstrap substitution method}
\citep{asymppostsel}.  Let $\beta$ denote $\bar \theta$, the grand mean of
$\theta$.  Originally, the authors' main idea is to approximate the law of
$v^TY$ used to construct the TG statistic \eqref{eq:tg_statistic} with the
bootstrapped distribution of $v^T(Y- \beta)$ by bootstrapping the residuals,
$y-\bar y$.  Here, the empirical grand mean $\bar y$ represents the simplest
model with no changepoints.  While this estimate will usually restore validity,
it is expected to produce overly conservative p-values if there exist
\textit{any} changepoints (panel C of \Fref{fig:pseudoreal}).

%


Hence, we instead consider the bootstrapped distribution of $v^T(Y - \theta)$,
by bootstrapping the residuals, $y - \hat{\theta}$,
where 
$\hat \theta$ is a piecewise constant estimate of $\theta$.  For our instance,
we use a $k$-step binary segmentation model to estimate $\hat \theta$, where we
choose $k$ using two-fold cross validation from a two-fold split of the data $y$
into odd and even indices. This procedure is not valid in general and should be
used with caution.
In order to combat the main risk of over-fitting of $\hat \theta$, we may
further modify this procedure by
excluding shorter segments in $\hat \theta$ prior to bootstrapping. For our
dataset, these potential downsides do not seem to come to fruition in
practice. At the sample size $n \simeq 100$ and signal-to-noise ratio of our
current dataset, the resulting p-values in both heavy-tailed and Gaussian data
are convincingly uniform (panel D of \Fref{fig:pseudoreal}).

\section{Copy Number Variation (CNV) data application} \label{sec:application}

Array CGH analyses
detect changes in expression levels (measured as a log ratio in fluorescence intensity
between test and reference samples) across the genome. 
Aberrations found are linked 
with the presence of
a wide range of genetically driven diseases -- as many types of cancer,
Alzheimer's disease, and autism, see, eg. \citet{
  international2008rare, 
  bochukova2010large}.  
%
%

The datasets we study in this paper are originally from \cite{Snijders2001}, and
have been studied by numerous works in the statistics literature, e.g.
\citet{sara, lai2008}. In each dataset consist of
individual cell lines with $2,000$ measurements or more across
23 chromosomes. 
Our analysis focuses on middle-to-middle duplication, the setting that was
studied in \Fref{sec:simulation}.




In our analysis, we use a 4-step wild binary segmentation and perform
marginalized saturated model tests on two cell lines GM01524 and GM01750 in
\Fref{fig:analysis}.  Recall that the 14th chromosome of the latter cell line
was shown in \Fref{fig:intro}. As decribed in \Fref{sec:practicalities}, we
pre-cut both analyses at chromosome boundaries since the ordering of 1 through
23 is essentially arbitrary. In GM01524, we can see that the our choice of
methods -- segment test inferences on changepoints recovered from pre-cut
wild binary segmentation, after decluttering -- deems two changepoint locations
A and B of alternating directions in chromosome 6 to be significant, and two
other locations to be spurious, at the signifance level $\alpha = 0.05$ after
Bonferroni correction. This result is consistent with karyotyping results of a
single middle-to-middle duplication. Likewise, in GM01750, the wild binary
segmentation inference correctly identified the two start-to-middle duplications
in chromosomes 9 and 14 which were confirmed with karyotyping, and correctly
invalidated the rest.  

\section{Conclusions}
%
  We have described an approach to conduct post-selection inference on
changepoints detected by common segmentation algorithms, using the same data for
detection and testing. 
Through simulations, we demonstrated the detection probability and power over
signal-to-noise ratios in a variety of settings, as well as our tools' robustness to heavy-tailed data.
Finally, we demonstrated the
application in array CGH data, where we show that our methods effectively
provide a statistical filter 
that retains the changepoints that validated by karyotyping and discards the rest.

Future work in this area could improve the practical applicability of these
methods. One useful extension would be to incorporate more complex and realistic
noise models. 
For example, the selected model testing framework can be extended to include
other exponential family models. 
The methodology for inference after changepoint detection may also be extended
to multiple streams of copy number variation data 
in order to make more powerful 
inferences about changepoint locations. These and other methodological
extensions can be useful for newer types of copy number variation data from
recent technology, such as next-generation sequencing.



\begin{figure}[ht!]
  \centering
  \includegraphics[width=0.9\textwidth]{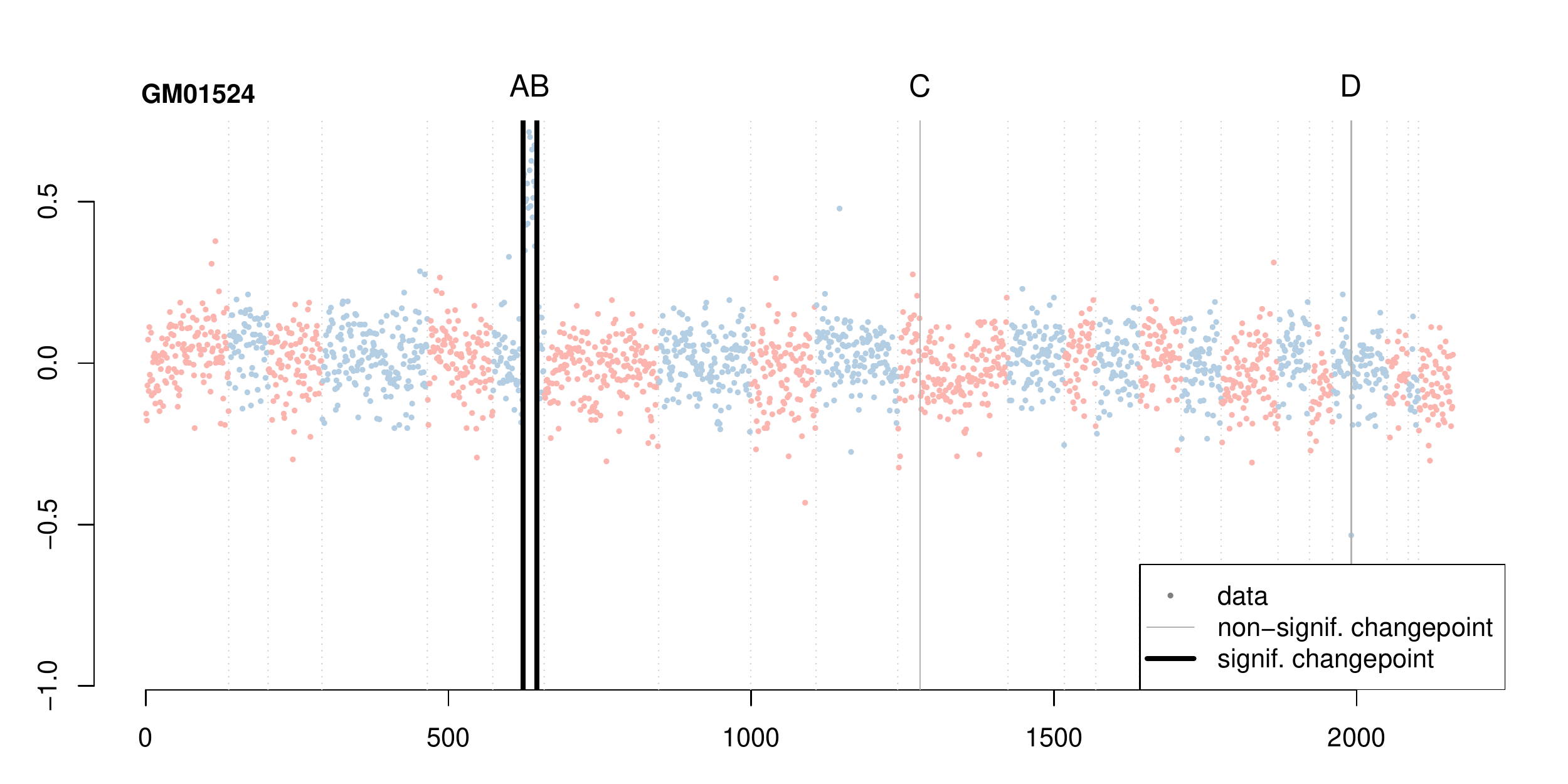} 
  \includegraphics[width=0.9\textwidth]{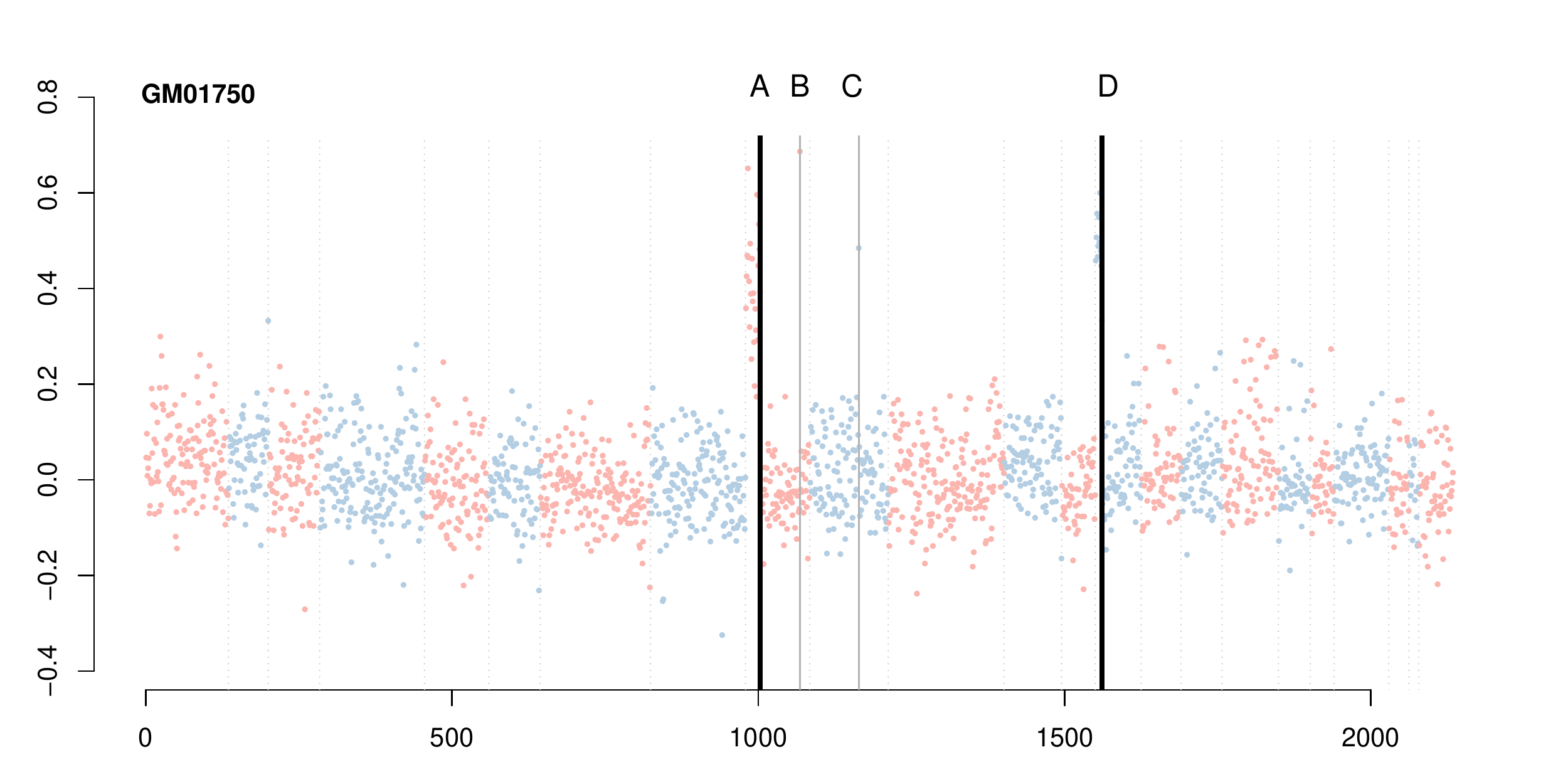}
  \caption{\it \small ``Pre-cut'' changepoint inference using
   saturated model tests for wild binary segmentation 
  marginalized over random intervals conducted on two cell lines,
    from \citet{Snijders2001}. Data points are colored in two alternating
    tones, to visually depict the chromosomal boundaries. 
    For each cell line, the letters A through D
    denote the estimated changepoints, $\hat b_1$ through $\hat b_4$ respectively.
    The bolded lines denote changepoints that were rejected under the null hypothesis 
    $H_0: v^T\theta = 0$ at a Type-I error control level $\alpha = 0.05$ after Bonferroni-correction.
    (Top): The analysis for the cell line GM01524, with all 23 chromosomes shown.
    (Bottom): The same setup as above, but for the cell line GM01750.
  }
  \label{fig:analysis}
\end{figure}

\begin{small}
\section{Code and supplemental material}
The code to perform estimation as well as saturated model tests are in
\url{https://github.com/robohyun66/binseginf}, while the code to perform
selected model tests are additionally in
\url{https://github.com/linnylin92/selectiveModel}.

The following is a brief summary of the supplements.
Appendix  A contains the proofs omitted from the main text.
Appendix  B contains the algorithmic details for the 
selected model test sampler in the known $\sigma^2$ setting.
Appendix  C contains numerous additional simulations results and details.
Appendix  D contains a description of the procedure
to choose $k$ adaptively and its corresponding simulation results.
Appendix  E contains additional results on our array CGH application.

\section{Acknowledgment}
The authors used Pittsburgh Supercomputing Center resources (Proposal/Grant
Number: DMS180016P). Sangwon Hyun was supported by supported by NSF grants
DMS-1554123 and DMS-1613202. Max G'Sell was supported by NSF grant
DMS-1613202. Ryan Tibshirani was supported by NSF grant DMS-1554123.

\bibliographystyle{agsm}
\bibliography{all}
\end{small}

\newpage
\appendix

\section{Additional proofs} \label{sec:proofs}

\subsection{Proof of \Fref{prop:wbs-polyhedral-event}, (WBS)}  \label{app:wbs_polyhedra}
\begin{proof}
The construction of $\Gamma$ is basically the same as that for BS in
\Fref{prop:bs-polyhedral-event}; the only difference is that, at step $k$, the
inequalities defining the new rows of $\Gamma$ are based on the intervals
\smash{$w_{j_k}$} and $w_\ell$, $\ell \in J_k \backslash \{j_k\}$, instead of
\smash{$I_{j_k}$} and $I_\ell$, $\ell \neq j_k$, respectively. 
To compute the upper bound on the number of rows $m$, observe that in step $\ell \in \{1,\ldots, k\}$, 
there are at most $B-\ell+1$ intervals remaining. Among these, the interval $j_k$ contributes $p-2$ inequalities, and the remaining $B-\ell$ intervals contributes $p-1$ inequalities.
\end{proof}

\subsection{Proof of \Fref{prop:cbs-polyhedral-event}, (CBS)} \label{app:cbs_polyhedra}

\begin{proof}
The proof follows similarly to the proof of 
\Fref{prop:bs-polyhedral-event}.
Observe that for any $k' < k$, the model $M^{\mathrm{CBS}}_{1:k'}(y_\obs)$ is strictly contained
in the model $M^{\mathrm{CBS}}_{1:k}(y_\obs)$. Hence, we can proceed using induction, and 
let $ b_i$ for $i \in \{1, \ldots, k\}$ denote $\hat b_i$ for simplicity, 
and do the same for $a_i$, $d_i$ and $j_i$. Let $C(x,2) = {x \choose 2}$ for simplicity as well.

For $k=1$, the following
$2 \cdot (C(n-1, 2)- 1)$ inequalities characterize the selection of the changepoint
model $\{a_1, b_1, d_1\}$,
 \begin{align*}
d_1 \cdot g^T_{(1, a_1, b_1, n)}y \geq g^T_{(1,r,t,n)}y, \quad \text{and}\quad   d_1 \cdot g^T_{(1, a_1, b_1, n)}y \geq -g^T_{(1,r,t,n)}y,
\end{align*}
for all $r, t \in \{1, \ldots, n-1\}$ where $r < t$, $r \neq a_1$ and $t \neq b_1$.

By induction, assume we have constructed the polyhedra for the model,
$M^{\mathrm{CBS}}_{1:(k-1)}(y_\obs) = \{a_{1:(k-1)}, b_{1:(k-1)}, d_{1:(k-1)}\}$.
To construct $M^{\mathrm{CBS}}_{1:k}(y_\obs)$,
all that remains is to characterize the $k$th parameters $\{a_k, b_k, d_k\}$.
To do this, assume that 
$j_k$ corresponds with the interval $I_k$ having the form $\{s_k, \ldots, e_k\}$.
Within this interval, we form the first $2 \cdot (C(|I_{j_k}|-1, 2)-1)$ inequalities of the form,
\begin{equation*}
 d_k \cdot g^T_{(s_k, a_k, b_k, e_k)}
 y \geq g^T_{(s_k, r,t, e_k)} y  \quad\text{and}\quad
  d_k \cdot g^T_{(s_k, a_k,b_k, e_k)}
 y \geq -g^T_{(s_k, r,t, e_k)} y 
  \end{equation*}
for all $r,t \in \{s_k, \ldots, e_k - 1\}$ where $r < t$ and $r \neq a_k$ and $t \neq b_k$.
 The remaining inequalities originate from the remaining intervals.
 For each interval $I_\ell$, for $\ell \in \{1, \ldots, 2 k -1\}\backslash \{j_k\}$,
 let $I_\ell$ have the form $\{s_\ell, \ldots, e_\ell \}$.
 We form the next $2 \cdot C(|I_\ell| -1, 2)$ inequalities of the form
 \begin{equation*}
 d_k \cdot g^T_{(s_k, a_k, b_k, e_k)}
 y \geq g^T_{(s_\ell, r, t, e_\ell)} y  \quad\text{and}\quad
  d_k \cdot g^T_{(s_k, a_k, b_k, e_k)}
 y \geq -g^T_{(s_\ell, r, t, e_\ell)} y 
 \end{equation*}
for all $r,t \in \{s_\ell, \ldots, e_\ell - 1\}$ where $r < t$.
\end{proof}

\subsection{Proof of \Fref{prop:additive_noise}, (Marginalization)}

\begin{proof}
For concreteness, we write the proof where $W$ represents additive noise, but
the proof generalizes to the setting where $W$ represents random intervals easily.
First write $T(y_\obs, v)$ as an integral
over the joint density of $W$ and $Y$,
\begin{align}
T(y_\obs, v) &= P(v^T Y \ge v^Ty_\obs|M(Y+W) = M(y_\obs + W),
\Pi_v^\perp Y = \Pi_v^\perp y_\obs) \nonumber\\
       &= \int \one (v^Ty \ge v^Ty_\obs) f_{W,Y|E_1,E_2}(w,y) dwdy . \label{eq:orig-randtg2}
\end{align}
Then the joint density $f_{W,Y|E_1,E_2}(w,y)$ partitions into two components,
whose latter component (a probability mass function) can be rewritten using
Bayes rule. For convenience, denote $g(w)=\mathbb{P}(E_1 |W = w, E_2)$.
\begin{align*}
  f_{W,Y|E_1,E_2}(w,y) dy dw &= f_{Y|W=w,E_1,E_2}(y) \cdot
                               f_{W|E_1,E_2}(w) \;dy\; dw\\
                             &= f_{Y|W=w,E_1,E_2}(y) \cdot
  \frac{\mathbb{P}(E_1|W=w,E_2)f_{W|E_2}(w)}{\mathbb{P}(E_1|E_2)} \; dy \; dw\\
                             &= f_{Y|W=w,E_1,E_2}(y) \cdot
  \frac{g(w) f_{W}(w)}{\int g(w') f_{W}(w') dw'} \; dy \; dw,
\end{align*}
where we used the independence between $W$ and $E_2$ in the last equality.
With this, $T(y_\obs, v) $ from \eqref{eq:orig-randtg2} becomes:
\begin{equation*}
T(y_\obs, v)  = \int \one (v^Ty \ge v^Ty_\obs) \cdot g(w) \cdot \frac{f_{W|E_2}(w)}{\int g(w') f_{W}(w') dw'}
  \cdot f_{Y|W=w,E_1,E_2}(y)\;dy\; dw.
\end{equation*}
Now, rearranging, we get:
\begin{align}
T(y_\obs, v)  &= \int   \underbrace{ \left[\int \one (v^Ty \ge v^Ty_\obs)\cdot
         f_{Y|W=w,E_1,E_2}(y) dy  \right]}_{T(y_\obs, v, w) } \underbrace{\frac{g(w)}{\int
         g(w') f_{W}(w') dw'}}_{a(w)} f_{W}(w)dw  \nonumber \\
       &= \int T(y_\obs, v, w) a(w) \; f_{W}(w)\; dw. \label{eq:simpler-final-form}
\end{align}
This proves the first equality in \Fref{prop:additive_noise}. To show 
what the weighting factor $a(w)$ equals,
observe that
by applying Bayes rule to the numerator of $a(w_\obs)$, and rearranging:
\begin{align*}
  a(w) &=   \frac{g(w)}{\int g(w') f_W(w')\;dw'} = \frac{\mathbb{P}(E_1|E_2, W=w)}{P(E_1|E_2)} = \frac{\mathbb{P}(W=w |E_1, E_2) }{\mathbb{P}(W=w |E_2)}\\
  &= \frac{\mathbb{P}(W=w |E_1,E_2)}{\mathbb{P}(W=w)}.
\end{align*}

Finally, to show the seocnd equality in \Fref{prop:additive_noise}, observe that
we can also represent $a(w)$ as
\begin{equation} \label{eq:a}
a(w) = \frac{g(w)}{\mathbb{E}[g(w)]}
\end{equation}
by definition, where the denominator is the
expectation taken with respect to the random variable $W$.
Leveraging the geometric theorems of \cite{lee2016exact,tibshirani2016exact},
it can be shown that
\begin{equation} \label{eq:g}
g(w) = P\Big(M(Y+W) = M(y_\obs + W) ~|~ \Pi^\perp_vY = \Pi^\perp_v y_\obs\Big)
= \Phi(\cV_{\text{up}}/\tau) - \Phi(\cV_{\text{lo}}/\tau).
\end{equation}
Also from the same references as well as stated in \Fref{sec:randomization}, we know that 
\begin{equation} \label{eq:t}
T(y_\obs, v, w) = \frac{\Phi(\cV_{\text{up}}/\tau) - \Phi(v^Ty_\obs/\tau)}{\Phi(\cV_{\text{up}}/\tau) - \Phi(\cV_{\text{lo}}/\tau)}
\end{equation}
Putting \eqref{eq:a}, \eqref{eq:g} and \eqref{eq:t} together into
\eqref{eq:simpler-final-form}, we complete the proof by obtaining
\begin{align*}
T(y_\obs, v)  = \frac{\int T(y_\obs, v, w)  g(w) f_W(w) dw}{ \int g(w) f_W(w) dw}
= \frac{\int \Phi(\cV_{\text{up}}/\tau) - \Phi(v^Ty_\obs/\tau) f_W(w) dw}{ \int \Phi(\cV_{\text{up}}/\tau) - \Phi(\cV_{\text{lo}}/\tau) f_W(w) dw}.
\end{align*}
\end{proof}

\section{Selected model tests, hit-and-run sampling for known
$\sigma^2$} \label{app:known_sigma}

The following is the hit-and-run sampler to estimate the tail probability of the law 
of \eqref{eq:selective-distribution-saturated}.
This is for the known $\sigma^2$ setting, which differs from the setting described in the main
text in \Fref{sec:computation}.
This was briefly described in \cite{fithian2015selective} 
but the authors have later implemented it in ways not originally
described in the above work to make it more efficient.
We do not claim novelty for the following algorithm, but simply state it for completion.
The original code can be found the repository \url{https://github.com/selective-inference},
and we reimplemented it to suite our coding framework and simulation setup.

We specialize our description to test the null hypothesis $H_0: v^T\theta = 0$ against the one-sided
alternative $H_1: v^T\theta > 0$.
There are some notation to clarify prior to describing the algorithm. 
Let $v \in\mathbb{R}^n$ denote the vector such that
\[
v^T y = \bar{y}_{(\hat c_j + 1):\hat c_{j+1}} - \bar{y}_{(\hat c_{j-1} + 1):\hat c_{j}}.
\]
As in \Fref{sec:computation}, let $A \in \mathbb{R}^{k\times n}$ denote the matrix such that the
last $k$ equations in the above display are satisfied if and only if $AY = Ay_\obs$.
Based on \Fref{sec:polyhedra},
observe that our goal reduces to sampling from the $n$-dimensional distribution 
\begin{equation}\label{eq:full_gaussian}
Y \sim \cN(0, \sigma^2 I_n), \quad \text{conditioned on} \quad \Gamma Y \geq 0,\; AY = Ay_\obs.
\end{equation}
where $I_n$ is the
$n \times n$
identity matrix.

The first stage of the algorithm \emph{removes the nullspace} of $A$ in the following sense.
Construct any matrix $B \in \mathbb{R}^{n \times n}$ such that it has full rank and the last 
$k$ rows are equal to $A$. 
Then, consider the following $n$-dimensional distribution.
\begin{equation} \label{eq:no_nullspace_gaussian1}
Y' \sim \cN(0, \sigma^2 B^TB), \quad \text{conditioned on} \quad  \Gamma B^{-1} Y' \geq 0,\; (Y')_{(n-k+1):n} = Ay_\obs.
\end{equation}
Note that $B^{-1}Y'$ has the same law as \eqref{eq:full_gaussian}.
Observe that the above distribution is a conditional Gaussian, meaning we can remove
the last conditioning event. Towards that end,
let $\Gamma''$ denote the first $n-k$ columns of the matrix $\Gamma B^{-1}$, and
let $u''$ denote the last $k$ columns of $\Gamma B^{-1}$ left-multiplying $Ay_\obs$.
Also, consider the following partitioning of the matrix $B^TB$,
\[
\sigma^2 B^TB = \begin{bmatrix}
B_{11} & B_{12} \\
B_{12}^T & B_{22}
\end{bmatrix},
\]
where $B_{11}$ is a $(n-k) \times (n-k)$ submatrix, $B_{12}$ is a $(n-k) \times k$ submatrix,
and $B_{22}$ is a $k\times k$ submatrix. 
Then, consider the following $n-k$-dimensional distribution.
\begin{equation}\label{eq:no_nullspace_gaussian2}
Y'' \sim \cN\Big(B_{12}B_{22}^{-1}(Ay_\obs), \; B_{11} - B_{12}B_{22}^{-1}B_{12}^T\Big), \quad \text{conditioned on} \quad  \Gamma'' Y'' \geq -u''.
\end{equation}
Note that $Y''$ has the same law as the first $n-k$ coordinates of \eqref{eq:no_nullspace_gaussian1}.

The next stage of the algorithm \emph{whitens} the above distribution so its covariance is the
identity.
Let $\mu''$ and $\Sigma''$ denote the mean and variance of the unconditional form
of the above distribution \eqref{eq:no_nullspace_gaussian2}. Let $\Theta$ be the matrix such that
$\Theta \Sigma'' \Theta^T = I_n$. This must exist since $\Sigma''$ is positive definite.
Consider the following $n-k$ dimensional distribution,
\begin{equation}\label{eq:conditional_gaussian}
Z \sim \cN (0, I_n) , \quad \text{conditioned on} \quad  \Gamma'' \Theta^{-1} Z \geq -u'' -\Gamma''\mu''.
\end{equation}
Note that $\Theta^{-1}Z+ \mu''$ has the same law as \eqref{eq:no_nullspace_gaussian2}. Hence, 
we have constructed linear mapping $F$ and $G$ between
\eqref{eq:full_gaussian} and \eqref{eq:conditional_gaussian} 
such that $F(Y) \overset{d}{=} Z$, and
$G(Z) \overset{d}{=} Y$.

In order to set up a hit-and-run sampler, generate $p$ unit vectors $g_1, \ldots, g_p$. (The choice of $p$ is arbitrary, and the specific method of generating these $p$ vectors is also arbitrary.)
Our hit-and-run sampler with move in the linear directions dictated by $g_1, \ldots, g_p$.
We are now ready to describe the hit-and-run sampler in \Fref{alg:hitandrun_knownsigma}, which leverages many of the same
calculations in  \eqref{eq:tg_statistic} and \eqref{eq:vlo_vup}.
The similarity arises since  $\Pi_{g_i}^{\perp} Z = \Pi_{g_i}^{\perp}(Z +  g_i)$ by definition of
projection.

\begin{algorithm}[t]
Choose a number $M$ of iterations.\\
Set $z^{(0)} = F(y_\obs)$, as described in the text.\\
Generate $p$ unit directions $g_1, \ldots, g_p$, each vector of length $n$.\\
Compute $U = \Gamma'' \Theta^{-1} z^{(0)} +u'' + \Gamma'' \mu''$, which represents the ``slack'' of each constraint.\\
Compute the $p$ vectors, $\rho_i = \Gamma'' \Theta^{-1} g_i$ for $i \in \{1,\ldots, p\}$.\\
 \For{$m \in \{1,\ldots,M\}$}{
 Select an index $i$ uniformly from $1$ to $p$.\\
Compute the truncation bounds
\[
\cV_{\text{lo}} = g_i^Tz^{(m-1)} - \min_{j:(\rho_i)_j > 0} U_j/(\rho_i)_j, \quad\text{and}\quad
\cV_{\text{up}} = g_i^Tz^{(m-1)} - \max_{j:(\rho_i)_j < 0} U_j/(\rho_i)_j.
\]\\
Sample $\alpha^{(m)}$ from a Gaussian with mean $g_i^Tz^{(m-1)}$ and
variance $1$, truncated to lie between $\cV_{\text{lo}}$ and $\cV_{\text{up}}$.\\
Form the next sample
\[
z^{(m)} = z^{(m-1)} + \alpha^{(m)} g_i, \quad \text{and}
\quad y^{(m)} = G(z^{(m)}).
\]\\
Update the slack variable,
\[
U \leftarrow U + \alpha^{(m)} \rho_i.
\]
}
Return the approximate for the tail probability of \eqref{eq:selective-distribution-selected-known-sigma},
$
\sum_{m=1}^{M} \one[v^Ty^{(m)} \geq v^Ty_\obs]/M.
$
 \caption{MCMC hit-and-run algorithm for selected model test with known $\sigma^2$} \label{alg:hitandrun_knownsigma}
\end{algorithm}

The computational efficiency of the above algorithm comes from the fact that little
multiplication needs to be done with the polyhedron matrix
$\Gamma'' \Theta^{-1}$, a potentially huge matrix. 
$U$ and $\rho_1, \ldots, \rho_p$, each vectors of the
same length, carry all the information needed about polyhedron 
throughout the entire procedure of generating $M$ samples.

\section{Additional simulation results}\label{app:simulations}

\subsection{Power comparison using unique detection}
\label{app:unique-detection}

Fused lasso was
      appeared to have a large drop in power compared to segmentation
      algorithms. 
In addition to these three measures shown in \Fref{sec:simulation}, for multiple
changepoint problems like middle mutations it is useful to measure performance
using an alternative measure of detection called unique detection. This is
useful because some algorithms -- mainly fused lasso, but to also binary
segmentation to some extent, primarily in later steps -- admit ``clumps'' of
nearby points. If this clumped detection pattern occurs in early steps, the
algorithm requires more steps than others to fully admit the correct
changepoints. In this case, detection alone is not an adequate metric, and
unique detection can be used in place.
\begin{equation} 
  \text{Unique detection probability} =  \frac{\# \text{changepoints which were
      approximately  detected}}{ \#
    \text{number of true changepoints.}}\label{eq:powdef4}\\
\end{equation}
In plain words, unique detection is measuring how many of the true changepoint
locations have been approximately recovered.

We present a simple case study. In addition to a 2-step fused lasso, imagine
using a 3-step fused lasso, but with post-processing. For post-processing,
declutter by centroid clustering with maximum distance of 2, and test the
$k_0<3$ changepoints, pitting the resulting segment test p-values against
$0.05/k_0$. A 2-step fused lasso's detection does not reach 1 even at high
signals ($\delta=4$) because of the aforementioned clumped detection
behavior. The resulting segment tests are also not powerful, since the segment
test contrast vectors consist of left and right segments which do not closely
resemble true underlying piecewise constant segments in the data. However, when
detection is replaced with unique detection, two things are noticeable. First,
decluttered lasso's detection performance is noticeably improved when going from
2 to 3 steps. Also, when unconditional power is calculated using unique
detection, binary segmentation does not have as large of an advantage over the
the several variants of fused lasso.
This is shown in \Fref{fig:unique-power-comparison}.
We see from the right figure (compared to
      the left) that the a ``decluttered'' version of 2- or 3-step fused lasso
      has much closer unconditional power to binary segmentation. 

\begin{figure}
  \centering
    \includegraphics[width=.33\linewidth]{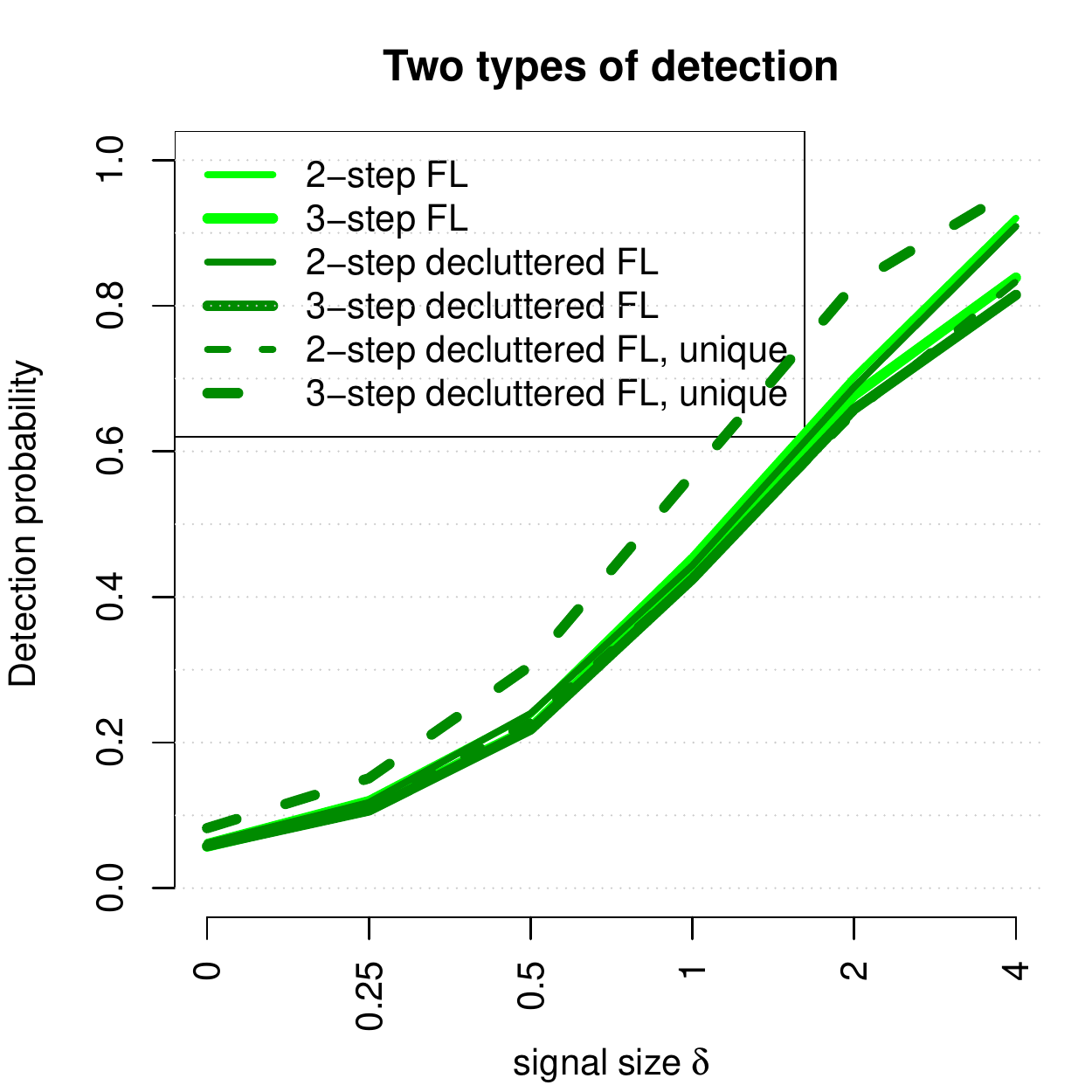}\hspace{-3mm} 
    \includegraphics[width=.33\linewidth]{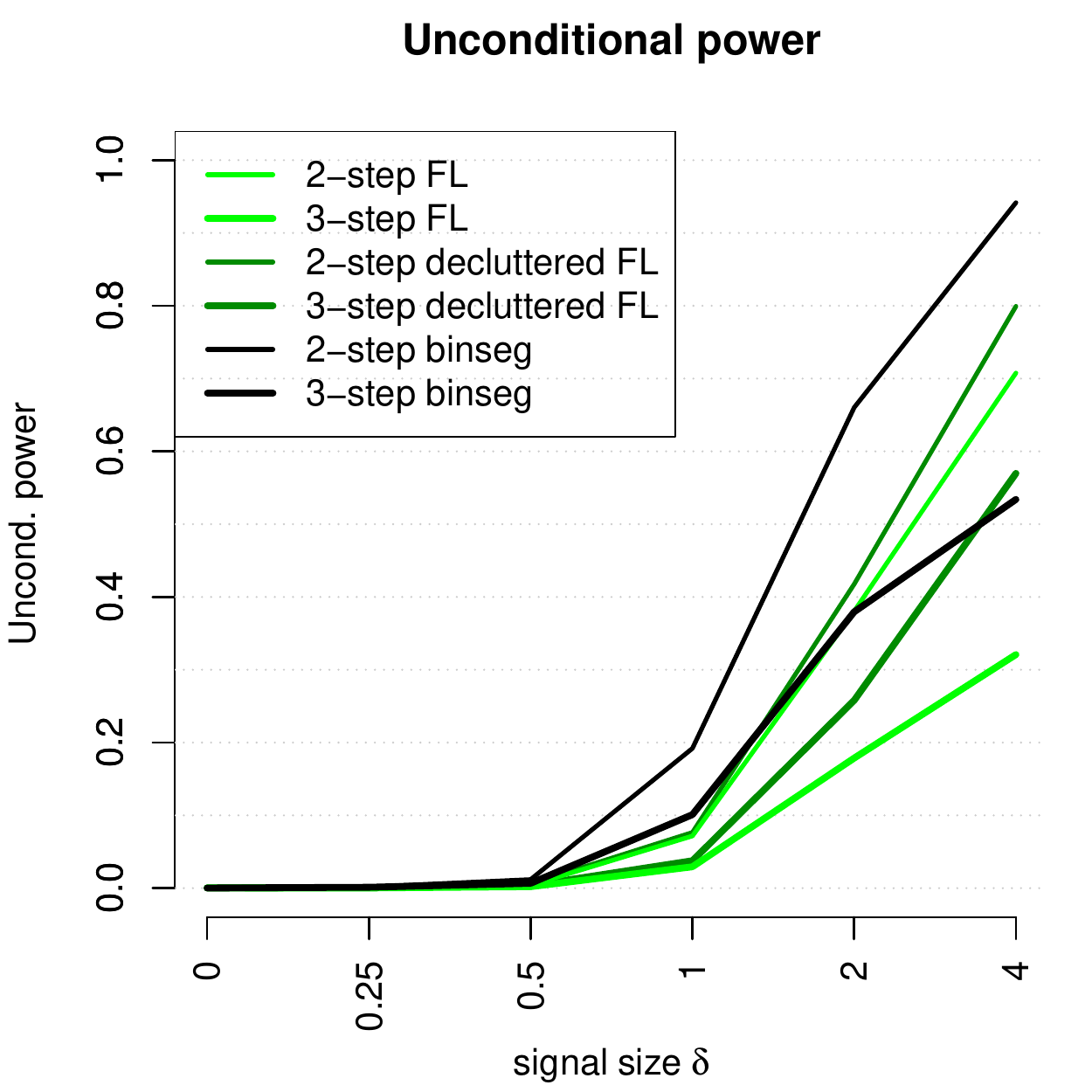}\hspace{-3mm}
    \includegraphics[width=.33\linewidth]{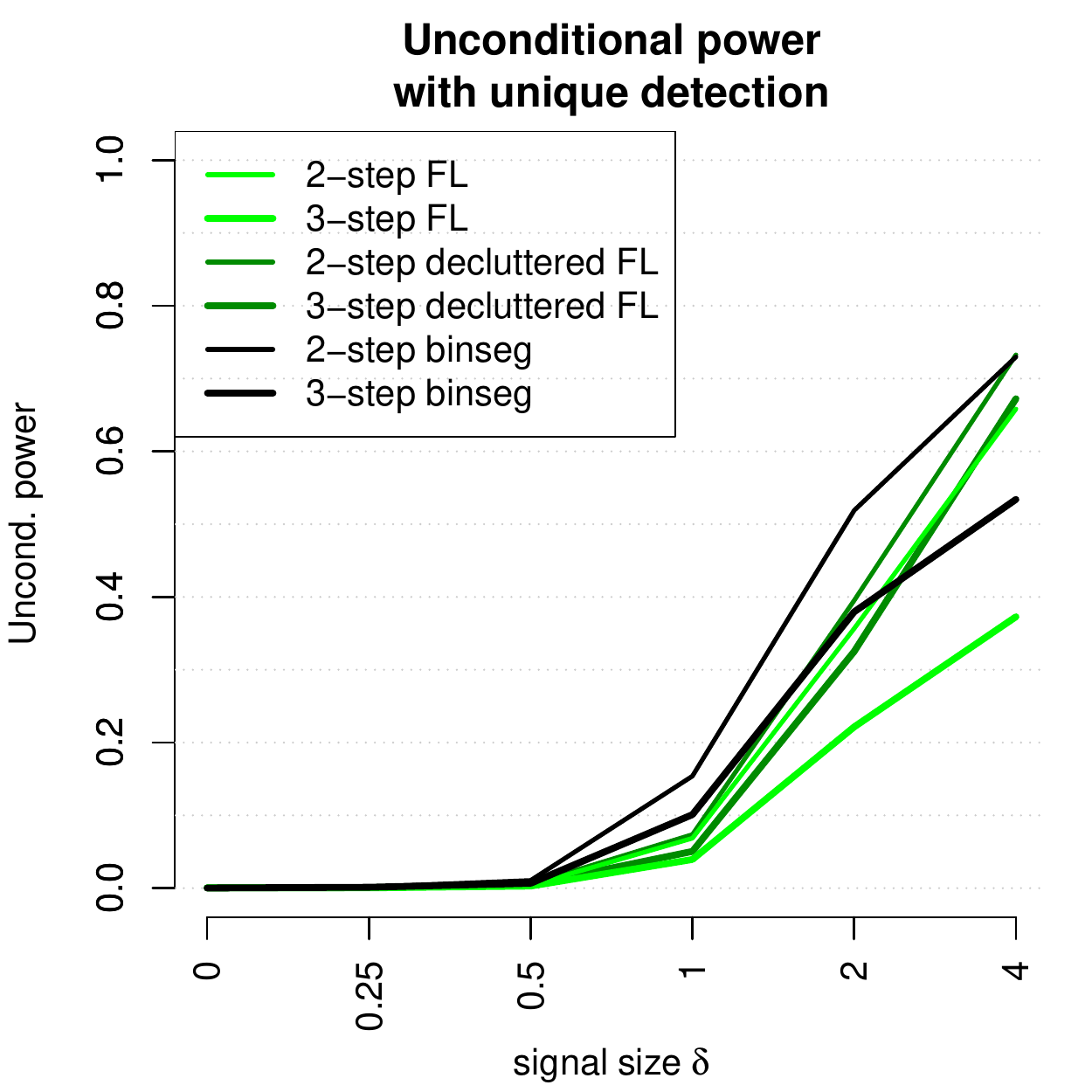}\hspace{-3mm}
    \caption{\it\small 
    (Left): Various detections for FL, either using 2 or 3 steps, and either using decluttering
    or not.       
    (Middle): The unconditional power of various segmentation algorithms.
    (Right): The unconditional power, but defined as the conditional power multiplied by
    the unique detection probability.
    } \label{fig:unique-power-comparison}
\end{figure}

\subsection{Power comparison with different mean shape}
\label{app:edge-mutation}

The synthetic mean discussed here consists of a single upward changepoint piece-wise
constant mean, as shown in \eqref{eq:edge-mutation} and
\Fref{fig:power-comparison-data-edge}. This is chosen to be another realistic
example of the mutation phenomenon as observed in array CGH datasets from
\citet{Snijders2001}, in addition to the case shown in the main text. We focus
on the \textit{duplication} mutation scenario, but the results apply similarly
to deletions. As before, the sample size $n=200$ was chosen to be in the scale
of the data length in a typical array CGH dataset in a single chromosome. An
example of this synthetic dataset can be seen in Figure
\ref{fig:power-comparison-data}.  For saturated model tests, WBS no longer outperforms
binary segmentation in power. This is expected since there is only a single
changepoint not accompanied by opposing-direction changepoints.

\begin{equation}\label{eq:edge-mutation}
\hspace{-20mm}\textbf{Edge mutation:}\hspace{5mm}
  y_i \sim \cN(\theta_i, 1), \;\; 
  \theta_i = \begin{cases}
 \delta & \text{ if } 161\le i \le 200\\
 0 & \text{ if otherwise }\\
  \end{cases}
\end{equation}

\begin{figure}[ht]
  \centering
  \includegraphics[width=.5\linewidth]{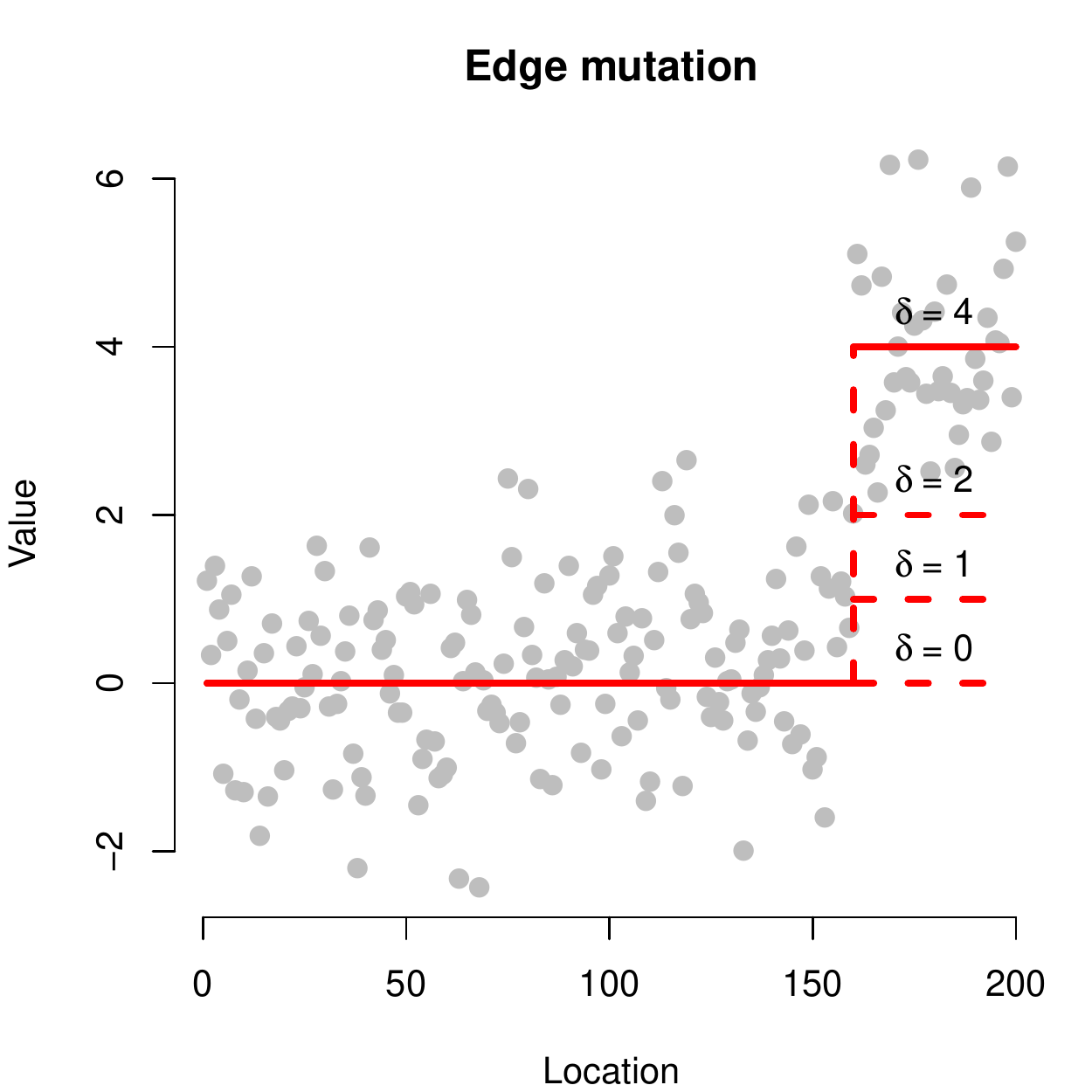}
  \caption{\it\small Analogous to \Fref{fig:power-comparison-data} but representing edge mutations.}
  \label{fig:power-comparison-data-edge}
\end{figure}

\begin{figure}
  \centering
    \includegraphics[width=.33\linewidth]{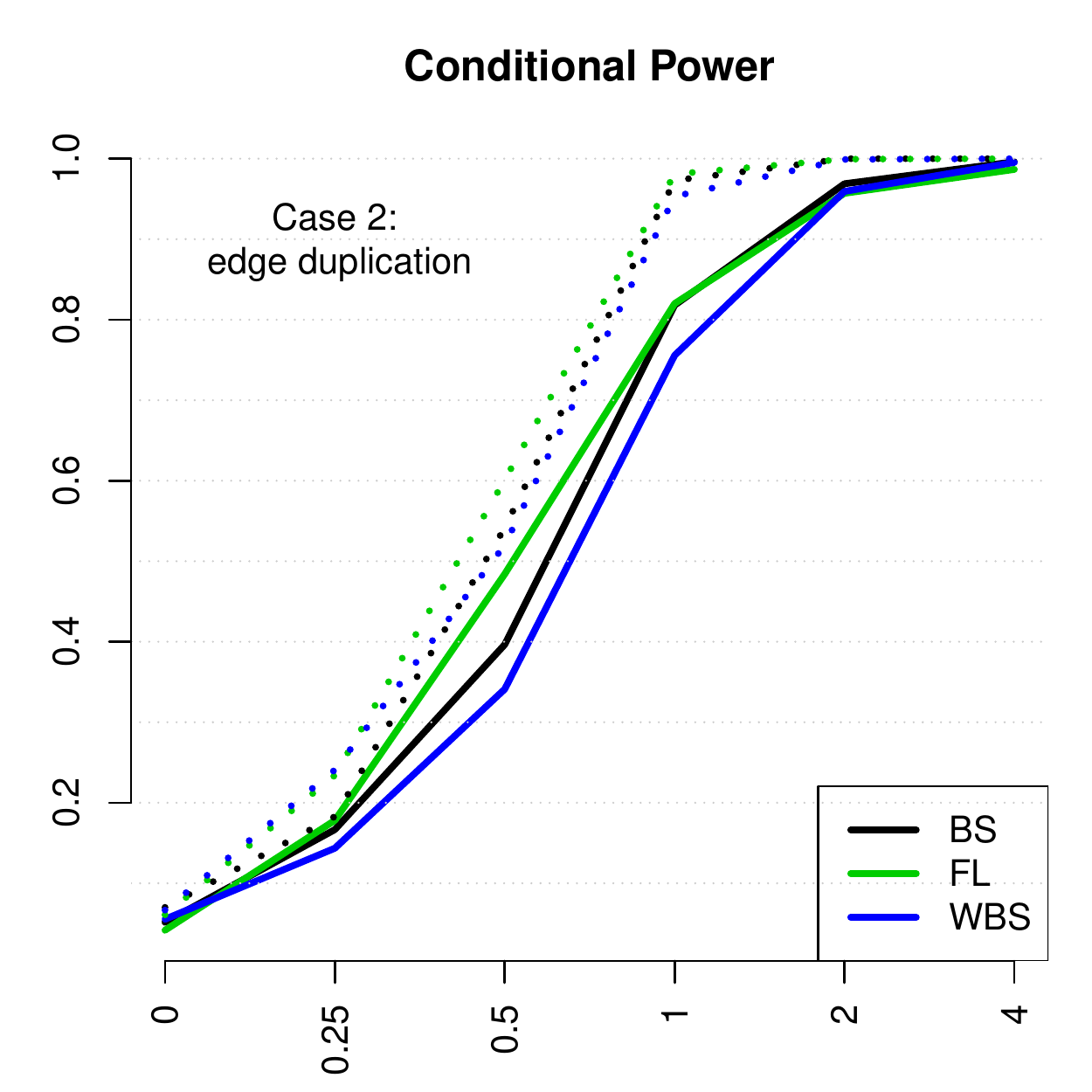}\hspace{-3mm}
    \includegraphics[width=.33\linewidth]{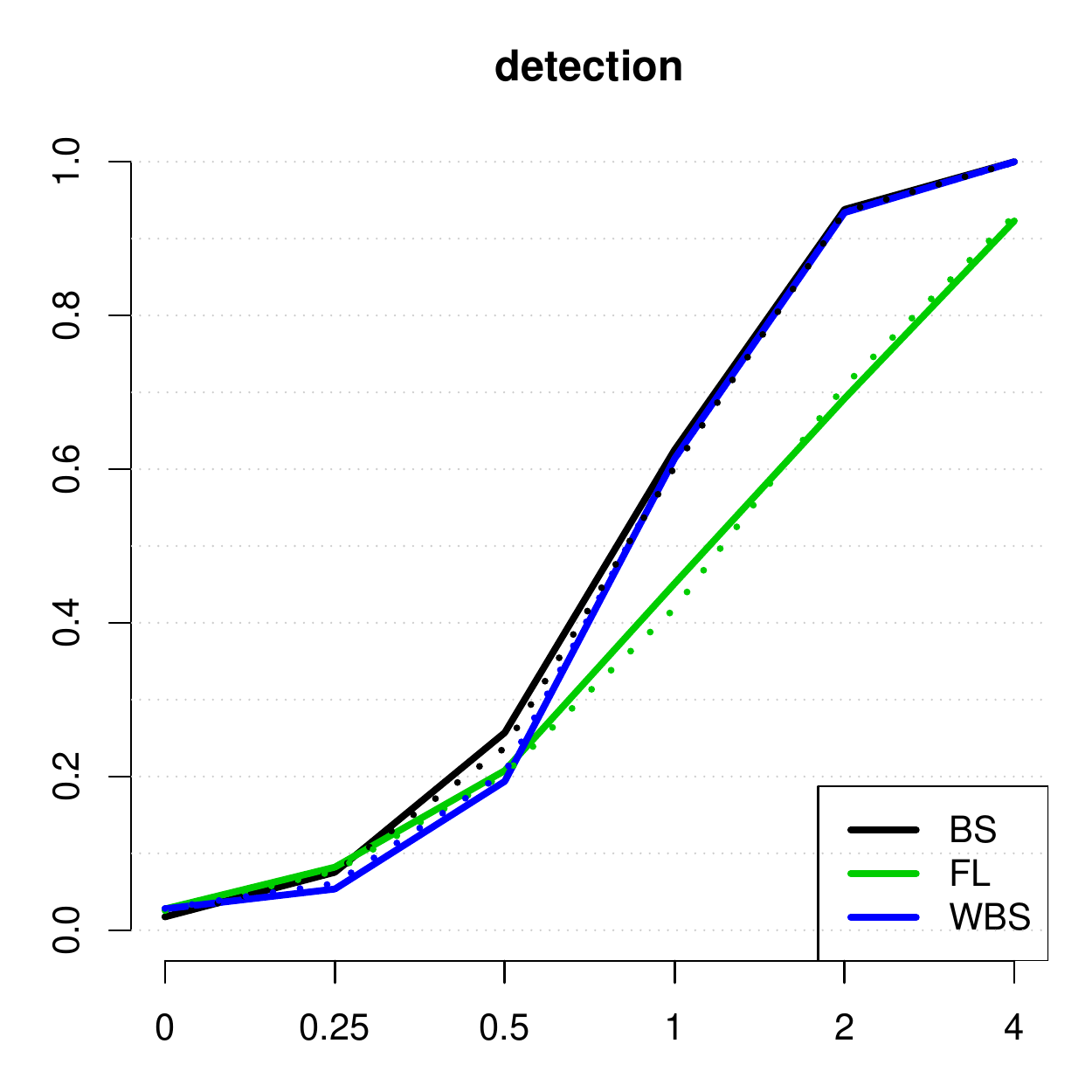}\hspace{-3mm}
    \includegraphics[width=.33\linewidth]{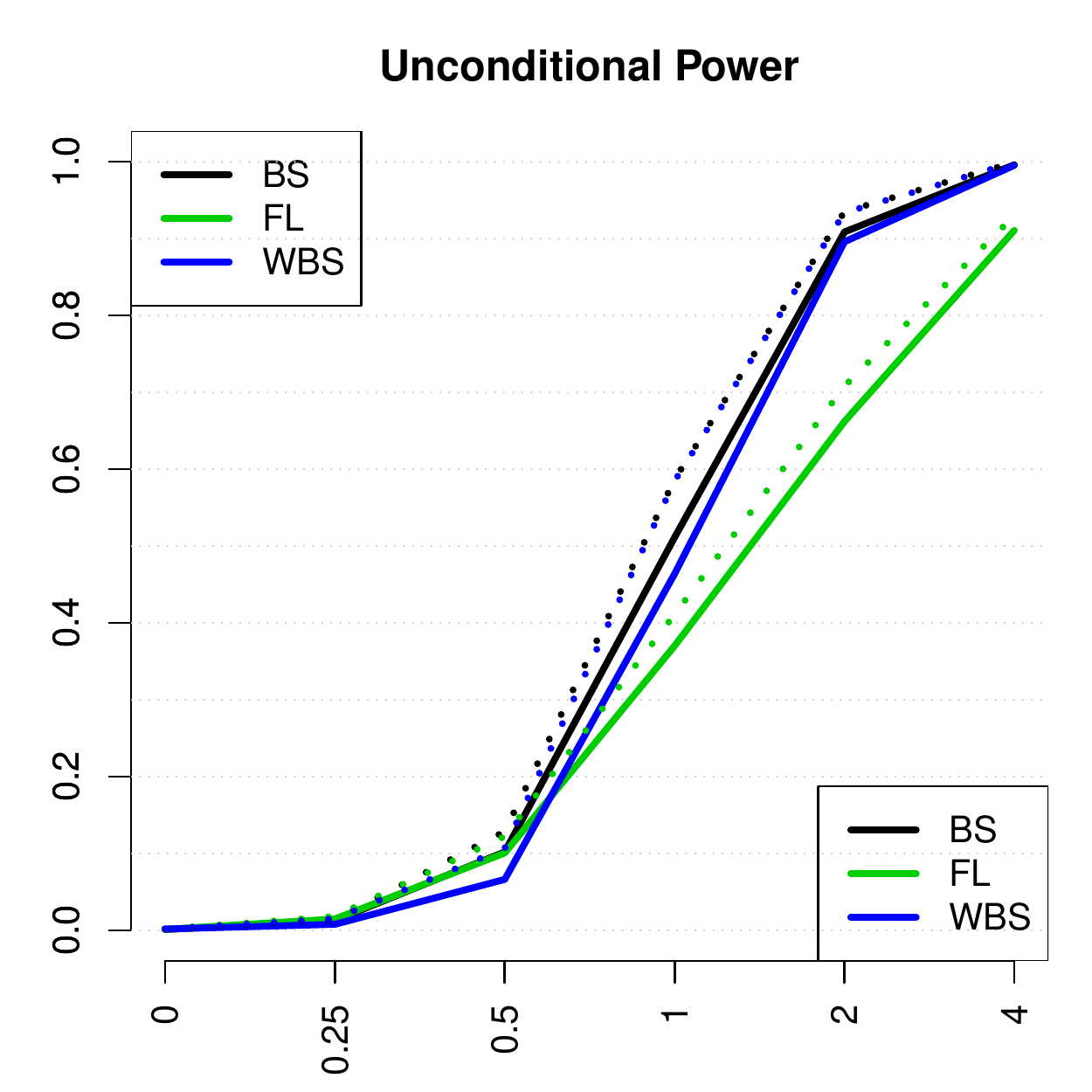}\hspace{-3mm}
    \caption{\it\small 
    Same setup as \Fref{fig:power-comparison} but
    for edge-mutation data. }
\label{fig:power-comparison-edge}
\end{figure}

\subsection{Sample splitting (continued)}
The results in \Fref{fig:samplesplit} were based on approximate detection where,
for methods used on the entire dataset of length $n$,
we defined a detection event as estimating $\pm 2$ of the true changepoint locations.
For sample splitting, this was defined as estimate $\pm 1$ of the true changepoint location based
on half the dataset. This choice of approximate detection is somewhat arbitrary, and
it is informative to see if the results would change if we considered only
exact detection.
We can see from \Fref{fig:samplesplit-exact} that randomized TG p-values have
comparable power with sample splitting inferences, among tests that are
regarding exactly the right changepoints.

\begin{figure}[ht]
  \centering
  \includegraphics[width=.33\linewidth]{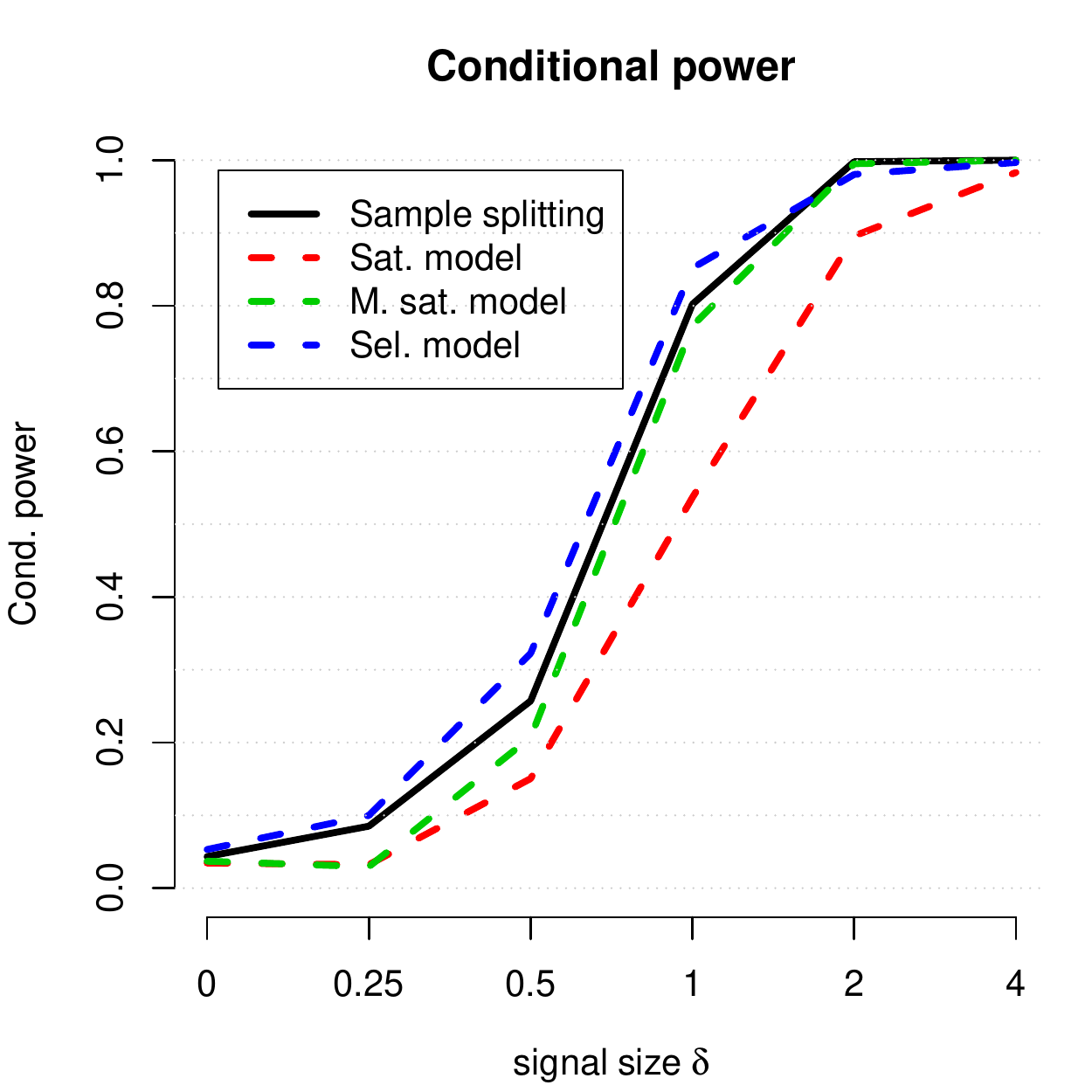}\hspace{-3mm}
  \includegraphics[width=.33\linewidth]{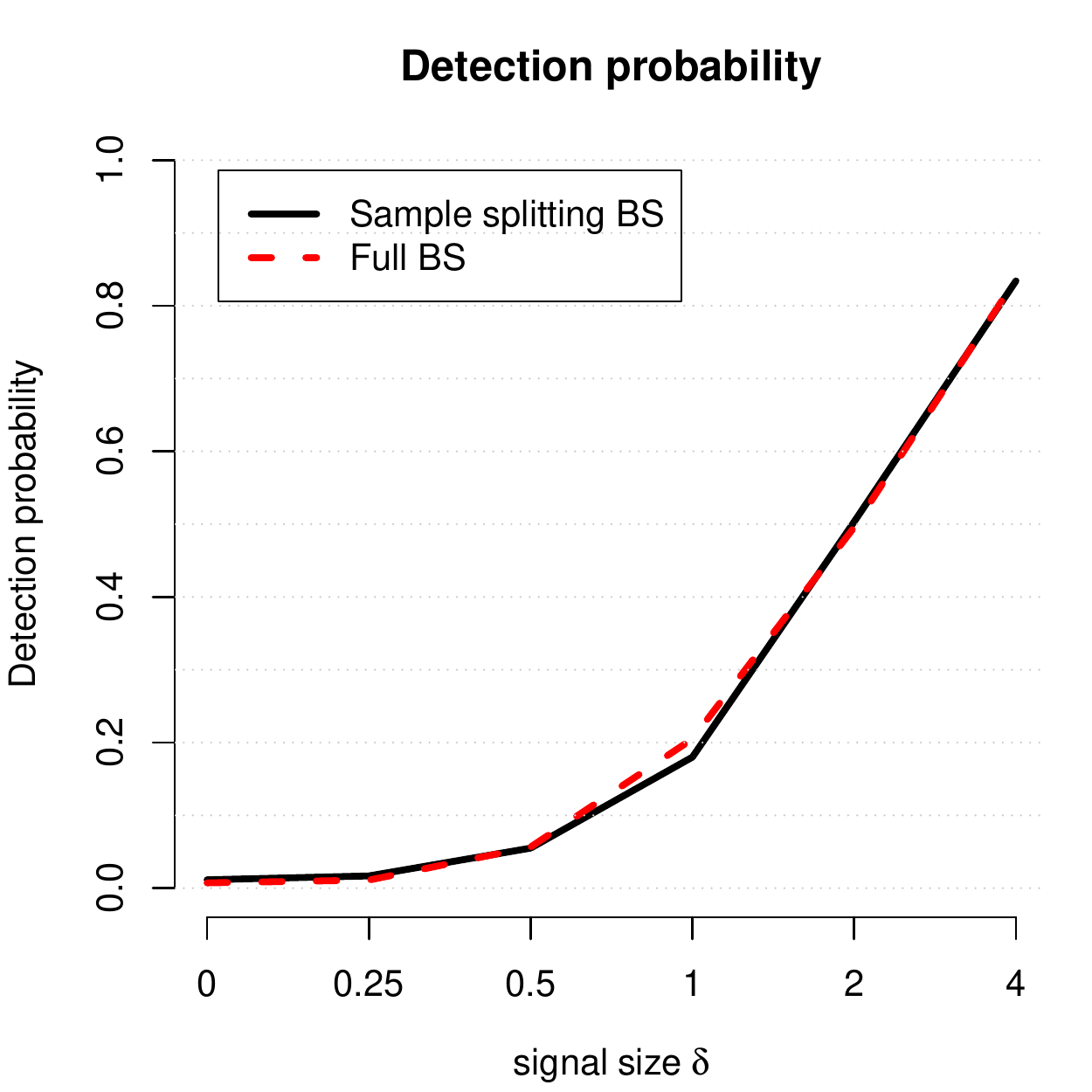}\hspace{-3mm}
  \includegraphics[width=.33\linewidth]{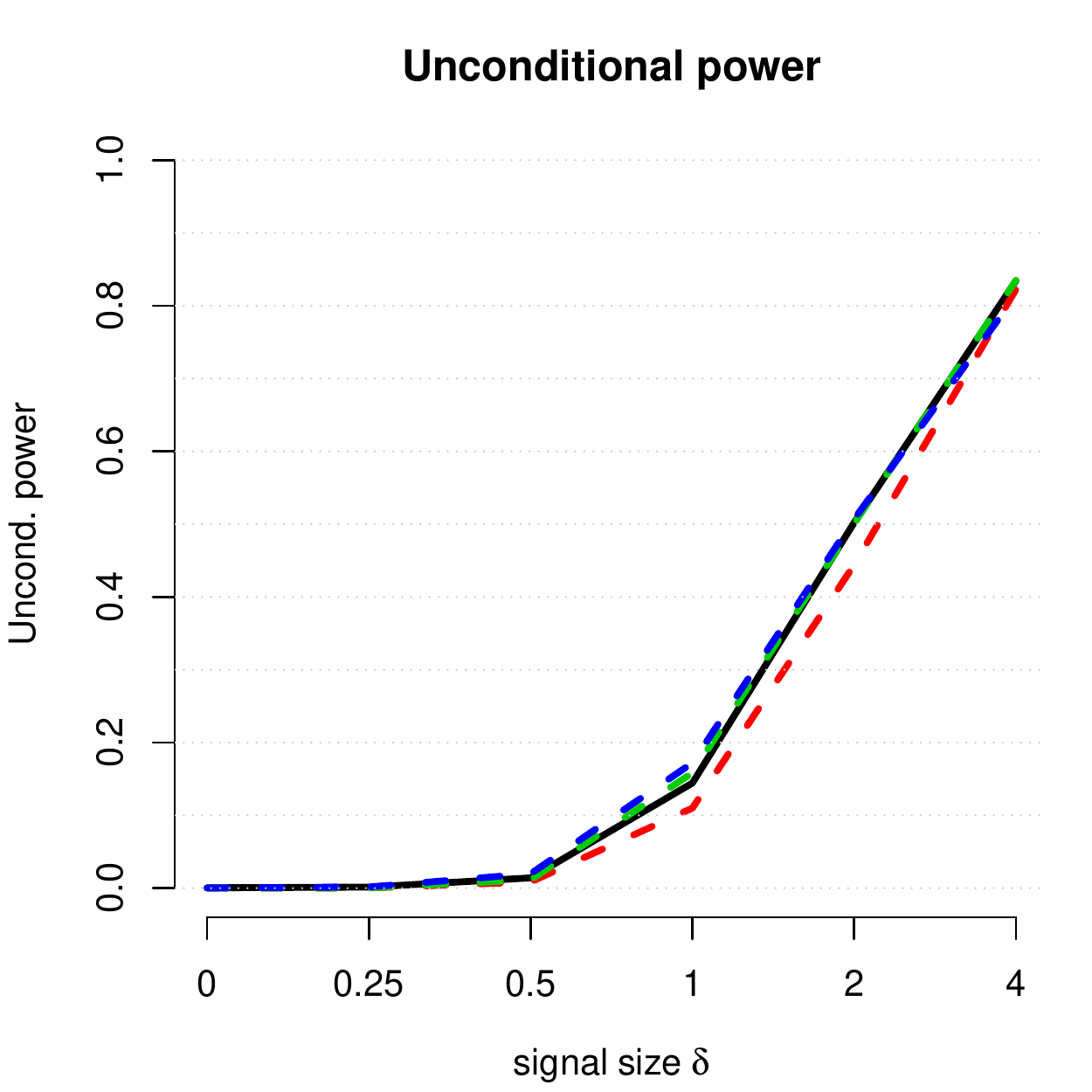}\hspace{-3mm}
  \caption{\it\small The same setup as in \Fref{fig:samplesplit} but with exact detection.}
  \label{fig:samplesplit-exact}
\end{figure}

\section{Model size selection using information criteria}
\label{app:ic}

Throughout the paper we assume that the number of algorithm steps $k$ is fixed.
\citet{hyun2018exact} introduces a stopping rule based on information criteria
(IC) which can be characterized as a polyhedral selection event. The IC for the
sequence of models $M_{1:\ell}, \ell=1,\ldots, n-1$ is
\begin{equation}
  J(M_{1:\ell}) = \|y - \hat y_{M_{1:\ell}(y)}\|^2_2 + p\big(M_{1:\ell}(y)\big).
\end{equation}
We omit the dependency on $y$ when obvious.
We use the BIC complexity penalty $p(M_k) = \sigma^2 \cdot k \cdot \log(n)$ for
this paper.  Also define
$S_\ell(y) = \sign\left(J(M_{1:\ell}) -J(M_{1:(\ell-1)})\right)$ to be the
sign of the difference in IC between step $\ell-1$ and $\ell$. This is a $+1$ for a rise
and $-1$ for a decline. A data-dependent stopping rule $\hat k$ is defined as
\begin{equation}\label{eq:stoprule}
  \hat k (y) = \min\{k : S_k(y) = S_{k+1}(y) = \ldots = S_{k+q}(y) = 1\}
\end{equation}
which is a local minimization of IC, defined as the first time $q$ consecutive
rises occur. As discussed in \cite{hyun2018exact}, $q=2$ is a reasonable choice
for the changepoint detection. To carry out valid selective inference, we
condition on the selection event $\one[ S_{1:(k+q)}(y) = S_{1:(k+q)}(y_\obs)]$,
which is enough to determine $\hat k$. A $k$-step model for $k$ chosen by
\eqref{eq:stoprule} can be understood to be $ M_{1:\hat k}(Y) =
M_{1:k}(y_\obs)$. The corresponding selection event $P_{M_{1:\hat{k}}}$ is with the
additional halfspaces, as outlined in \cite{hyun2018exact}. Simulations in Figure
\ref{fig:ic-power} show that introducing IC stopping is valid, by controlled
type-I error, but comes at the cost of considerable power
loss. 

\begin{figure}

  \includegraphics[width=.31\linewidth]{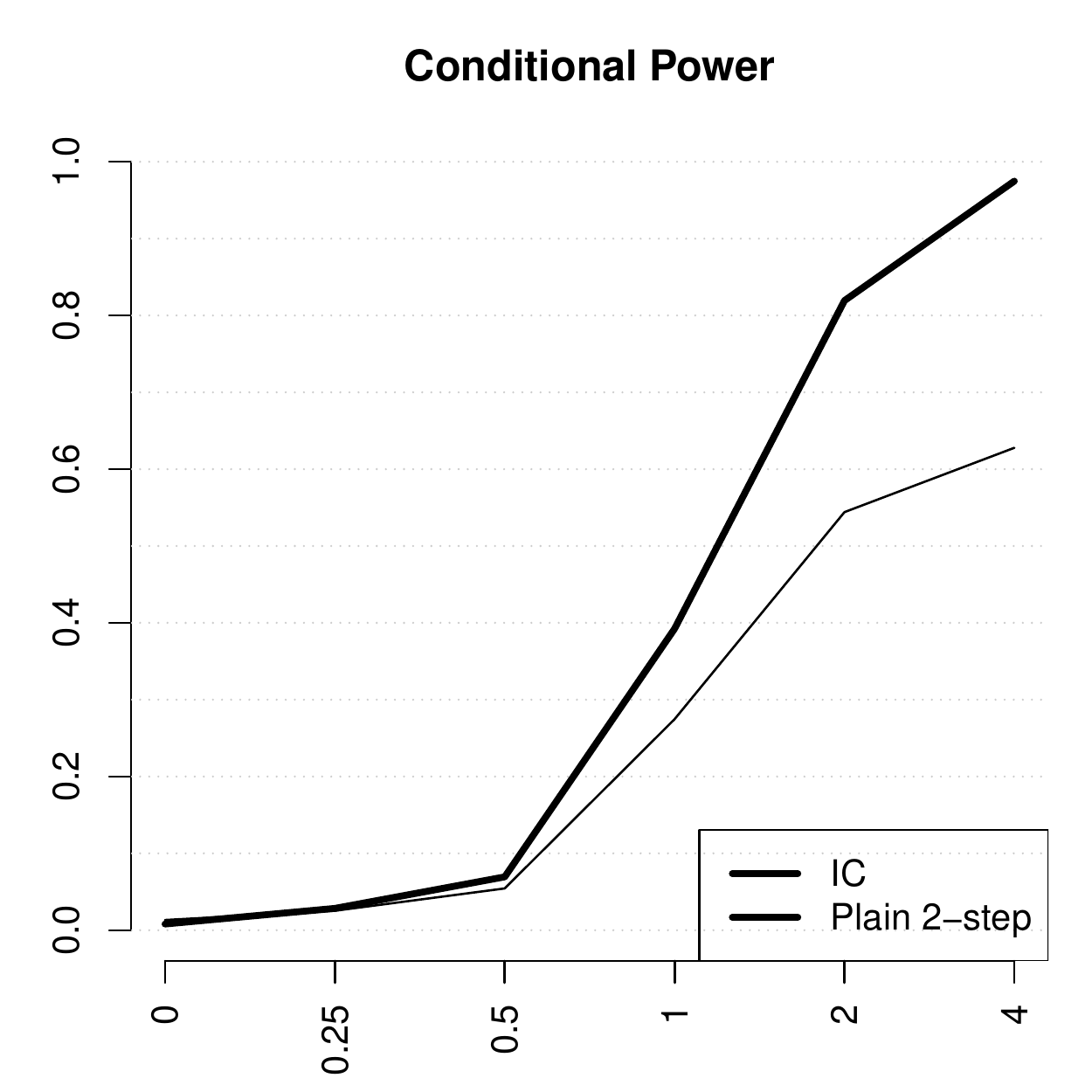}
  \includegraphics[width=.31\linewidth]{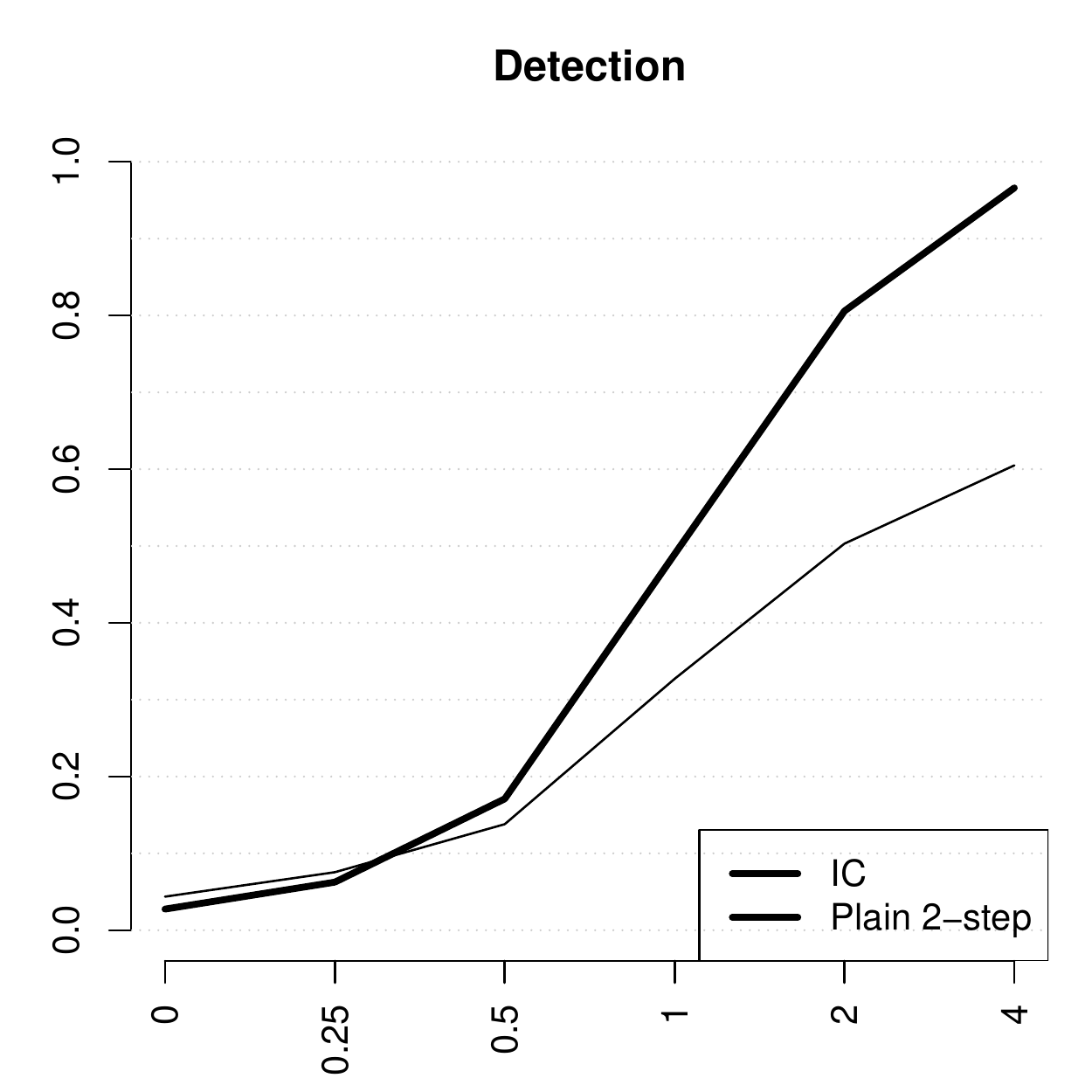}
  \includegraphics[width=.31\linewidth]{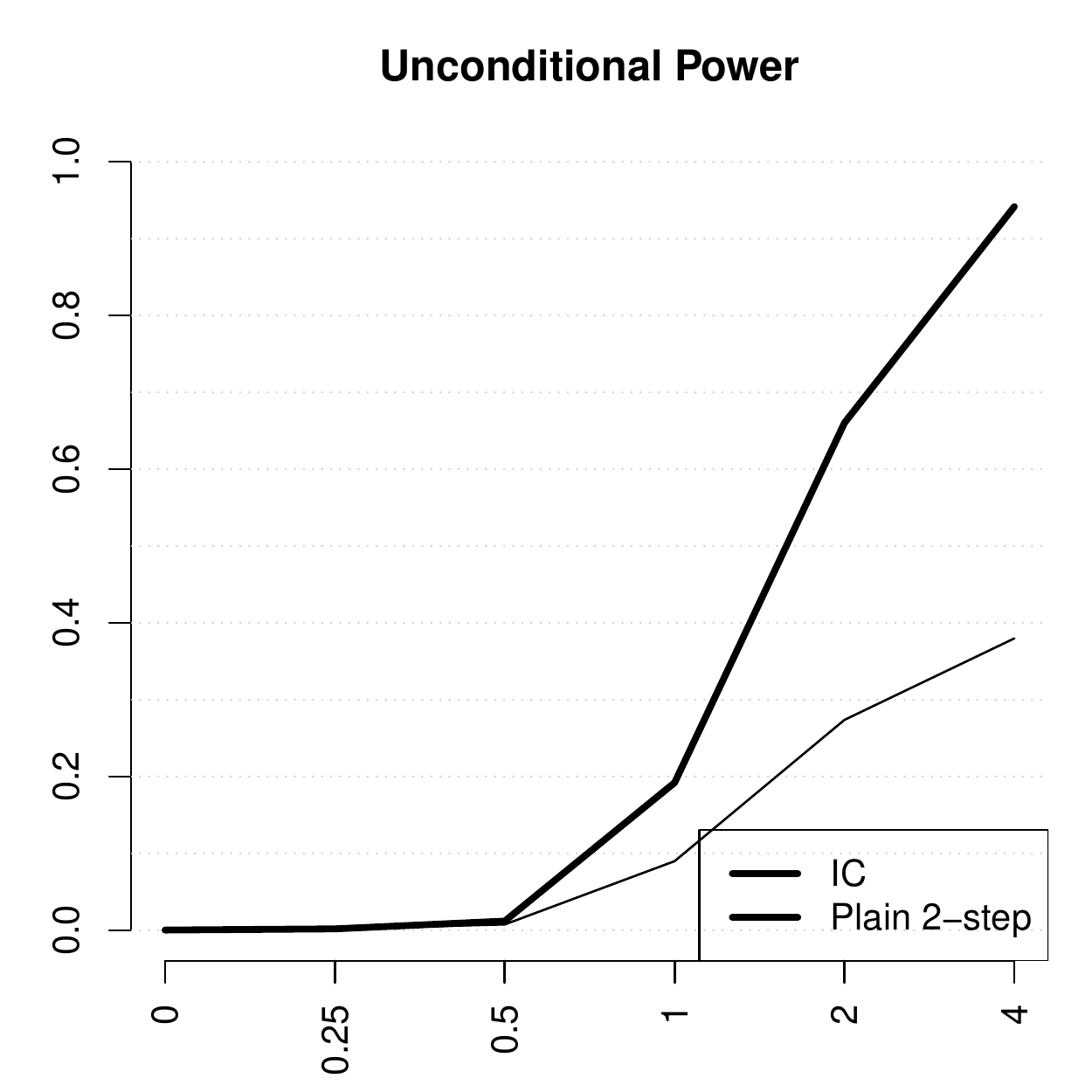}
  \caption{\small\it
  Similar setup as \Fref{fig:power-comparison}.
    In the middle-mutation data example from
    \eqref{eq:middle-mutation}. IC-stopped binary segmentation inference (bold
    line) is compared to a fixed 2-step binary segmentation inferences (thin
    line). We can see that the power and detection are considerably lower. The
    average number of steps taken per each $\delta$ on x-axis ticks are
    $1.34, 1.86, 3.02, 3.64, 3.77, 3.72$, respectively.}
  \label{fig:ic-power}
\end{figure} 
\end{document}